\documentclass{lmcs}

\newcommand{\takeout}[1]{\empty}

\usepackage[utf8]{inputenc}
\usepackage{amssymb}
\usepackage{amsthm}
\usepackage{stmaryrd}
\usepackage{amsmath}
\usepackage{dsfont}
\usepackage{thmtools}
\usepackage{mathtools}
\usepackage{makecell}
\usepackage{tikz-cd}
\tikzset{shiftarr/.style={
        rounded corners,%
        to path={--([#1]\tikztostart.center)
                     -- ([#1]\tikztotarget.center) \tikztonodes
                     -- (\tikztotarget)},
}}
\usepackage[final]{fixme}
\fxuselayouts{inline,nomargin}
\fxusetheme{color}
\FXRegisterAuthor{fb}{afb}{FB}
\FXRegisterAuthor{nw}{anw}{NW}
\FXRegisterAuthor{sm}{asm}{SM}
\FXRegisterAuthor{hu}{ahu}{HU}

\usepackage[inline]{enumitem}
\setlist[enumerate,1]{label=(\arabic*),font=\normalfont,align=left,leftmargin=0pt,labelindent=0pt,listparindent=\parindent,labelwidth=0pt,itemindent=!,topsep=3pt,parsep=0pt,itemsep=3pt,start=1}
\setlist[enumerate,2]{label=(\alph*),font=\normalfont,labelindent=*,leftmargin=*,start=1}

\usepackage[notref,notcite,final]{showkeys}
\usepackage{seqsplit}
\usepackage{xstring}
\newcommand{\defaultshowkeysformat}[1]{%
\StrSubstitute{#1}{ }{\textvisiblespace}[\TEMP]%
\parbox[t]{2.5cm}{\raggedright\normalfont\small\ttfamily\(\{\){\color{red!50!black}\expandafter\seqsplit\expandafter{\TEMP}}\(\}\)}%
}

\renewcommand*\showkeyslabelformat[1]{%
\noexpandarg%
\defaultshowkeysformat{#1}%
}

\PassOptionsToPackage{hyphens}{url}
\makeatletter
\usepackage{hyperref}
\hypersetup{
  final,
  hidelinks,
  pdftitle={\@title},
  pdfauthor={Fabian Lenke, Stefan Milius, Henning Urbat},
}
\makeatother

\usepackage{xspace}
\usepackage{bussproofs}
\EnableBpAbbreviations 

\newcommand{\ifarx}[2]{}

\usepackage{cite} 

\def\resettheorembrackets{
\def\theorembracketopen{(}
\def\theorembracketclose{)}
}
\resettheorembrackets
\makeatletter
\def\@spopargbegintheorem#1#2#3#4#5{\trivlist
      \item[\hskip\labelsep{#4#1\ #2}]{#4{\theorembracketopen}#3{\theorembracketclose}\@thmcounterend\ }#5}
\makeatother

\newcommand{\resetCurThmBraces}{%
  \gdef\curThmBraceOpen{(}%
  \gdef\curThmBraceClose{)}}
\resetCurThmBraces
\newcommand{\removeThmBraces}{%
  \gdef\curThmBraceOpen{}%
  \gdef\curThmBraceClose{}}
\resetCurThmBraces

\newenvironment{notheorembrackets}{\removeThmBraces}{\resetCurThmBraces}

\usepackage{etoolbox}
\patchcmd{\thmhead}{(#3)}{\curThmBraceOpen #3\curThmBraceClose }{}{}

\renewcommand{\P}{\mathcal{P}}
\newcommand{\Pfin}{{\P_{\mathsf{f}}}}

\newcommand{\power}{\mathbin{\pitchfork}}
\newcommand{\slice}{\mathbin{\downarrow}}
\DeclareMathOperator{\uncurry}{\mathsf{uncurry}}
\DeclareMathOperator{\flip}{\mathsf{flip}}

\DeclareMathOperator{\Ran}{\mathsf{Ran}}
\DeclareMathOperator{\Lan}{\mathsf{Lan}}

\DeclareMathOperator{\Spec}{\mathsf{Spec}}
\DeclareMathOperator{\Cnt}{\mathsf{C}}
\DeclareMathOperator{\Id}{\mathsf{Id}}
\DeclareMathOperator{\id}{\mathsf{id}}

\DeclareMathOperator{\Cl}{Cl}
\DeclareMathOperator*{\colim}{colim}

\newcommand{\Kl}{\mathbf{Kl}}
\newcommand{\EM}{\mathbf{EM}}
\newcommand{\y}{\mathbf{y}}
\newcommand{\Cody}[1]{\mathsf{Cody}({#1})}
\newcommand{\V}{\mathbb{V}}
\newcommand{\T}{\mathcal{T}}
\newcommand{\Sh}{\mathsf{Sh}}
\newcommand{\At}{\mathsf{At}}
\newcommand{\Cstar}{\mathsf{C}^*\mathbf{Alg}}
\newcommand{\Stalk}{\mathsf{Stalk}}
\newcommand{\K}{\mathbb{K}}
\newcommand{\Viet}{\mathbb{V}}
\newcommand{\Dist}{\mathcal{D}}

\newcommand{\Exp}{\mathcal{E}}
\newcommand{\Giry}{\mathcal{G}}
\newcommand{\Act}{\mathcal{M}}
\newcommand{\Rad}{\mathcal{R}}
\newcommand{\N}{\mathbb{N}}
\newcommand{\Z}{\mathbb{Z}}

\newcommand{\Sierp}{\mathbb{S}}
\newcommand{\Bo}{\mathcal{B}}
\newcommand{\SMeas}{\mathcal{M}_{S}}
\newcommand{\Ffrm}{\mathcal{L}}

\newcommand{\Sob}{\mathcal{S}}

\newcommand{\ol}{\overline}

\newcommand{\Cpx}{\mathbb{C}}

\newcommand{\U}{\mathcal{U}}
\newcommand{\UP}{\mathcal{U}\mathcal{P}}
\newcommand{\I}{\mathcal{I}}
\newcommand{\F}{\mathcal{F}}

\newcommand{\subto}{\hookrightarrow}
\newcommand{\cat}[1]{\ensuremath{\mathbf{#1}}\xspace}
\newcommand{\Poc}{\cat{Poc}}
\newcommand{\Med}{\cat{Med}}
\newcommand{\Fam}{\cat{Fam}}
\newcommand{\A}{\cat{A}}
\newcommand{\B}{\cat{B}}
\newcommand{\C}{\cat{C}}
\newcommand{\D}{\cat{D}}

\newcommand{\EA}{\cat{EA}}
\newcommand{\sEA}{\cat{EA_{\sigma}}}
\newcommand{\sEMod}{\cat{EMod_{\sigma}}}
\newcommand{\EMod}{\cat{EMod}}

\newcommand{\CMon}{\cat{CMon}}
\newcommand{\Mat}{\mathbf{Mat}}
\newcommand{\tp}{\intercal}
\newcommand{\Mod}{\mathbf{Mod}}
\newcommand{\Stone}{\cat{Stone}}
\newcommand{\Set}{\cat{Set}}
\newcommand{\Ab}{\cat{Ab}}
\newcommand{\CHAb}{\cat{CHAb}}
\newcommand{\Setf}{\cat{Set}_\mathsf{f}}
\newcommand{\BSh}{\cat{BSh}}
\newcommand{\Top}{\cat{Top}}
\newcommand{\Frm}{\cat{Frm}}

\newcommand{\CHaus}{\cat{CHaus}}
\newcommand{\CMet}{\cat{CMet}}
\newcommand{\Setc}{\cat{Set}_\mathsf{c}}
\newcommand{\Meas}{\cat{Meas}}

\newcommand{\JSL}{\cat{JSL}}
\newcommand{\CJsl}{\cat{CJSL}}

\newcommand{\MSL}{\cat{MSL}}
\newcommand{\Caba}{\cat{CABA}}
\newcommand{\Cabac}{\cat{CABA}_\mathsf{c}}
\newcommand{\Vect}{\cat{Vect}}
\newcommand{\BA}{\cat{BA}}
\newcommand{\DL}{\cat{DL}}
\newcommand{\Pos}{\cat{Pos}}

\newcommand{\f}{\mathsf{f}}
\newcommand{\cnt}{\mathsf{c}}
\newcommand{\fd}{\mathsf{fd}}
\newcommand{\fp}{\mathsf{fp}}
\newcommand{\op}{\mathsf{op}}

\newcommand{\epi}{\twoheadrightarrow}

\newcommand{\incl}{\hookrightarrow}
\newcommand*\diff{\mathop{}\!\mathrm{d}}
\newcommand{\xto}{\xra}

\newcommand{\xra}[1]{\xrightarrow{~#1~}}
\newcommand{\Sier}{\mathbf{Sier}}
\newcommand{\inl}{\mathsf{inl}}
\newcommand{\inr}{\mathsf{inr}}

\numberwithin{equation}{section}

\usepackage{cleveref}
\newcommand{\iref}[2]{\Cref{#1}\ref{#1:#2}}
\crefname{cor}{Corollary}{Corollary}
\crefname{exa}{Example}{Examples}
\crefname{defi}{Definition}{Definitions}
\crefname{rem}{Remark}{Remark}
\crefname{lem}{Lemma}{Lemma}
\crefname{thm}{Theorem}{Theorem}
\crefname{thmC}{Theorem}{Theorem}
\crefname{prop}{Proposition}{Proposition}

\begin{document}


\title{Codensity Monads via Density and Duality}

\author[F.~Lenke]{Fabian Lenke\lmcsorcid{0000-0001-5890-9485}}
\thanks{Fabian Lenke is supported by Deutsche Forschungsgemeinschaft (DFG, German Research Foundation) -- project number 470467389.}

\author[H.~Urbat]{Henning Urbat\lmcsorcid{0000-0002-3265-7168}}
\thanks{Henning Urbat is supported by Deutsche Forschungsgemeinschaft (DFG, German Research Foundation) -- project numbers 470467389 and 569130867.}

\author[S.~Milius]{Stefan Milius\lmcsorcid{0000-0002-2021-1644}}
\thanks{Stefan Milius is supported by Deutsche Forschungsgemeinschaft (DFG, German Research Foundation) -- project number 517924115.}

\author[N.~Wittrock]{Nico Wittrock\lmcsorcid{0009-0004-8740-117X}}
\thanks{  Nico Wittrock is supported by Fundação para a Ciência e Tecnologia (FCT, Portuguese Foundation for Science and Technology) -- grant number 2023.02568.BD -- and partly by the FCT KaleidosQope project -- project number 2023.13603.PEX.}

\address{Friedrich-Alexander-Universität Erlangen-Nürnberg, Germany}
\email{\{fabian.lenke,stefan.milius,henning.urbat\}@fau.de, }%

\address{International Iberian Nanotechnology Laboratory (INL), Portugal\and High-Assurance Software Laboratory (HASLab), Portugal\and Department of Computer Science, University of Minho, Portugal}
\email{nico.wittrock@inl.int}

\begin{abstract}
Codensity monads provide a universal method to generate complex monads from simple functors. Recently, a wide range of important monads in logic, denotational semantics, and probabilistic computation, such as several incarnations of the ultrafilter monad, the Vietoris monad, and the Giry monad, have been presented as codensity monads, using complex arguments. We propose a unifying categorical approach to codensity presentations of monads, based on the idea of relating the presenting functor to a \emph{dense} functor via a suitable \emph{duality} between categories. We prove a general presentation result applying to every such situation and demonstrate that most codensity presentations known in the literature emerge from this strikingly simple duality-based setup, drastically alleviating~the complexity of their proofs and in many cases completely reducing them to standard duality results. Additionally, we derive a number of new codensity presentations using our framework, including the filter, lower Vietoris and Stone-\v{C}ech compactification monads on topological spaces, and the canonical extension and expectation monads on sets.
\end{abstract}
\maketitle
\section{Introduction}
\label{sec:introduction}

Monads are among the most fundamental concepts of category theory. They provide a common
abstraction of algebraic theories~\cite{Manes76} and notions of computation~\cite{moggi91},
and the tight interplay between both viewpoints has inspired decades of fruitful research in
theoretical computer science. While monads and their underlying structure can be defined
from scratch, it is often simpler and more informative to \emph{generate} monads via a
universal construction. A powerful method to do so was discovered by Kock~\cite{kock66}
in the 1960s: {every} functor $F\colon \C_0\to \C$ canonically induces a monad $\T=\Cody{F}$ on
the category $\C$, the \emph{codensity monad} of~$F$, provided that certain limits in $\C$
exist. It is constructed via the right Kan extension of $F$ along itself and generalizes the
familiar construction of monads from adjunctions. The prototypical example is the
{ultrafilter monad} on the category of sets (whose algebras are compact Hausdorff
spaces~\cite{Manes76}), which has been characterized as the codensity monad of the inclusion
$\Set_\f\subto \Set$ of the category of finite sets into $\Set$~\cite{kg71}. Generally, while
every monad can be presented as the codensity monad of \emph{some} functor (e.g.~the
forgetful functor of its Eilenberg-Moore category), the goal is to come up with the
\emph{simplest} possible presentation.

In recent years, codensity monads have found a growing number of applications in computer
science; in particular, they have been identified as an elegant categorical underpinning of
profinite methods in automata theory~\cite{camu16,uacm17,gpr20,acmu21}, program
optimizations in functional languages~\cite{hinze12}, and the construction of liftings of
monads along fibrations relevant in type theory~\cite{ksu18}. Furthermore, a number of key
monads appearing in logic, denotational semantics, and probabilistic computation have been
presented as codensity monads, most notably generalizations of the ultrafilter monad to
algebraic and topological categories~\cite{l13,as21}, the Vietoris hyperspace monad on Stone
spaces~\cite{gpr20}, and probability monads such as the all-important Giry monad on
measurable spaces and many of its variants~\cite{a16,vb22,r20,s24}. The importance of such
codensity presentations is that they relate the given (fairly complex) monads to
structurally much simpler generating functors and endow those monads with a universal
property. This provides new insights into the structure of the monads themselves,
facilitating their use in theory and practice. For example, desirable properties of
probability monads such as commutativity or affinity, which are instrumental in the
synthetic approach to probability theory based on Markov categories~\cite{cj2019,fritz20} or effectuses~\cite{j18e},
and somewhat tedious to prove directly, can be derived in a principled manner from their
codensity presentation~\cite{s24}.

All the above codensity presentations are non-trivial results. Their proofs in the literature rely on a careful analysis of the structure of the respective monads and on domain-specific knowledge, such as advanced results from measure theory for probability monads. Overall, this leads to technically challenging and largely ad hoc proofs of codensity presentations.

In the present paper, we address this issue by developing a \emph{simple}, \emph{general} and
\emph{uniform} method to synthesize codensity presentations of monads, putting a common
umbrella over most known instances, and many more. The core insight underlying our
contribution is that
\enlargethispage*{10pt}
\[
  \textbf{Codensity Monads} \;=\; \textbf{Density} \;+\; \textbf{Duality}.
\]
More specifically, our approach rests on a simple but effective
observation: all `interesting' codensity presentations (where $\T$ is the codensity
monad of some functor $F$)
follow the pattern shown in the diagram below; that is, the functor $F$ decomposes as the pseudoinverse of a \emph{dual equivalence} $E$, the opposite of a
\emph{dense} functor $G$, and a \emph{contravariant right adjoint}~$R$
generating the monad $\T$; in symbols: $F\cong RG^\op E^{-1}$.
\begin{equation}\label{eq:codensity-sit-intro}
  \begin{tikzcd}[column sep=30]
    \C_{0}
    \rar{F} &
    \C
    \ar[out=30,in=-30,loop,looseness=5, "\T"]
    \ar[yshift=0pt, bend right=20,]{d}[swap]{L}
    \ar[yshift=0pt, bend left=20, leftarrow, swap]{d}[swap]{R}
    \ar[phantom]{d}{\dashv}
    \\
    \D_{0}^\op \ar{r}{G^\op}[swap]{\text{($G$ dense)}}
    \ar{u}{E}[swap]{\rotatebox{-90}{$\simeq$}}
    &
    \D^\op
  \end{tikzcd}
\end{equation}
For instance, the presentation of the ultrafilter monad as the
codensity monad of the inclusion functor $I\colon \Set_\f\subto\Set$ is captured by the
diagram below:
\begin{equation}\label{eq:codensity-sit-ultrafilter-intro}
  \begin{tikzcd}[column sep=20, row sep=20]
    \Set_\f
    \ar[hook]{r}{I} &
    \Set
    \ar[out=25,in=-30,loop,looseness=4, "\U",pos=0.48]
    \ar[yshift=0pt, bend right=20,]{d}[swap]{\Set(-,2)}
    \ar[yshift=0pt, bend left=20, leftarrow, swap]{d}[swap]{\BA(-,2)}
    \ar[phantom]{d}{\dashv}
    \\
    \BA_\f^\op \ar[hook]{r}{J^\op}
    \ar{u}{\rotatebox{90}{$\simeq$}}
    &
    \BA^\op
  \end{tikzcd}
\end{equation}
Here we use the familiar duality $\BA_\f^\op\simeq \Set_\f$ between finite Boolean algebras
and finite sets, and density of the inclusion functor $J\colon \BA_\f\subto \BA$ corresponds to the
standard fact that every Boolean algebra is a canonical colimit of finite
Boolean algebras.
Such a decomposition, which we call a \emph{codensity setting}, is all that is needed to get
a codensity presentation of the monad $\T$. Our main result (\Cref{thm:codensity}) states
that:

\smallskip
\centerline{
  \emph{In every codensity setting \eqref{eq:codensity-sit-intro}, the monad $\T$ is the codensity monad of the functor $F$.}
}

\smallskip
\noindent While not technically difficult, this theorem is the conceptual gist behind most
known codensity presentations of monads.  In fact, it can be applied in two different ways: to
characterize the codensity monad $\T$ of a given functor $F$, and conversely, to discover a
simple functor~$F$ presenting a given monad $\T$ as a codensity monad.  In both cases, it
simplifies proofs of codensity presentations by reducing a question regarding
\emph{codensity} to one about \emph{density}.  Formally these are of course dual concepts,
but in practice density is well-understood, especially in algebraic contexts, and typically
follows from general results, while codensity rarely occurs when working with `everyday'
categories of algebras or spaces.

We demonstrate the strength and scope of our main theorem by deriving codensity presentations for a wide variety of different monads.
Our applications not only cover most codensity presentations known in the literature, including all of the monads mentioned above~\cite{l13,as21,gpr20,vb22,r20,kg71,a16,s24}, but also several new examples, notably the first non-trivial codensity presentations for the filter monad on topological spaces, for the (lower) Vietoris hyperspace monad and the Stone-\v{C}ech compactification monad on topological spaces, and for the expectation monad on sets.
In all these cases, establishing the codensity presentation boils down to simply instantiating our theorem to a suitable codensity setting~\eqref{eq:codensity-sit-intro}. Much like in the setting \eqref{eq:codensity-sit-ultrafilter-intro} for the ultrafilter monad, the choice of the duality and of the dense functor is usually fairly obvious and suggested by standard results on the corresponding algebraic categories, and the chosen dualities are very basic ones between categories of finite algebras or categories of relations.

For all monads with a known codensity presentation covered in our paper, the principled and modular nature of our approach leads to proofs that are dramatically shorter and more transparent compared to those found in the original literature.
In fact, a lot of the complexity of the latter can be attributed to the fact that the authors (implicitly) rediscover the arguments underlying the duality and density results appearing in the respective instantiation of \eqref{eq:codensity-sit-intro}. By applying our categorical framework, we can instead appeal to well-known properties of the respective categories and get this work entirely for free. In this way, all our codensity presentations of (ultra)filter and Vietoris-type monads become an essentially straightforward instance of our main theorem. For probability monads like the Giry monad, the (only) non-straightforward part lies in identifying a suitable dual adjunction inducing the given monad, which requires representation theorems relating linear functionals to probability measures. This identifies precisely the part of the reasoning where measure theory is needed.

\paragraph*{Outline. } We recall some categorical preliminaries in \Cref{sec:preliminaries,sec:codensity-monads}.
In \Cref{sec:main-theorem} we state the abstract main theorem.
In Sections \ref{sec:simpl-codens-monads} to \ref{sec:probability-monads}
we spell out different classes of examples in detail. We conclude with a concrete application of the codensity presentation of the ultrafilter monad in \Cref{sec:appl}.

This paper is a revised and extended version of our STACS'26 conference
paper~\cite{LenkeEA26}. In addition to including detailed proofs, we have added some
new material:
\begin{itemize}
\item some remarks on pushforward monads and their relation to the presented codensity monads;
\item a generalized and simplified treatment of double-dualization monads (\Cref{sec:dd-monads});
\item several new codensity presentations of well-known monads such as the canonical extension monad on Boolean algebras (\Cref{sec:ultrafilter}) and the Stone-\v{C}ech compactification monad (\Cref{sec:stone-cech});
\item a treatment of additional codensity presentations known from the literature (median algebras in \Cref{sec:meas-medi-algebr} and ultraproducts in \Cref{sec:ultraproducts}) within our framework;
\item a new proof of Börger's characterization~\cite{b87} of the ultrafilter functor as an application of how to use codensity presentations (\Cref{sec:appl}).
\end{itemize}

%
%
%
%

\section{Categorical Background}
\label{sec:preliminaries}

We assume some familiarity with basic category theory~\cite{m98}, such as equivalence functors, (co)limits, adjunctions, and monads. In the following we introduce the notation used in the paper, and review the key concepts of (co)dense functors and codensity monads.

\paragraph{Monads.}\label{page:mnds}
Recall that a \emph{monad} $\T=(T,\eta,\mu)$ on a category $\C$ is given by an endofunctor $T$
on~$\C$ and two natural transformations, the 
\[\text{\emph{unit} }\eta\colon \Id_\C\to T\qquad\text{and}\qquad
\text{\emph{multiplication} }\mu\colon TT\to T,\] satisfying the usual unit and associative
laws. We denote by $\Kl(\T)$ the Kleisli category for $\T$; its objects are those of $\C$, and morphisms from $X$ to $Y$ are morphisms $f\colon X\to TY$ with the usual Kleisli composition. Moreover, we write $\EM(\T)$ for the \emph{Eilenberg-Moore} category for $\T$, the category of $\T$-algebras and their morphisms. The categories $\Kl(\T)$ and $\EM(\T)$ come with forgetful functors $U_\T\colon \Kl(\T)\to \C$, mapping $X\in \C$
to $TX$ and a Kleisli morphism $f\colon X\to TY$ to $f^\#=\mu_Y\circ Tf\colon TX\to
TY$, and $U^\T \colon \EM(\T) \rightarrow \C$ mapping an algebra to its carrier.
Moreover, there is a full embedding $I_\T\colon \Kl(\T)\subto \EM(\T)$ given by $X\mapsto
(TX,\mu_X)$ and $f\mapsto f^\#$ that identifies the Kleisli category $\Kl(\T)$ with the full
subcategory of $\EM(\T)$ given by free $\T$-algebras. For a monad $\T$ on $\Set$, the category
of sets and functions, we write $\Kl_\f(\T)$
for the full subcategory of $\Kl(\T)$ given by finite sets, or equivalently, the category of finitely generated free $\T$-algebras. We denote the domain restrictions of the above functors also by $U_\T\colon \Kl_\f(\T)\to \Set$ and $I_\T\colon \Kl_\f(\T)\subto \EM(\T)$.

\paragraph{Algebraic Categories.} One important class of monads are \emph{free-algebra
  monads} on $\Set$. Every algebraic theory $(\Sigma,E)$, specified by a signature a $\Sigma$ of
finitary operation symbols and a set $E$ of equations between $\Sigma$-terms, induces a monad
$\T$ on $\Set$ which assigns to each set~$X$ the free $(\Sigma,E)$-algebra generated by $X$,
carried by the set of $\Sigma$-terms modulo equations. The category $\EM(\T)$ is isomorphic to
the category of all $(\Sigma,E)$-algebras. The monad $\T$ is \emph{finitary} (i.e.~preserves directed colimits), and conversely, every finitary monad on $\Set$ is induced by some algebraic theory. In fact, \emph{all} monads on $\Set$ have an algebraic presentation by operations and equations, provided that large signatures and infinitary operations are admitted. For example, the category of algebras for the finite power set monad $\P_\f$ is isomorphic the category $\JSL$ of join-semilattices with bottom (equivalently, the category $\MSL$ of meet-semilattices with top). The full power set monad $\P$ yields the category $\CJsl$ of complete (semi)lattices and join-preserving maps. We tacitly identify isomorphic categories; for instance, the functor $I_{\P}\colon \Kl(\P)\subto \EM(\P)$ is identified with the functor $I_\P\colon \Kl(\P)\subto \CJsl$ sending a set $X$ to the free complete semilattice $\P X$. See Manes~\cite{Manes76} for more background on the relation between monads and algebraic theories.

\paragraph{Dual Adjunctions.}
Another natural source of monads are adjunctions. We use the notation
\(L \dashv R\colon \D\to \C\), or simply $L\dashv R$, for an adjunction where
$L\colon \C\to \D$ has a right adjoint $R\colon \D\to \C$. Every adjunction with unit \(\eta \colon \Id_\C \rightarrow RL\) and counit \(\varepsilon \colon LR \rightarrow \Id_\D\) induces a monad \((RL,\eta,R\varepsilon L)\) on \(\C\). We denote this monad by $RL$, leaving the structure implicit.

A \emph{dual adjunction} is an adjunction of type $L\dashv R\colon \D^\op\to \C$. Many dual
adjunctions of interest are given by \emph{dualizing objects}~\cite[Sec.~VI.4]{j82}, that is, $\C$ and $\D$ are concrete categories (with respective forgetful functors $|{-}|$ to $\Set$), and there are objects $S\in \C$ and $T\in \D$ with
\[ |S|=|T|,\qquad |LX| \cong \C(X,S), \qquad\text{and}\qquad |RY| \cong \D(Y,T). \]
Then the induced monad $RL$ is given by a `double hom-set' construction. For example, the dual adjunction 
$\Set(-,2)\dashv \Set(-,2)\colon \Set^\op\to \Set$ where $2=\{0,1\}$ yields the \emph{neighbourhood monad} on $\Set$. Similarly, by reading $2$ as a Boolean algebra, we obtain a dual adjunction $\Set(-,2)\dashv \BA(-,2)\colon \BA^\op\to \Set$ between the categories of sets and Boolean algebras that induces the \emph{ultrafilter monad}.
Both monads are studied in detail in \Cref{sec:simpl-codens-monads}.

\begin{rem}
  It is no coincidence that the neighbourhood monad and the ultrafilter monad look so similar:
  \emph{every} dual adjunction of type  $L\dashv R\colon \D^\op\to \Set$ is induced by a dualizing object.
  To see this, note that the contravariant left adjoint $L$ preserves colimits, thus sends copowers in $\Set$ to powers in $\D$.
  Since every set $X \in \Set$ is a copower $X \cdot 1$ of $1$, we have
  \[LX \cong L(X \cdot 1) \cong X \power L1,\]
  so $L \cong (-) \power D \colon \Set \rightarrow \D^\op$ for an object $D = L1 \in \D$.
  This determines $R \cong \D(-, D)$, since $(-) \power D \dashv \D(-, D)$, and therefore
  $RL \cong \D((-) \power D, D)$.
\end{rem}

%
%
%
%
%
%
%

\paragraph{Kan Extensions.}
 We let $[\C,\D]$ denote the possibly superlarge category of all functors between
    two categories $\C$ and $\D$ and natural transformations between them. The \emph{right Kan extension} of a functor \(F \colon \A \rightarrow \C\) along a functor \(J \colon \A \rightarrow \B\) is given by a functor \(\Ran_{J} F \colon \B \rightarrow \C\) and a bijection
  \begin{equation}
    \label{eq:ran-prop}
    [\B,\C](G, \Ran_{J}F) \cong [\A,\C](GJ, F)
  \end{equation}
  natural in \(G\). 
  Concretely, this means that there is a natural transformation \(\varphi \colon (\Ran_J F) J
  \rightarrow F \)
  (called the \emph{counit}) such that every natural transformation \(\alpha \colon GJ
  \rightarrow F\) factorizes as
  \begin{align}\label{eq:ran-prop-explicit}
    \alpha
    =
    (\, GJ \xto{\hat\alpha J} (\Ran _J F) J \xto{\varphi} F\,)\qquad \text{ for some unique \(\hat{\alpha} \colon G \rightarrow \Ran_J F\)}.
	\end{align} 
If the limit below exists for all $X\in \B$, the right Kan extension is given by 
	\begin{equation}\label{eq:ran}
		(\Ran_{J}F)X = \lim_{f \colon X \rightarrow JA} FA.
	\end{equation}
This holds, for instance, if the category $\A$ is small and $\C$ is complete. Right Kan
extensions of this type are called \emph{pointwise}. All Kan extensions 
in our applications are pointwise.

A functor $H \colon \C \rightarrow \D$ \emph{preserves} the right Kan extension $\Ran_{J}F$ of $F$ along $J$ if $H \cdot \Ran_J F$ is the right Kan extension of $HF$  along $J$:
\[
  \begin{tikzcd}
    \A \dar[swap]{J} \ar[pos=.75,phantom]{drr}{\cong} \ar{rr}{F} & & \C \dar{H} \\
    \B \ar[outer sep=-2pt]{urr}{\Ran_J F} \ar[swap,dashed]{rr}{\Ran_J(HF)} & & \D
  \end{tikzcd}
\]
%

%

%
%
%

\paragraph{Dense Functors.}
Given a  functor \(F \colon \A \rightarrow \B\) and \(B \in \B\), the  \emph{slice category} \(F \slice B\) has as objects all morphisms \(f \colon FA \rightarrow B\) for \(A \in \A\), and a morphism from $(f\colon FA\to B)$ to $(g\colon FA'\to B)$ is a morphism $h\colon A\to A'$ of $\A$ with \(f = g \cdot Fh\).
There is an obvious projection
\[\pi_{B} \colon F \slice B \rightarrow \A, \qquad (f \colon FA\rightarrow B) \mapsto A.\]
The functor $F$ is \emph{dense}~\cite{gu71} if the following equivalent conditions are satisfied:
\begin{enumerate}
\item Each $B\in \B$ is the colimit of the (possibly large) diagram
  \[
    F\slice B \xra{\pi_B} \A \xra{F} \B,
  \]
  with the colimit injections $f\colon FA \to B$ ($f\in F\slice B$).
  \item The functor $\B\to [\A^\op,\Set]$ given by $B\mapsto \A(F-,B)$ is fully
    faithful.
  \item The left Kan extension $\Lan_F F \colon \B \rightarrow \B$ (dual to the notion of right Kan extension) of $F$ along itself exists, is pointwise and isomorphic to the identity $\Id_\B$.
\end{enumerate}
A subcategory $\A\subto \B$ is \emph{dense} if its inclusion functor is dense. Informally, this states that each object of $\B$ can be canonically constructed from objects of $\A$.

\begin{exa}\label{ex:dense-set}
   The full subcategory $1 \incl \Set$ on the singleton is dense in $\Set$: every set is a coproduct of singletons.
\end{exa}

\begin{exa}\label{ex:dense-yoneda}
For every small category $\C$, the Yoneda embedding $\tilde{y}\colon \C^\op\subto [\C,\Set]$ given by $C\mapsto \C(C,-)$ is dense~\cite[Sec.~III.7, Thm.~1]{m98}.
\end{exa}

The example above is particularly powerful when combined with the following result:
\begin{thmC}[{\cite[Thm 5.13]{k82}}]\label{lem:dense-cancel}
  Let $G \colon \B \rightarrow \C$ be fully faithful.
  For every functor $F \colon \A \rightarrow \B$, if $GF$ is dense then so are $F$ and $G$.
\end{thmC}

\begin{exa}\label{ex:dense-alg}
In most of our applications, we will consider dense full subcategories of $\EM(\T)$ for a monad $\T$ on $\Set$. There are several generic choices of such subcategories:%
\begin{enumerate}
\item\label{ex:dense-alg:kleisli} The full subcategory $I_\T\colon \Kl(\T)\subto \EM(\T)$ of
        free algebras is dense.
        This is well-known; a quick argument is as follows:
        By~\cite[Theorem 14]{s72} the category $\EM(\T)$ is isomorphic to a full subcategory of $[\Kl(\T)^\op, \Set]$ that contains all representables, i.e.\ the Yoneda embedding factors as $\Kl(\T) \incl \EM(\T) \incl [\Kl(\T)^\op, \Set]$.
        Since this composite is dense by \Cref{ex:dense-yoneda} the statement follows by \Cref{lem:dense-cancel}.
\item\label{ex:dense-alg:kleisli-finitary} If $\T$ is finitary, the full subcategory
  $I_\T\colon \Kl_\f(\T)\subto \EM(\T)$ of finitely generated free algebras is also
  dense; see e.g.~\cite[Lem.~4.2]{arv10} for the corresponding result on Lawvere theories, which are equivalent to finitary monads on $\Set$. Moreover, if there is an algebraic theory presenting the monad $\T$ with an upper bound $n>0$ to the arities of operations, the one-object full subcategory $\{\T n\}\subto \EM(\T)$ given by the free algebra on $n$ generators is dense~\cite[Sec.~2.2]{isbell60}.
\item\label{ex:dense-alg:fp} If $\T$ is finitary, another dense full subcategory
  $(\EM(\T))_\fp\subto \EM(\T)$ is given by the finitely presentable algebras (i.e.~algebras
  presentable by finitely many generators and relations). This follows from the fact that
  $\EM(\T)$ is locally finitely presentable~\cite[Cor.~3.7]{ar94} and that in any such
  category the finitely presentable objects form a dense
  subcategory~\cite[Ex.~1.24.1]{ar94}. Note that in case the monad $\T$ preserves
  finite sets, $(\EM(\T))_\fp$ coincides with the full subcategory $(\EM(\T))_\f\subto
        \EM(\T)$ of all finite algebras.
        This is for example the case for sets ($\Set \simeq \EM(\Id)$) and join-semilattices ($\JSL \simeq \EM(\Pfin)$).
\end{enumerate}
\end{exa}

Dual to density, we have the notion of a \emph{codense} functor $F\colon \A\to \B$: every object $B\in \B$ is the limit of the canonical diagram of all morphisms $f\colon B\to FA$ with $A\in \A$.

\begin{nota}
	For \(F\) and \(B\) as above and \(G \colon \A \rightarrow \C\), we write
	\[\colim_{\substack{f \colon FA \rightarrow B \\ A \in \A}} GA \qquad\text{or simply}\qquad 
\colim_{\substack{f \colon FA \rightarrow B}} GA\] for the colimit of the diagram $(F \slice B) \xra{\pi_B} \A \xra{G} \C$; similarly for limits.
\end{nota}

\section{Codensity Monads}
\label{sec:codensity-monads}

We recall the definition and some basic facts about codensity monads. The notion was originally
introduced by Kock~\cite{kock66} who gave it its name. Later, Applegate and Tierney
independently defined the dual concept~\cite{ApplegateT69}, calling it \emph{model-induced cotriple}.

\begin{defi}[Codensity Monad]
  The \emph{codensity monad} of a functor \(F\colon \A \to \C \) is the monad \[\Cody{F}=(\Ran_F F, \eta, \mu)\]
  where $\Ran_F F$ is the right Kan extension of $F$ along itself (if it exists), and the unit and multiplication 
  \[\eta\colon \Id\to \Ran_F F \qquad \text{and}\qquad \mu\colon (\Ran_F F)(\Ran_F F) \to \Ran_F F\] are the natural transformations corresponding to the natural transformations below, where $\varphi$ denotes the counit of $\Ran_F F$:
  \[ F\xto{\id} F \qquad\text{and}\qquad  (\Ran_F F)(\Ran_F F) F \xto{(\Ran_F F) \varphi} (\Ran_F F) F \xto{\varphi} F. \]
\end{defi}
We say that $\Cody{F}$ is \emph{pointwise}, if so is the right Kan extension $\Ran_F F$. In
this case,
\[\Cody F X = \lim_{f \colon X \rightarrow FA} FA.\] 
Moreover, the action of $\Cody{F}$ on a morphism $h\colon X'\to X$ and the unit and multiplication at $X\in \C$ are uniquely determined by the commutative diagrams below, where $f$ ranges over all $f\colon X\to FA$ with $A\in \A$ and $\pi_f\colon \Cody{F} X \to FA$ is the corresponding limit projection:
\[
\begin{tikzcd}[column sep=40]
\Cody{F} X' \ar{dr}[swap]{\pi_{f\circ h}} \ar{r}{\Cody{F} h} & \Cody{F} X \ar{d}{\pi_f} \\
& FA 
\end{tikzcd}
\qquad 
\begin{tikzcd}[column sep=30]
X \ar{dr}[swap]{f} \ar{r}{\eta_X} & \Cody{F} X \ar{d}{\pi_f} & \Cody{F}\Cody{F} X \ar{l}[swap]{\mu_X} \ar{dl}{\pi_{\pi_f}} \\
& FA 
\end{tikzcd}
\]
\begin{defi}[Codensity Presentation]
  A \emph{codensity presentation} of a monad $\T$ on $\C$ is a functor $F\colon \A\to \C$
  such that $\T$ is the codensity monad of $F$; in symbols: \[ \T\cong \Cody F.\]
\end{defi}
The definition of the codensity monad as a right Kan extension endows it with a universal property as a functor. It also has a universal property as a \emph{monad}, due to Street~\cite{s72} (in fact, he defines codensity monads in a 2-categorical setting via this universal property). Recall that an \emph{$\mathcal{S}$-module over $\A$} for a monad $\mathcal{S}$ on $\C$ is an algebra for the monad $\mathcal{S}_*$ on
$[\A, \C]$ given by postcomposition: $\mathcal{S}_*(F) = SF$. Thus an $\mathcal{S}$-module is given by a functor $F\colon \A\to \C$ together with a natural transformation $\alpha\colon SF\to F$ satisfying the associative and unit laws. In particular, $\mathcal{S}$-modules over $\A=1$ (the terminal category) are just $\mathcal{S}$-algebras.
\begin{thmC}[{\cite{s72}}]\label{prop:monad-ext-univ-prop}
    Let $F \colon \A \rightarrow \C$ be a functor whose codensity monad exists, and let
    $\mathcal{S} = (S, \eta, \mu)$ be a monad on $\C$.
    There exists an isomorphism, natural in $\mathcal{S}$,
    between $\mathcal{S}$-module structures on $F$ and monad morphisms $\mathcal{S} \rightarrow \Cody F$.
  \end{thmC}
\begin{rem}\label{R:cody}
\begin{enumerate}
\item The codensity monad may be seen as a measure of codensity of the functor \(F\); indeed,~$F$
is codense if and only if its codensity monad is pointwise and the identity monad.
\item\label{R:cody:2} The construction of codensity monads generalizes the
  construction of monads from adjunctions: If \(R\) has a left adjoint \(L\), then the codensity monad of \(R\) exists and is isomorphic to the induced monad \(RL\)~\cite{l13}. 
 In particular, since every monad arises from an adjunction, every monad is a codensity
 monad. Hence the challenge is not to find \emph{some} codensity presentation of a given monad, but rather to find a {simple} and natural~one.
\end{enumerate}
\end{rem}

  Codensity monads are an instance of a more general concept that extends arbitrary monads along a given functor, known in the literature under the name of \emph{monad extension}~\cite{s72} or
    \emph{pushforward}~\cite{m25}
  \begin{defi}[Pushforward]
    The \emph{pushforward} of a monad $\mathcal{T}$ on $\A$ along a functor $F \colon \A \rightarrow \C$ is the monad $F_{\#} \mathcal{T}$ on $\C$ given by the
    right Kan extension $F_{\#}T = \Ran_{F}FT$ of $FT$ along $F$ (if it exists), and unit and multiplication analogous to the codensity monad.
    \[
      \begin{tikzcd}
        \A
        \ar{r}{T}
        \ar{d}[swap]{F}
        &
        \A
        \ar{d}{F}
        \\
        \C
        \ar[dashed]{r}{F_{\#}T}
        &
        \C
      \end{tikzcd}
    \]
  \end{defi}
  The codensity monad of $F$ is thus the pushforward $F_{\#} \Id$ of the identity monad  on $\A$.
  \begin{rem}\label{rem:monad-ext}
    In this paper we just focus on codensity monads.
    While this might seem like a restriction, it formally is not: As observed by Do\~{n}a Mateo \cite[Cor.~1.11]{m25}, the pushforward $F_{\#} \T$ is
    isomorphic to the codensity monad of $\EM(\T) \xra{U^\T} \A \xra{F}\C$.
    In fact, his argument works not just for $U^\T$ but all right adjoints $U \colon \B \rightarrow \A$ whose induced monad is isomorphic to $\T$:
    \[\Ran_{FU}FU \cong \Ran_{F}(\Ran_U FU) \cong \Ran_{F}(F\,\Ran_U U) = \Ran_F(F\, \Cody{U}) \cong \Ran_FFT = F_{\#}\T.\]
    In the second step we use that right Kan extensions along right adjoints are \emph{absolute}, that is, preserved by any functor.

    We will sometimes encounter codensity presentations of the form~$FU^\T$, so it makes
    sense to think of them as `pushforward presentations'.
    Pushforward presentations often give a conceptually clearer
    description than a codensity presentation of a monad (compare e.g.~the codensity and pushforward presentations of the filter monad on $\Set$ in \Cref{sec:filter}).
  \end{rem}

\section{Codensity Monads = Density + Duality}
\label{sec:main-theorem}

In this section we present our core technical result, a simple and general criterion for a monad induced by a given dual adjunction to be the codensity monad of a given functor.
The following setting generalizes a common pattern: a dual adjunction between two categories restricting to an adjoint equivalence of certain subcategories.

\begin{defi}[Codensity Setting]\label{def:codensity-sit}
A \emph{codensity setting} is given by categories and functors as shown in diagram \eqref{eq:codensity-sit}, where
(1)~$L$ is left adjoint to $R$, (2)~$E$ is an equivalence, (3)~$G$ is dense, and~(4)~the outside commutes up to natural isomorphism ($RG^\op\cong FE$).
\begin{equation}\label{eq:codensity-sit}
  \begin{tikzcd}
    \C_{0}
    \rar{F} &
    \C
    \ar[yshift=0pt, bend right=20,]{d}[swap]{L}
    \ar[yshift=0pt, bend left=20, leftarrow, swap]{d}[swap]{R}
    \ar[phantom]{d}{\dashv}
    \\
    \D_{0}^\op \rar{G^\op}
    \ar{u}{E}[swap]{\rotatebox{-90}{$\simeq$}}
    &
    \D^\op
  \end{tikzcd}
  \end{equation}
\end{defi}

The dual adjunction $L\dashv R$ of a codensity setting induces a monad on the category~$\C$, and the decomposition of
the functor $F$ is enough to deduce that $F$ is a codensity presentation of that
monad:

\begin{thm}\label{thm:codensity}
  In every codensity setting \eqref{eq:codensity-sit}, the codensity monad of the functor
  $F$ exists, is pointwise, and is isomorphic to the monad induced by the adjunction
  $L\dashv R$:
  \begin{align*}\label{eq:codensity}
    RL \cong \Cody{F}.
  \end{align*}
\end{thm}

\begin{proof}
  \takeout{
  The theorem is essentially a rephrasing of the following fact due to Mateo:
  \begin{notheorembrackets}
    \begin{thmC}[{\cite[Prop.\ 2.15]{m25}}]\label{prop:mateo}
      Let $R \colon\C \rightarrow \D$ and $H \colon \A \rightarrow \C$ be functors such that $R$ is a right adjoint and $H$ is codense.
      Then $RH$ and $R$ have isomorphic codensity monads: $\Cody{RH} \cong \Cody{R}$.
    \end{thmC}
  \end{notheorembrackets}
  The theorem now simply follows as
  \[\Cody{F} \cong \Cody{RG^\op E} \cong \Cody{R} \cong RL,\]
  since $G^\op E$ is \emph{co}dense and the codensity monad of $R$ is $RL$, see~\iref{R:cody}{2}.
  
  To get an intuition for concrete codensity monads given by pointwise right Kan extensions,
  it is instructive to spell out the isomorphism:}
  We have the following isomorphisms that are natural in $X\in \C$:
    \begin{align*}
      RLX ~&\cong~R(\colim_{{f\colon GD\to LX}} GD) && \text{$G$ dense} \\
           &\cong~\lim_{{f\colon GD\to LX}} RGD && \text{right adjoints preserve limits}  \\
           &\cong~\lim_{{g\colon X\to RGD}} RGD && L\dashv R\colon \D^\op\to \C  \\
           &\cong~\lim_{{g\colon X\to FED}} FED && RG^\op\cong FE  \\
           &\cong~\lim_{{g\colon X\to FC}} FC && \text{$E$ equivalence functor} \\
           &\cong~(\Ran_F F)X  && \text{limit formula for~$\Ran$}.
    \end{align*}

    We verify that this natural isomorphism \(\omega\colon RL\cong \Cody{F}\) is an
    isomorphism of monads: it preserves the unit and multiplication. In the following, we
    denote by $\eta$ and $\mu$ the unit and multiplication of the monad $RL$, and $\eta^F$
    and $\mu^F$ denote those of $\Cody{F}$.

To show preservation of the unit, for every \(g \colon X \rightarrow FC\) we take some \(D\) with \(ED \cong C\), and let \(f\colon GD\to LX\) be the adjoint transpose of
\[\hat{f}\colon X \rightarrow FC \cong FED \cong RGD. \]
The diagram below then proves that $\omega$ preserves the unit: the outer path commutes, since it commutes when postcomposed with the projections \(\pi_g\) of the limit cone defining \(\Cody F\).
\newcommand{\ds}{\displaystyle}
\[
  \begin{tikzcd}[column sep=6, outer sep=0]
    RLX \rar[phantom]{\cong} \drar[near end]{Rf} & \ds{R(\colim_{f \colon GD \rightarrow LX} GD)} \rar[phantom]{\cong} \dar{R \kappa_f} & \ds{R(\colim_{\hat{f} \colon X \rightarrow RGD} GD)} \rar[phantom]{\cong} \dar{R \kappa_{\hat{f}}} & \ds{\lim_{\hat{f} \colon X \rightarrow RGD} RGD} \rar[phantom]{\cong} \dlar[swap]{\pi_{\hat{f}}}
    & \ds{\lim_{g \colon X \rightarrow FC} FC} \dlar[swap]{\pi_g} \\
    & RGD \rar[equals]{} & RGD \rar[equals]{\cong} & FC & \\
    & & X \ular[swap]{\hat{f}} \ar{ur}{g} \ar[bend left=26]{uull}{\eta} \ar[bend right=20, swap]{uurr}{\eta^F} & &
  \end{tikzcd}
\]
Here $\kappa_f$ are the coprojections of the canonical colimit. Furthermore, note that the
upper half without its right-hand part yields
\begin{equation}\label{eq:aux}
  R\hat g = \big( RLX \xra{\omega} \lim_{X \to RGD} RGD \xra{\pi_g} RGD\big)
  \qquad
  \text{for every $g\colon X \to RGD$}.
\end{equation}
(That is, we take this part for $f = \hat g\colon GD\to LX$.)

For the multiplication,
we simplify \[\Cody F \cong \lim_{\substack{X \rightarrow RGD \\ D \in \D_0}} RGD\] to make the diagram more digestible.
Given \(g \colon X \rightarrow RGD\) note that the adjoint transpose \(\widehat{\pi_g}\) of \(\pi_g \colon \Cody F X \rightarrow RGD\) is given by
\begin{equation}
  \label{eq:pi_g}
\widehat{\pi_g}  = \big( GD \xra{\varepsilon_{GD}} LRGD \xra{L \pi_g} L \Cody F X\big).
\end{equation}
Therefore, we have the following commutative diagram:
\[
  \begin{tikzcd}
    RLRLX \ar{rr}{RL \omega} \dar{R \varepsilon L} \drar{RLR \hat{g}} & & \ds{RL \lim_{X \rightarrow RGD} RGD} \dlar{RL \pi_g} \dar[phantom]{\rotatebox{90}{\(\cong\)}} \\
    RLX \dar{\omega} \drar{R \hat g} &  RLRGD \dar{R \varepsilon} & RL \Cody F X \dar[phantom]{\rotatebox{90}{\(\cong\)}} \dlar[swap]{R \widehat{\pi_g}} \\
    \ds{\lim_{X \to RGD} RGD } \rar{\pi_g} & RGD & \ds{R \colim_{GD \to L \Cody F X} GD } \lar[swap, near start]{R \kappa_{\hat{\pi}_g}} \dlar[swap, xshift=20pt, equals]{\rotatebox{30}{$\sim$}}  \\
    & \ds{\lim_{\Cody F X \to RGD} \lim RGD} \ular{\mu^F} \uar{\pi_{\pi_g}} &
  \end{tikzcd}
\]
The upper triangle commutes, since it is \eqref{eq:aux} under $RL$; similarly, the middle
left-hand triangle is \eqref{eq:aux}. The right trapezoid commutes by \eqref{eq:pi_g}, the left trapezoid is naturality of \(\varepsilon\), and the rest are just properties of (co)limit (co)projections.
Since the multiplication \(\mu\) of \(RL\) is \(R \varepsilon L\), this states precisely that \(\omega \cdot \mu = \mu^F \cdot (\omega \circ \omega)\).
\end{proof}

\begin{rem}\label{rem:dona-mateo}
   Modulo the fact that $\Cody{F}$ is pointwise, \Cref{thm:codensity} is a consequence of 
   a recent result due to Do\~na Mateo~\cite[Prop.\ 2.15]{m25}: for a right adjoint $R
   \colon\C \rightarrow \D$ and a codense functor $H \colon \A \rightarrow \C$, the
   codensity monads of $RH$ and $R$ are isomorphic. 

  Indeed, given a codensity setting \eqref{eq:codensity-sit}, take a pseudoinverse $E'\colon \C_0 \to \D_0^\op$ of $E$,
   use that $H=G^\op E'$ is \emph{co}dense (since $G^\op$ is codense and $E'$ is an equivalence functor) and that the codensity monad of $R$ is $RL$
   (\iref{R:cody}{2}) to obtain
   \[
     \Cody{F} \cong \Cody{RG^\op E'} \cong \Cody{RH} \cong \Cody{R} \cong RL.
   \]
The key insight behind our \Cref{thm:codensity} is that it suggests a simple and concrete recipe to construct the codense functor $H$, namely as the composite of a dual equivalence with the opposite of a dense functor.  Our applications in the subsequent sections show that this duality-theoretic refinement of  Do\~na Mateo's result can be very powerful in practice.
\end{rem}

\begin{rem}\label{rem:abstract-codensity-situation}
  \begin{enumerate}
    \item\label{rem:abstract-codensity-situation:adj} If $\C_0$ is small and $\C$ is complete then
          the codensity monad of any functor $F \colon \C_0 \rightarrow \C$ has a codensity setting: $\Cody F$ is induced by the  dual adjunction with the functor category $[\C_0, \Set]$
          \begin{equation}\label{eq:abs-sit}
            \begin{tikzcd}
              \C_{0}
              \rar{F} &
              \C
              \ar[yshift=0pt, bend right=20,]{d}[swap]{N_F}
              \ar[yshift=0pt, bend left=20, leftarrow, swap]{d}[swap]{R_F}
              \ar[phantom]{d}{\dashv}
              \ar[out=20,in=-20,loop,looseness=5, "\Cody{F}",pos=0.49]
              \\
              (\C_{0}^\op)^\op \rar{y^\op}
              \ar[equals]{u}{}
              &
              {[\C_0, \Set]^\op},
            \end{tikzcd}
          \end{equation}
          where $N_F C = \C(C, F(-))$ is the \emph{conerve} and its right adjoint $R$ is \emph{totalization} given by the right Kan extension $\Ran_{\tilde{y}}F$ of $F$ along the contravariant Yoneda embedding $\tilde{y} \colon \C_0 \rightarrow [\C_0, \Set]^\op$.
          The Yoneda embedding $y \colon \C_0 \rightarrow [\C_0^\op, \Set]$ is dense, and the adjunction induces $\Cody F$ by definition~\cite[Section 2]{l13}.

    \item Di Liberti~\cite{dl20} showed that if one additionally assumes density of $F$ and cocompleteness of $\C$, then the adjunction~\eqref{eq:abs-sit} factorizes through Isbell's dual adjunction between presheaves and copresheaves over $\C_0$ (cf.\ \Cref{sec:isbell-duality}).
          This result is conceptually useful for the investigation of algebras of such codensity monads, yet does not appear to simplify finding concrete codensity presentations.
          In addition, for a wide class of codensity presentations, such as for the probability monads considered in \Cref{sec:probability-monads}, the density assumption is not satisfied.
    \item\label{rem:abstract-codensity-situation:compositional} Codensity settings are compositional under right adjoints~\cite[Lem.~1.9]{m25}: If $U \colon \C \rightarrow \C'$ is another right adjoint with left adjoint $P$ then \[\Cody{UF} \cong U_{\#} \Cody F \cong URLP \cong U \Cody F P.\]
  \end{enumerate}
\end{rem}

In the following sections, we will illustrate the wide scope of \Cref{thm:codensity} and use it to derive codensity presentations for a number of popular monads.
To show that a given monad~$\T$ on~$\C$ is the codensity monad of a functor $F\colon \C_0\to \C$, we employ a uniform recipe: 
\begin{enumerate}
\item Identify a suitable dual adjunction $L\dashv R\colon \D^\op\to \C$ inducing the monad $\T$.
\item Extend the functor $F$ and the adjunction $L\dashv R$ to a codensity setting~\eqref{eq:codensity-sit}.
\item Apply \Cref{thm:codensity} to conclude $\T\cong \Cody{F}$.
\end{enumerate}

In all our applications, the crucial (and sometimes non-trivial) step is the choice of
the dual adjunction $L\dashv R$ in Step~(1). Then Step~(2) is in most cases straightforward:
The functor~$E$ in \eqref{eq:codensity-sit} is given by some simple dual equivalence $\D_0^\op\simeq \C_0$ known in the literature, and the likewise density of the functor $G$ amounts to some standard property of the (typically algebraic) category $\D$, like the ones from~\Cref{ex:dense-alg}.

\section{Ultrafilter and Double Dualization Monads}
\label{sec:simpl-codens-monads}

The first class of codensity monads that we cover as instances of \Cref{thm:codensity} are (ultra)filter monads and their close relatives, double dualization monads~\cite{l13,as21,dl20,kg71}.

\subsection{Ultrafilter Monads}\label{sec:ultrafilter}
We start with what is probably the best known instance of a codensity monad: the characterization of the ultrafilter monad on $\Set$ as the codensity monad of the inclusion $\Set_\f\subto \Set$ of finite sets into sets. This classical result goes back to Kennison and Gildenhuys~\cite{kg71}; see also Leinster~\cite{l13} for a streamlined exposition.

Let $\BA$ denote the category of Boolean algebras. An \emph{ultrafilter} on a Boolean
algebra~$B$ is a non-empty subset $U\subseteq B$ such that (1)~$U$ is upwards closed,  (2)~$U$ is
closed under meet, and (3)~for every $b\in B$, either $b\in U$ or $\neg b\in
U$. Equivalently, an ultrafilter is given by a morphism \(\chi \in \BA(B, 2)\), where
$2=\{0,1\}$ is the two-element Boolean algebra, by identifying $\chi$ with the preimage
$\chi^{-1}(1)\subseteq B$. An \emph{ultrafilter} on a set $X$ is an ultrafilter on the
Boolean algebra $\Set(X,2)\cong \P X$ of predicates (equivalently subsets) of $X$. The
\emph{ultrafilter monad} $\U$ on $\Set$ sends a set $X$ to the set of
$\U X = \BA(\Set(X,2),2)$ of its ultrafilters; more precisely, $\U$ is the monad induced by the adjunction in \eqref{eq:codensity-setting-ultrafilter} below.
The category $\EM(\U)$ of algebras for $\U$ is isomorphic to the category of compact Hausdorff spaces and continuous maps~\cite[1.5.24--33]{Manes76}; hence the ultrafilter monad provides a bridge between algebra, Boolean logic, and topology.

\noindent The ultrafilter monad $\U$ is captured by codensity setting~\eqref{eq:codensity-setting-ultrafilter} shown below:
\begin{equation}\label{eq:codensity-setting-ultrafilter} 
\begin{tikzcd}[column sep=25, row sep=25]
    \Set_\f
    \ar[hook]{r}{I} &
    \Set
    \ar[out=20,in=-20,loop,looseness=3, "\U",pos=0.49]
    \ar[yshift=0pt, bend right=20,]{d}[swap]{\Set(-,2)}
    \ar[yshift=0pt, bend left=20, leftarrow, swap]{d}[swap]{\BA(-,2)}
    \ar[phantom]{d}{\dashv}
    \\
    \BA_\f^\op \ar[hook]{r}{J^\op}
    \ar{u}{\BA(-,2)}[swap]{\rotatebox{-90}{$\simeq$}}
    &
    \BA^\op
  \end{tikzcd}
\end{equation}

Here $I$ and $J$ are the inclusions of the full subcategories of finite sets and
finite Boolean algebras, respectively, and the equivalence functor on the left
is the restriction of \emph{Birkhoff duality}~\cite{b37}  between distributive lattices and finite
posets to finite sets (discrete posets). The functor~$J$ is dense by \iref{ex:dense-alg}{fp}; note that since the free Boolean algebra $2^{2^X}$ on a finite set $X$ is finite, we have $\BA_\fp=\BA_\f$. From
\Cref{thm:codensity} we obtain the following theorem due to Kennison and Gildenhuys:

\begin{thmC}[\cite{kg71}]\label{thm:codensity-ultrafilter}
The ultrafilter monad $\U$ on $\Set$ is the codensity monad of the inclusion $\Set_\f\hookrightarrow \Set$.
\end{thmC}
In \Cref{sec:appl} we show how to use this presentation of the ultrafilter monad to
recover Börger's characterization~\cite{b87} of the ultrafilter functor.

\begin{rem}\label{rem:codensity-ultrafilter-3}
  Leinster~\cite{l13} noted that it suffices to restrict $I \colon \Set_\f \incl \Set$ to any full subcategory $\C \subseteq \Set_\f$ containing a set with at least three elements.
  The duality $\Set_\f \simeq \BA_\f^\op$ underlying the codensity setting \eqref{eq:codensity-setting-ultrafilter} gives a simple explanation for this fact:
  If the sets in $\C$ have at most two elements, their dual
  Boolean algebras $2^1 = 2^{2^0}$ and $2^2 = 2^{2^1}$  are free with zero and one generators, respectively.
  However, a dense subcategory of $\BA$ requires a free algebra with at least two generators
  (\iref{ex:dense-alg}{kleisli-finitary});
  for example, one can show that for the inclusion $K \colon \{2^1, 2^2\} \incl \BA$ the canonical colimit $\colim_{KA \rightarrow 2^3}2^3$ is isomorphic to the free algebra $2^{2^3}$ on three generators, rather than to the algebra $2^3$.
\end{rem}

Next, in the codensity setting~\eqref{eq:codensity-setting-ultrafilter} we may reverse the roles of $\Set$ and $\BA$:
\begin{equation}\label{eq:codensity-setting-can-ext} 
\begin{tikzcd}[column sep=25, row sep=25]
    \BA_\f \ar[hook]{r}{J}
    & \BA
    \ar[out=20,in=-20,loop,looseness=3, "\mathcal{C}",pos=0.49]
    \ar[yshift=0pt, bend right=20,]{d}[swap]{\BA(-,2)}
    \ar[yshift=0pt, bend left=20, leftarrow, swap]{d}[swap]{\Set(-,2)}
    \ar[phantom]{d}{\dashv}
    \\
    \Set_\f^\op
    \ar{u}{\Set(-,2)}[swap]{\rotatebox{-90}{$\simeq$}}
    \ar[hook]{r}{I^\op} 
    & \Set^\op
  \end{tikzcd}
\end{equation}

  The yields another codensity setting since the inclusion $\Set_\f \incl \Set$  is dense and the outside of
  \eqref{eq:codensity-setting-can-ext} commutes.
  From \Cref{thm:codensity} we obtain:
  \begin{thm}\label{thm:canonical-extension}
    The monad $\mathcal{C} = \Set(\BA(-, 2), 2)$ on \BA is the codensity monad of the inclusion $J \colon \BA_\f \incl \BA$.
  \end{thm}
  \begin{rem}
    For a Boolean algebra $A$, the Boolean algebra $\mathcal{C} A$ is the \emph{canonical
      extension} of~$A$, which is conceptually important for applications in logic and
    duality theory in general; for an introduction to canonical extensions see~\cite{gv10}.
    Note that $\mathcal{C} A$ is a complete atomic Boolean algebra. Hence, it may be seen as
    a \emph{compactification} of $A$, since the category $\Caba$ of all complete atomic
    Boolean algebras is equivalent to the category $\cat{CHausBA}$ of compact Hausdorff
    Boolean algebras (i.e.~Boolean algebras carried by a compact Hausdorff topological space
    with continuous operations)~\cite[Cor.~4.11]{j82}. So if we regard $\mathcal{C} A$ as a
    complete atomic Boolean algebra, then we obtain a  left adjoint to the forgetful functor
    \[
      \Caba \simeq \cat{CHausBA} \rightarrow \BA.
    \]
    Indeed, we have a natural bijection
    \begin{equation*}
      \Caba(\mathcal{C} A, 2^X) \cong \Set(X, \BA(A, 2)) \cong \BA(A, 2^X)
      \ \ 
      \text{for every $A \in \BA$ and $2^X \in \Caba$},
    \end{equation*}
    where we use that $2^X$ is the $X$-fold power of $2$ in $\BA$.

    For a more detailed investigation of algebras for codensity monads in connection with compact Hausdorff algebras, see~\cite{d16,dl20}.
  \end{rem}

  A topological version of \cref{thm:codensity-ultrafilter} was given by
  Sipo\c{s}~\cite{s18} for the ultrafilter monad~$\ol{\U}$ on the category $\Top$ of
  topological spaces and continuous maps.  It assigns to a space \(X\) the space
  \(\ol \U X = \BA(\Top(X, 2), 2)\) of ultrafilters of clopens; that is, $\ol{\U}$ is the
  monad given by the adjunction on the right in the diagram below:
  \begin{equation}\label{eq:codensity-setting-ultrafilter-top}
    \begin{tikzcd}[column sep=25, row sep=25]
      \Set_\f
      \ar[hook]{r}{I} &
      \Top
      \ar[out=20,in=-20,loop,looseness=3, "\ol\U"]
      \ar[yshift=0pt, bend right=20,]{d}[swap]{\Top(-,2)}
      \ar[yshift=0pt, bend left=20, leftarrow, swap]{d}[swap]{\BA(-,2)}
      \ar[phantom]{d}{\dashv}
      \\
      \BA_\f^\op \ar[hook]{r}{J^\op}
      \ar{u}{\BA(-,2)}[swap]{\rotatebox{-90}{$\simeq$}}
      &
      \BA^\op
    \end{tikzcd}
  \end{equation}
Here $2$ is the two-element discrete space, $\Top(X,2)\cong \Cl X$ is the Boolean algebra of
clopen subsets of a space~$X$, and $\BA(B,2)$ is the space of ultrafilters of a Boolean algebra~$B$, viewed as subspace of the space $2^{|B|}$ equipped with the product topology.
Algebras for $\ol{\U}$ correspond to Stone spaces (compact Hausdorff spaces with a basis of clopens)~\cite{s18}.

To capture the topological ultrafilter monad in our setting, we simply just use $\Top$ in lieu of $\Set$ in \eqref{eq:codensity-setting-ultrafilter}, which leads to the codensity setting \eqref{eq:codensity-setting-ultrafilter-top}. Here $I$ identifies a finite set with a discrete topological space. From \Cref{thm:codensity} we recover Sipo\c s' theorem:

\begin{thmC}[\cite{s18}]
The ultrafilter monad $\ol{\U}$ on $\Top$ is the codensity monad of the inclusion $\Set_\f\subto \Top$.
\end{thmC}

\subsection{Filter Monads}\label{sec:filter}
A natural generalization of ultrafilter monads are \emph{filter monads}. They are captured in our setting by working with semilattices in lieu of Boolean algebras.

In more detail, let $\MSL$ denote the category of meet-semilattices with a top element
(equivalently, algebras for the finite power set monad $\P_\f$). A \emph{filter} on a
meet-semilattice $M$ is a non-empty subset
$F\subseteq M$ that is both upwards closed and closed under meets. A filter corresponds to a morphism \(\chi \in \MSL(M, 2)\), where $2=\{0,1\}$ is the
two-element semilattice with the meet given by minimum, by identifying $\chi$ with the
preimage $\chi^{-1}(1)\subseteq M$. A \emph{filter} on a set
$X$ is a filter on the semilattice $\Set(X,2)\cong \P X$. The \emph{filter monad} $\F$ on
$\Set$ sends a set $X$ to the set $\F X = \MSL(\Set(X,2),2)$ of its filters; that is, $\F$ is the monad induced by the adjunction in \eqref{eq:codensity-setting-filter}. Algebras for $\F$ correspond to complete lattices where arbitrary meets distribute over directed joins~\cite{day75}. Two suitable codensity settings for $\F$ are as follows:
\noindent
\begin{minipage}[t]{.5\textwidth}
  \begin{equation}\label{eq:codensity-setting-filter}
    \begin{tikzcd}[column sep=25, row sep=25]
      \MSL_\f
      \ar{r}{U} &
      \Set
      \ar[out=20,in=-20,loop,looseness=3, "\F",]
      \ar[yshift=0pt, bend right=20,]{d}[swap]{\Set(-,2)}
      \ar[yshift=0pt, bend left=20, leftarrow, swap]{d}[swap]{\MSL(-,2)}
      \ar[phantom]{d}{\dashv}
      \\
      \MSL_\f^\op \ar[hook]{r}{J^\op}
      \ar{u}{\MSL(-,2)}[swap]{\rotatebox{-90}{$\simeq$}}
      &
      \MSL^\op
    \end{tikzcd}
  \end{equation}
\end{minipage}
\begin{minipage}[t]{.03\textwidth}
~
\end{minipage}
\begin{minipage}[t]{.45\textwidth}
  \vspace{-.1cm}
  \begin{equation}\label{eq:codensity-setting-filter-kleisli}
    \begin{tikzcd}[column sep=25, row sep=25]
      \Kl_\f(\Pfin)
      \ar{r}{U_{\Pfin}} &
      \Set
      \ar[out=20,in=-20,loop,looseness=3, "\F",]
      \ar[yshift=0pt, bend right=20,]{d}[swap]{\Set(-,2)}
      \ar[yshift=0pt, bend left=20, leftarrow, swap]{d}[swap]{\MSL(-,2)}
      \ar[phantom]{d}{\dashv}
      \\
      \Kl_\f(\Pfin)^\op \ar[hook]{r}{I_{\Pfin}^\op}
      \ar{u}{(-)^\op}[swap]{\rotatebox{-90}{$\simeq$}}
      &
      \MSL^\op
    \end{tikzcd}
  \end{equation}
\end{minipage}

In setting \eqref{eq:codensity-setting-filter}, $U$ is the forgetful functor, and $J$ is 
the inclusion (\iref{ex:dense-alg}{fp}) of finitely presentable algebras.
Here we use that $\EM(\Pfin) \cong \MSL$, and view $\MSL$ as the variety of commutative idempotent monoids. Note that finitely presentable semilattices coincide with the finite ones, since free finitely presentable semilattices are finite. Moreover, we make use of the self-duality $\MSL_\f^\op\simeq \MSL_\f$ of finite
meet-semilattices~\cite[Sec.~VI.3.6]{j82}, which maps a finite meet-semilattice~$M$ to the semilattice $\MSL(M,2)$
of its filters, with the semilattice structure defined pointwise. 

This codensity setting
restricts to~\eqref{eq:codensity-setting-filter-kleisli}; to see this, first recall the functors $U_{\Pfin}$ and
$I_{\Pfin}$ from \Cref{sec:preliminaries} (p.~\pageref{page:mnds}, `Monads')
and that $I_\Pfin$ is dense (\iref{ex:dense-alg}{kleisli-finitary}).
(Note that it maps a set $X$ to $\Pfin X$ equipped with $\cup$ as the binary operation.)
Here $\Kl_\f(\Pfin)^\op\simeq \Kl_\f(\Pfin)$ is the standard identity-on-objects self-duality of the category of finite sets and relations, sending $f\colon X\to \Pfin Y$ to $f^\op\colon Y\to \Pfin X$ given by
$f^\op(y)= \{x : y \in f(x)\}$. We show that the outside of~\eqref{eq:codensity-setting-filter} commutes.
Note that both $U_\Pfin$ and $I_\Pfin$ map $f\colon X \to \Pfin Y$ to the (meet preserving) map $f^\#\colon \Pfin X \to \Pfin Y$ defined by $f^\#(S) = \bigcup_{x \in S} f(s)$ for every finite subset $S \subseteq X$.
It is now easy to see that both paths in the diagram map a given $f\colon X \to \Pfin Y$ in $\Kl_f(\Pfin)^\op$ to the map $\Pfin Y \to \Pfin X$ that takes the preimage of $S\subseteq Y$ under~$f^\#$.

From \Cref{thm:codensity} we obtain two codensity presentations of the filter monad:

\begin{thm}
  The filter monad $\F$ on $\Set$ is the codensity monad of the forgetful functors $U\colon \MSL_\f\to \Set$ and $U_\Pfin\colon \Kl_\f(\Pfin) \to \Set$.
\end{thm}

Combined with \Cref{rem:monad-ext} we recover Do\~na Mateo's pushforward presentation of~$\F$:
\begin{notheorembrackets}
  \begin{cor}[{\cite[Thm.\ 3.10]{m25}}]\label{cor:filt-push}
    The filter monad is the pushforward of the finite powerset monad on $\Set_\f$ along the inclusion $\Set_\f \incl \Set$.
  \end{cor}
\end{notheorembrackets}
Similar to the ultrafilter monad, this reasoning carries over from $\Set$ to
suitable topological categories.
For example, the filter monad $\ol{\F}$ on $\Top$ is given by \(\ol \F X = \MSL(\Top(X, \Sierp), 2)\), that is, $\ol{\F}$ is induced by the adjunction in the diagram below.
\begin{equation}\label{eq:codensity-setting-filter-top}
\begin{tikzcd}[column sep=25, row sep=25]
    \Kl_\f(\Pfin)
    \ar{r}{U_\Pfin} &
    \Top
   \ar[out=20,in=-20,loop,looseness=3, "\ol\F",]
    \ar[yshift=0pt, bend right=20,]{d}[swap]{\Top(-,\Sierp)}
    \ar[yshift=0pt, bend left=20, leftarrow, swap]{d}[swap]{\MSL(-,2)}
    \ar[phantom]{d}{\dashv}
    \\
    \Kl_\f(\Pfin)^\op \ar[hook]{r}{I_\Pfin^\op}
    \ar{u}{(-)^\op}[swap]{\rotatebox{-90}{$\simeq$}}
    &
    \MSL^\op
  \end{tikzcd}
\end{equation}
Here $\Sierp$ is the \emph{Sierpinski space} carried by
$2=\{0,1\}$ with open sets $\emptyset$, $\{1\}$, $\{0,1\}$, and the set $\MSL(M,2)$ is
topologized as a subspace of the product space~$\Sierp^{|M|}$. As open subsets of a space \(X\)
correspond to continuous maps from $X$ to $\Sierp$, the monad $\ol{\F}$
sends $X$ to the space~$\ol{\F}X$ of filters of open sets.
Algebras for $\ol{\F}$ carried by $T_0$ spaces correspond to continuous lattices~\cite{day75}.
To capture $\ol{\F}$ as a codensity monad, we modify the codensity setting~\eqref{eq:codensity-setting-filter-kleisli} to the one shown in~\eqref{eq:codensity-setting-filter-top}.
Here we regard the functor $U_\Pfin\colon \Kl_f(\Pfin) \to \Set$ from \eqref{eq:codensity-setting-filter} as a functor to $\Top$ by identifying $U_\Pfin X=\Pfin X$ with the topological space $\Sierp^X$.

From \Cref{thm:codensity} we obtain a codensity presentation of $\ol{\F}$:
\begin{thm}
  The topological filter monad $\ol{\F}$ on $\Top$ is the codensity monad of the functor $U_\Pfin\colon \Kl_\f(\Pfin) \to \Top$ sending $X$ to $\Sierp^X$.
\end{thm}

{ 
\subsection{Ultrafilters on Median Algebras}
\label{sec:meas-medi-algebr}

Another class of ultrafilter monads was considered by Krenz~\cite{k25}.
In \Cref{sec:ultrafilter}, we have considered ultrafilters, which may be seen as finitely additive, $2$-valued probability measures on discrete spaces, and in \Cref{sec:probability-monads} we will consider classical, non-discrete measures.
An intermediate structure is given by {median algebras}~\cite{bk47,i80}, which may be seen as discrete representations of spaces:
A \emph{median algebra} is a set $M$ equipped with a ternary operation $\langle -, -, -\rangle \colon M^3 \rightarrow M$ satisfying
\[\langle x, x, y \rangle = x, \quad  \langle x, y, z \rangle = \langle y, x, z \rangle = \langle x, z, y \rangle \quad  \text{ and } \quad \langle \langle x, w, y \rangle , w, z \rangle = \langle x, z, \langle y, w, z \rangle \rangle.\]
The intuition is that $\langle x, y, z \rangle$ is the nearest point to $x$ between $y$ and $z$.
For example, every total order $L$ is a median algebra with $\langle x, y , z \rangle = y$ for $x \le y \le z$.

On the dual side we have \emph{poc-sets}, which are posets $P$ with a bottom element $\bot$ equipped with a \emph{negation operator} $\neg \colon P \rightarrow P^\op$ satisfying
\(\neg \neg = \id \) and \(  p \le \neg p \text{ implies } p = \bot \).
For example, every Boolean algebra is a poc-set.

It can be shown that for every median algebra $M$, the set $\Med(M, 2)$ is a poc-set, and conversely, for every poc-set $P$ the set $\Poc(P, 2)$ is a median algebra.
The elements of $\Poc(\Med(M, 2), 2)$ may be interpreted as `discrete integration operators'~\cite{k25}.
We thus obtain the following square:
\[
  \begin{tikzcd}[column sep = 40]
    \Med_\f
    \ar[hook]{r}{I} &
    \Med
   \ar[out=20,in=-20,loop,looseness=3, "\F",]
    \ar[yshift=0pt, bend right=20,]{d}[swap]{\Med(-, 2)}
    \ar[yshift=0pt, bend left=20, leftarrow, swap]{d}[swap]{\Poc(-, 2)}
    \ar[phantom]{d}{\dashv}
    \\
    \Poc_\f^\op \ar[hook]{r}{J^\op}
    \ar{u}[swap]{\rotatebox{-90}{$\simeq$}}
    &
    \Poc^\op
   \ar[out=20,in=-10,loop,looseness=3, "\T",]
  \end{tikzcd}
\]
The dual equivalence between finite median algebras and finite poc-sets is due to Isbell~\cite[Theorem 6.12]{i80}, who showed that the dual adjunction between $\Med$ and $\Poc$ restricts to an equivalence, hence both squares commute.
Density of the horizontal inclusions $J\colon \Poc_\f \incl \Poc$ and $I\colon \Med_\f \incl \Med$ is due to Krenz~\cite[Examples 5.4.1 and 5.4.2]{k25}.

From \Cref{thm:codensity} we obtain the following results due to Krenz:
\begin{thmC}[{\cite[Theorems 4.2.1 and 4.2.2]{k25}}]
  The codensity monads of $I$ and $J$ are isomorphic to the monads $ \F = \Poc(\Med(-, 2), 2)$ and $ \T = \Med(\Poc(-, 2), 2)$, respectively.
\end{thmC}

\subsection{Double Dualization Monads}
\label{sec:dd-monads}
Another interesting class of codensity monads that naturally emerge in our framework are certain \emph{double dualization monads}, which are monads given by dual adjunctions as shown in the first diagram below,
where $\C$ is a symmetric monoidal closed category, $D$ is a fixed object of $\C$, and $\C(C,D)\in \C$ denotes the internal hom object of $C,D\in \C$.
\[
\begin{tikzcd}[column sep=35]
\C \ar[bend right=15]{r}[swap]{\C(-,D)} \ar[phantom]{r}{\top} & {\C^{\op}} \ar[bend right=15]{l}[swap]{\C(-,D)}
\end{tikzcd}\;\;
\begin{tikzcd}[column sep=35]
\Set \ar[bend right=15]{r}[swap]{\Set(-,2)} \ar[phantom]{r}{\top} & {\Set^{\op}} \ar[bend right=15]{l}[swap]{\Set(-,2)}
\end{tikzcd}\;\;
\begin{tikzcd}[column sep=35]
\JSL \ar[bend right=15]{r}[swap]{\JSL(-,2)} \ar[phantom]{r}{\top} & {\JSL^{\op}} \ar[bend right=15]{l}[swap]{\JSL(-,2)}
\end{tikzcd}\;\;
\begin{tikzcd}[column sep=35]
\Vect \ar[bend right=15]{r}[swap]{\Vect(-,\K)} \ar[phantom]{r}{\top} & {\Vect^{\op}} \ar[bend right=15]{l}[swap]{\Vect(-,\K)}
\end{tikzcd}
\]
The other three adjunctions above are concrete examples.
Here $\Vect$ is the category of vector spaces over some fixed field $\K$ and
linear maps, and $X^*=\Vect(X,\K)$ is the dual space of a space $X$ with pointwise
structure. The induced monads are the \emph{neighbourhood monad} on $\Set$ (whose
algebras are complete atomic Boolean algebras~\cite[1.5.17--23]{Manes76}), the monad on
$\JSL$ sending a semilattice to its semilattice of ideals~\cite{as21}, and the monad sending a vector space $X$ to its double dual space $X^{**}$ (whose algebras are \emph{linearly compact vector spaces}~\cite[Thm.~7.8]{l13}).

We start with a slight generalization of the double dualization on $\Set$ to the \emph{endomorphism monad}.
Recall that a category $\C$ has ($\Set$-)\emph{powers} if we have for all $A, B \in \C$ and $X \in \Set$ a bijection
\[\Set(X, \C(A, B)) \cong \C(A, X \power B)\]
natural in all three variables.
Every category with products has powers via $X \power B = B^X$.
Given a category $\C$ with powers and an object $D \in \C$, we have the codensity setting below:
\begin{equation}\label{eq:codensity-presentation-double-dual-set}
\begin{tikzcd}[column sep=25, row sep=25]
    1
    \ar{r}{D!} &
    \C
    \ar[yshift=0pt, bend right=20,]{d}[swap]{\C(-,D)}
    \ar[yshift=0pt, bend left=20, leftarrow, swap]{d}[swap]{- \power D}
    \ar[phantom]{d}{\dashv}
    \\
    1^\op \ar[hook]{r}{I^\op}
    \ar{u}{\rotatebox{90}{$\simeq$}}
    &
    \Set^\op
  \end{tikzcd}
\end{equation}
Here $1$ is the terminal category, which is equivalent to the full subcategory of $\Set$ containing only the singleton, and $D!$ is the constant functor at $D$. The inclusion $I\colon 1 \subto \Set$ is dense by \Cref{ex:dense-set}, and the outer square obviously commutes since $1 \power D \cong D$, so we recover from \Cref{thm:codensity} the following folklore result (which is also easy to prove directly):
\begin{thm}
  The endomorphism monad $\C(-, D) \power D$ is the codensity monad of the constant functor $D!$.
\end{thm}
\begin{exa}\label{ex:endo}
  \begin{enumerate}
    \item\label{ex:endo:dd} For $\C=\Set$ and $D = 2$, we obtain the neighbourhood monad $2^{2^{-}}$ as the codensity monad of the constant functor $2!\colon 1 \to \Set$.
    \item\label{ex:endo:ce} For $\C = \BA$ and $D = 2$, we obtain the monad the canonical
      extension monad $2^{\BA(-, 2)}$ from \Cref{thm:canonical-extension} as the codensity monad of the constant functor $2!\colon 1 \to \BA$. 
  \end{enumerate}
\end{exa}
\begin{rem}\label{rem:endo}
  \fbnote{Maybe take this out.}
  Combining \Cref{prop:monad-ext-univ-prop} with the codensity presentation \eqref{eq:codensity-presentation-double-dual-set} we recover a folklore universal property of $\C(-, D) \power D$ (using that $\T$-module structures on a constant functor $D!$ are in bijection with $\T$-algebra structures on $D$): For every monad $\T$ on $\C$, 
  \[
    \text{$\T$-algebra structures on $D$ are in bijection with monad morphisms $\T \rightarrow \C(-, D) \power D$}.
  \]
  As pointed out by Kock~\cite{k70}, double-dualization monads thus are categorical versions of the \emph{endomorphism ring} $\Ab(A, A)$ of abelian group $A$, since for every ring $R$, the $R$-module structures over $A$ correspond to ring homomorphisms $R \rightarrow \Ab(A, A)$.
  Note that this is more than a mere analogy.
  First, codensity monads can be defined in every bicategory.
  Now in the monoidal category of abelian groups, seen as a one-object bicategory, the codensity monad of an object (1-cell) $A$ is precisely given by the internal hom $\Ab(A, A)$.
  Since monads in this bicategory are precisely rings, the correspondence between ring homomorphisms $R \rightarrow \Ab(A, A)$ and $R$-module structures on $A$ is just the universal property of the codensity monad $\Ab(A, A)$.
\end{rem}

\begin{rem}
  Here is a different codensity setting for the double-dualization monad $2^{2^-}$:
\[
    \begin{tikzcd}[column sep=25, row sep=25]
      \BA_\f
      \ar{r}{U} &
      \Set
      \ar[yshift=0pt, bend right=20,]{d}[swap]{\Set(-,2)}
      \ar[yshift=0pt, bend left=20, leftarrow, swap]{d}[swap]{\Set(-,2)}
      \ar[phantom]{d}{\dashv}
      \\
      \Set_\f^\op \ar[hook]{r}{I^\op}
      \ar{u}{\Set(-, 2)}[swap]{\rotatebox{90}{$\simeq$}}
      &
      \Set^\op
    \end{tikzcd}
  \]
  The bottom inclusion is dense, and the outer square commutes.
  This gives a different characterization of the double-dualization monad $2^{2^-}$:
  Since $\BA_\f$ is (isomorphic to) the category $\EM(\mathcal{B})$ of the free Boolean algebra monad $\mathcal B$ on $\Set_\f$, we obtain from \Cref{rem:monad-ext} that the double-dualization monad $2^{2^-}$ is the pushforward $I_{\#}\mathcal{B}$ of the monad $\mathcal{B}$ along the inclusion $I\colon\Set_\f \incl \Set$.
\end{rem}

\begin{rem}[Monotone Neighbourhood Monad]
  There exists a variation on the neighbourhood monad used e.g.~in positive logic~\cite{hk04} called the \emph{monotone neighbourhood monad} $\mathcal{N} \colon \Set \rightarrow \Set$.
  It is the submonad of the neighbourhood monad consisting of upwards closed neighbourhoods:
  \[\mathcal{N} X = \{\mathcal{A} \in \mathcal{P}(\Set(X, 2)) \mid \mathcal{A}(A) = 1, A \subseteq B \Rightarrow \mathcal{A}(B) = 1\}.\]
  Since $\mathcal{N}X = \Pos(\Set(X, 2), 2)$, we obtain the following codensity presentation of $\mathcal{N}$: 
  \[
    \begin{tikzcd}[column sep=40, row sep=25]
      \DL_\f
      \ar{r}{U} &
      \Set
      \ar[yshift=0pt, bend right=20,]{d}[swap]{\Set(-,2)}
      \ar[yshift=0pt, bend left=20, leftarrow, swap]{d}[swap]{\Pos(-,\{0 \le 1\})}
      \ar[phantom]{d}{\dashv}
      \\
      \Pos_\f^\op \ar[hook]{r}{J^\op}
      \ar{u}{\Pos(-, \{0 \le 1\})}[swap]{\rotatebox{90}{$\simeq$}}
      &
      \Pos^\op
    \end{tikzcd}
  \]
  Here $\DL_\f$ is the full subcategory of finite distributive lattices $U \colon \DL_\f \rightarrow \Set$ is the forgetful functor.
  Finite posets are dense in $\Pos$, since they are precisely the finitely presentable posets.
  The duality between finite posets and finite distributive lattices is \emph{Birkhoff duality}~\cite{b37}.
  By \Cref{thm:codensity}, we obtain that $U$ is a codensity presentation for the monad $\mathcal{N}$.

  Moreover, since~$\DL_\f$ is (isomorphic to) the category $\EM(\mathcal D)$ of the free distributive lattice monad $\mathcal D$ on $\Set_\f$, we obtain from \Cref{rem:monad-ext} that the monotone neighbourhood monad $\N$ is the pushforward $I_{\#}\mathcal{D}$ of the monad $\mathcal{D}$ restricted to $\Set_\f$ along the inclusion $I\colon \Set_\f \incl \Set$.
\end{rem}

Next we consider modules.
Recall that a \emph{semiring} $S$ consists of a commutative monoid $(S, + ,0)$ and a monoid $(S, \cdot, 1)$ such that the usual distributive laws hold: 
\[
  r\cdot (s+t) = r\cdot s + r\cdot t,
  \qquad
  (r+s) \cdot t = r\cdot t + s\cdot t,
  \qquad
  r\cdot 0 = 0 = 0\cdot r \qquad \text{for all $r,s,t \in S$}.
\]
An \emph{$S$-module} is a commutative monoid $(A, +, 0)$ together with an 
action $S \times A \xra{\cdot} A$ of the multiplicative monoid of $S$ denoted by juxtaposition $rm$ for $r \in S$ and $m \in A$ such that for every $r,s\in S$ and $m,n\in A$, the following laws hold:
\[
  \begin{array}{r@{\,}l@{\qquad}r@{\,}l}
    (r+s)m & = rm + sm, & r(m+n) & = rm + rn, \\
    0m & =0, & r0 &= 0, \\
    1m &= m, & r(sm) &= (r\cdot  s)m.
  \end{array}
\]
Equivalently, the action is represented by a semiring morphism $S \rightarrow \CMon(A, A)$ into the endomorphism ring with composition as multiplication; here $\CMon$ denotes the category of all commutative monoids and their morphisms.
A morphism $h\colon A \to A'$ of $S$-modules is a monoid morphism which preserves the action:
\[
  h(m+n) = h(m) + h(n)\qquad\text{and} \qquad h(rm) = r \cdot h(m), \qquad
  \text{for all $m, n\in A$ and $r \in S$.}
\]
All $S$-modules and their morphisms form a category $\Mod_S$, which is (isomorphic to) the
category of algebras $\EM(\mathcal{S})$ for the \emph{free semimodule monad} $\mathcal{S}$
on $\Set$ given by
\[
  \mathcal{S}X = \{f \colon X \rightarrow \mathcal{S} \mid f(x) = 0 \text{ for all but finitely many $x$}\}.
\]
Note that for every finite set $n$, we have $\mathcal{S}n \cong S^n$.
\begin{exa}
  \begin{enumerate}
    \item The natural numbers $\N$ and the integers $\Z$ form semirings, and $\Mod_\N = \CMon $, and $\Mod_\Z = \Ab$.
    \item The two-element lattice $\mathbf{2} = (\{0 < 1\}, \lor, \land, 0, 1)$ is a semiring, and its modules are semilattices $\Mod_\mathbf{2} \cong \JSL$: Every $\mathbf{2}$-module $A$ satisfies $a = 1a = (1 \lor 1)a = 1a + 1a = a + a$, so it is a commutative idempotent monoid.
    \item For every field $\K$, its modules are vector spaces over $\K$, that is,  $\Mod_\K = \Vect_\K$.
  \end{enumerate}
\end{exa}
The subcategory $\Kl_\f(\mathcal{S}) \incl \Kl(\mathcal{S})$ on sets is equivalent to the following category $\Mat_{S}$ of matrices:
the objects of $\Mat_{S}$ are natural numbers $n \in \N$, and morphisms $n \rightarrow m$ are $m \times n$-matrices with entries in $S$.
The equivalence identifies an $m \times n$-matrix  $M \in \Mat_S(n, m)$ with the morphism $n \rightarrow \mathcal{S}m$  corresponding to matrix-vector multiplication
\[I(M) \colon S^n \rightarrow S^m, \qquad v \mapsto Mv.\]
Note that for every $S$-module $A$, every matrix $M  \in \Mat_S(n,m)$ also induces an $S$-module morphism
\begin{equation}\label{eq:mat-precomp}
  J_A(M) \colon A^n \rightarrow A^m, \qquad v \mapsto Mv,
\end{equation}
which makes $J_A$ a functor $\Mat_S \rightarrow \Mod_S$.
Taking for $A$ the semiring $S$ as $S$-module over itself, we have $J_S = I$.

The category $\Mat_S$ is self-dual, with inverse given by transposition $(-)^\tp \colon \Mat_S \simeq \Mat_S^\op$.
We therefore have for every $S$-module $A$ the following codensity setting: 
\begin{equation}\label{eq:codensity-setting-smod}
\begin{tikzcd}[column sep = 40]
    \Mat_S
    \ar{r}{J_A} &
    \Mod_S
    \ar[yshift=0pt, bend right=20,]{d}[swap]{\Mod_S(-,A)}
    \ar[yshift=0pt, bend left=20, leftarrow, swap]{d}[swap]{\Mod_S(-,A)}
    \ar[phantom]{d}{\dashv}
    \\
    \Mat_S^\op \ar[hook]{r}{I^\op}
    \ar{u}{\rotatebox{90}{$\simeq$}}
    \ar[swap]{u}{(-)^\tp}
    &
    \Mod_S^\op
  \end{tikzcd}
\end{equation}
The inclusion $I \colon \Mat_S \simeq \Kl_\f(\mathcal{S}) \incl \Mod_S$ is dense by \iref{ex:dense-alg}{kleisli-finitary}.
We show that the outer square commutes. On objects $\Mod_S(I(n), A) \cong \Mod_S(S^n, A) \cong A^n = J_A(n^\tp)$, since $S^n = \mathcal{S}n$ is free on~$n$. On morphisms, given an $m \times n$-matrix $M \in \Mat_S(n,m) = \Mat_S^\op(m,n)$ and $h \in \Mod_S(S^m, A)$, we have that $\Mod_S(I(M),A)(h) = h \cdot I(M) = J_A(M^\tp)(h)$: $h\colon S^m \to A$ corresponds to the (row) vector $a \in A^m$, so $h \cdot I(M)\colon S^n \to  A$ corresponds to the (column) vector $aM = M^\tp a^\tp$, which is the same vector that $J_A(M^\tp)(h)$ corresponds to.
Thus, from \Cref{thm:codensity} we obtain:

\begin{thm}\label{thm:s-mod-dd}
  For every commutative semiring $S$ and $S$-module $A$ the double-dualization monad $\Mod_S(\Mod_S(-, A), A)$ is the codensity monad of $J$.
\end{thm}


For $S= \K$ and $A = S = \K$ we recover, since $\Mat_\K\simeq \Vect_\fd$ (the full subcategory of $\Vect$ given by finite dimensional spaces), Leinster's characterization of the double-dualization monad on vector spaces:
\begin{equation*}
  \begin{tikzcd}
    \Mat_\K \rar[phantom]{\simeq}
    &
    \Vect_\fd
    \ar[hook]{r}{J_\K = I} &
    \Vect
    \ar[yshift=0pt, bend right=20,]{d}[swap]{\Vect(-,\K)}
    \ar[yshift=0pt, bend left=20, leftarrow, swap]{d}[swap]{\Vect(-,\K)}
    \ar[phantom]{d}{\dashv}
    \\
    \Mat_\K^\op \rar[phantom]{\simeq}
    \ar{u}{\rotatebox{90}{$\simeq$}}[swap]{(-)^\tp}
    &
    \Vect_\fd^\op \ar[hook]{r}{I^\op}
    \ar{u}{\rotatebox{90}{$\simeq$}}
    &
    \Vect^\op
  \end{tikzcd}
\end{equation*}
\begin{thmC}[\cite{l13}]
  The codensity monad of the inclusion $\Vect_\fd \incl \Vect$ is the double-dualization monad $\Vect(\Vect(-, \K), \K)$ on vector spaces.
\end{thmC}

Similarly, for $A=S=\mathbf{2}$ we obtain the double-dualization monad on $\MSL$.  In the diagram below, note that the self-duality $\MSL_\f \simeq^\op \MSL_\f$ indeed restricts to $(-)^\tp$ on matrices, as it preserves free objects.
\begin{equation*}
\begin{tikzcd}
    \Mat_\mathbf{2} \rar[tail]{} \ar[shiftarr={yshift=20pt}]{rr}{J_{\mathbf{2}} \cong I}
    &
    \MSL_\f
    \ar[hook]{r}{} &
    \MSL
    \ar[yshift=0pt, bend right=20,]{d}[swap]{\MSL(-,2)}
    \ar[yshift=0pt, bend left=20, leftarrow, swap]{d}[swap]{\MSL(-,2)}
    \ar[phantom]{d}{\dashv}
    \\
    \Mat_\mathbf{2}^\op \rar[tail]{}
    \ar{u}{\rotatebox{90}{$\simeq$}}[swap]{(-)^\tp}
    &
    \MSL_\f^\op \ar[hook]{r}{}
    \ar{u}{\rotatebox{90}{$\simeq$}}
    &
    \MSL^\op
  \end{tikzcd}
\end{equation*}
From \Cref{thm:codensity} we obtain the following result due to Ad\'amek and Sousa~\cite{as21}:

\begin{thmC}
  The double-filter monad $\MSL(\MSL(-, 2), 2)$ on $\MSL$ is the codensity monad of $I \colon \Mat_{\mathbf{2}} \rightarrow \MSL$.
\end{thmC}

Finally, for $S = \Z$ we have $\Mod_S = \Ab$, and taking for $A = S^1$ the unit circle we get a codensity presentation of the double dualization monad $\mathcal{B} = \Ab(\Ab(-, S^1), S^1)$ on abelian groups, which is known as the \emph{Bohr compactification} (see e.g.~\cite{ak43}) of (discrete) abelian groups.
\begin{cor}\label{cor:bohr}
  The monad $\mathcal{B}$ on $\Ab$  is the codensity monad of $J_{S^1} \colon \Mat_\Z \rightarrow \Ab$.
\end{cor}


\begin{rem}\label{rem:V-dd}
  Kock~\cite{k70} considered double-dualization monads $\V(\C(-, D), D)$ over symmetric monoidal closed categories $\V$ as \emph{$\V$-enriched monads}  over $\V$, where they have much simpler presentations analogous to the double-dualization monad on $\Set$.
\end{rem}

  \takeout{

For a slightly more intricate example, we consider the category $\Ab$ of abelian groups and as the dualizing object the circle group $S^1=\{ x\in \mathbb{C} \mid |x|=1 \}$, a subgroup of the multiplicative group of non-zero complex numbers. It forms
a (compact Hausdorff) topological group w.r.t.~the Euclidian topology on
$\mathbb{C}$. A suitable codensity setting is given by:
\begin{equation*}\label{eq:codensity-setting-abelian}
\begin{tikzcd}[column sep=25, row sep=25]
    \{ S^1\times S^1 \}
    \ar{r}{U} &
    \Ab
    \ar[yshift=0pt, bend right=20,]{d}[swap]{\Ab(-,S^1)}
    \ar[yshift=0pt, bend left=20, leftarrow, swap]{d}[swap]{\Ab(-,S^1)}
    \ar[phantom]{d}{\dashv}
    \\
    \{\mathbb{Z}\oplus \mathbb{Z}\}^\op \ar[hook]{r}{J^\op}
    \ar{u}{\rotatebox{90}{$\simeq$}}
    &
    \Ab^\op
  \end{tikzcd}
\end{equation*}
While all previous examples are based on simple dualities, we now use
a corollary of a more advanced duality result: (discrete) \emph{Pontryagin
  duality}~\cite{pontryagin46,morris77}, the dual equivalence $\Ab^\op\simeq \CHAb$ between
$\Ab$ and the category $\CHAb$ of compact Hausdorff abelian groups and
continuous group morphisms. The equivalence sends an abelian group $A$
to the
topological group $\Ab(A,S^1)$ of all group morphisms into $S^1$, topologized as a subspace
of $(S^1)^{|A|}$. Under this duality, the additive group $\mathbb{Z} \oplus \mathbb{Z}$ (the
free abelian group on two generators) corresponds to the torus $S^1\times S^1$, so Pontryagin duality restricts to a
dual equivalence between the one-object full subcategories
$\{\mathbb{Z} \oplus \mathbb{Z}\}\subto \Ab$ and $\{S^1\times S^1\}\subto \CHAb$. The functor
$U$
is the composition of the inclusion $\{S^1\times S^1\}\subto \CHAb$ with the forgetful
functor to $\Ab$. The inclusion $J\colon \{\mathbb{Z} \oplus \mathbb{Z}\} \subto\Ab$
is dense (\iref{ex:dense-alg}{kleisli-finitary}), and the outside of the diagram
commutes, since Pontryagin duality is given by the functor $\Ab(-,S^1)\colon \Ab^\op \to \CHAb$.
From \Cref{thm:codensity} we obtain the following result:

}

\subsection{Isbell Duality}\label{sec:isbell-duality}
We continue with an example of a more abstract nature: for every small category~$\C$, the codensity monad of the Yoneda embedding $\C \subto [\C^\op,\Set]$ is the monad given by \emph{Isbell duality}~\cite{isbell66}. This result due to Kock~\cite{kock66} originally motivated the introduction of codensity monads, and many codensity settings may be seen as restrictions of Isbell duality, as has been observed by Di Liberti~\cite{dl20}.
Recall that Isbell duality is the dual adjunction in \eqref{eq:codensity-setting-isbell} between the categories of contravariant and covariant presheaves on $\C$ defined by
\[ \mathcal{O}(X) = (\, C\mapsto [\C^\op,\Set](X,\C(-,C)) \,)\quad\;\;\text{and}\;\;\quad \mathsf{Spec}(A) = (\, C\mapsto [\C,\Set](A,\C(C,-)) \,).  \]
The two Yoneda embeddings $y$ and $\tilde{y}$ give rise to the codensity setting shown below:\\[-0.2cm]
\begin{minipage}[c]{.5\textwidth}
  \begin{align*}
    y&\colon \C\subto [\C^\op,\Set],&  C&\mapsto \C(-,C),\\
    \tilde{y}& \colon \C^\op \subto [\C,\Set], & C&\mapsto \C(C,-).
  \end{align*}
\end{minipage}
\begin{minipage}{.1\textwidth}
\end{minipage}
\begin{minipage}[c]{.4\textwidth}
\begin{equation}\label{eq:codensity-setting-isbell}
\begin{tikzcd}[column sep=25, row sep=25]
    \C
    \ar[hook]{r}{y} &
    {[}\C^\op,\Set{]}
    \ar[yshift=0pt, bend right=20]{d}[swap,pos=.52]{\mathcal{O}}
    \ar[yshift=0pt, bend left=20, leftarrow, swap]{d}[swap,pos=.53]{\mathsf{Spec}}
    \ar[phantom]{d}{\dashv}
    \\
    (\C^\op)^\op \ar[hook]{r}{(\tilde{y})^\op}
    \ar[equals]{u}
    &
    {[}\C,\Set{]}^\op
  \end{tikzcd}
\end{equation}
\end{minipage}\\[2ex]
Indeed, the embedding $\tilde{y}$ is dense (\Cref{ex:dense-yoneda}), and the outside of the diagram commutes up to isomorphism by the Yoneda lemma: for all objects $C,D\in \C$, we have
\[ \mathsf{Spec}(\tilde{y}(D))(C) = [\C,\Set](\C(D,-),\C(C,-))\cong \C(C,D) = y(D)(C). \]
From \Cref{thm:codensity} we obtain:

\begin{thmC}[\cite{kock66,dl20}]
  For every small category $\C$, the monad on~$[\C^\op,\Set]$ induced by Isbell duality is
  the codensity monad of $y\colon \C\subto [\C^\op,\Set]$.
\end{thmC}

{ 
\subsection{The Ultraproduct Monad}
\label{sec:ultraproducts}
As the final example of the present section we show how to recover the codensity presentation of the ultraproduct monad due to Leinster~\cite{l13}.

The \emph{ultraproduct} $\int_I P \diff u $ of a family $P = (P_i)_{i \in I}$ of sets with respect to an ultrafilter $u \in \U I$ on $I$ is defined as the directed colimit indexed by the poset $u^\op$ (i.e.\ $u$ ordered by reverse inclusion):
\begin{equation}
  \label{eq:ultraproduct}
  \int_I P \diff u = \colim_{K \in u^\op}\prod_{i \in K}P_i.
\end{equation}
If all $P_i$ are non-empty, then the ultraproduct is simply the quotient
\[
  \int_I P \diff u = \big(\prod_{i \in I} P_i\big) /  \mathord{\sim} \qquad \text{with $(a_i)_i \sim (b_i)_i$  if $\{i \in I \mid a_i = b_i\} \in u$}.
\]
If $u$ is understood as a discrete measure telling us which subsets $K$ of $I$ are negligible (those not in $u$), then this definition states that two families are considered equal if they are `essentially equal'.
The ultraproduct construction is an important tool in model theory.

The ultraproduct construction gives rise to a monad~\cite{e74,k81}:
Its domain is the  category $\Fam(\Set^\op)$ whose objects are indexed families $(I, P)$ of sets,  consisting of an index set $I$ and a functor $P \colon I \rightarrow \Set$.
A morphism $(I, P) \rightarrow (J, Q)$ consists of a map $f \colon I \rightarrow J$ and a natural transformation $\varphi \colon Q \cdot f \rightarrow P$ as in \eqref{eq:fam-mor}, that is, a family of maps $\varphi_i \colon Q_{f(i)} \rightarrow P_i$.
\begin{equation}
    \label{eq:fam-mor}
    \begin{tikzcd}[column sep=small]
      I \ar{rr}{f} \drar[swap]{P} \ar{rr}[yshift=-20pt, phantom]{\overset{\varphi}{\Leftarrow}} &  & J \dlar{Q} \\
      & \Set &
    \end{tikzcd}
  \end{equation}
We denote its full subcategory of finitely-indexed families by $\Fam_\f(\Set^\op)$.
The categories $\Fam(\Set^\op)$ and $\Fam_\f(\Set^\op)$ are the free (finite-)coproduct completions of $\Set^\op$, respectively.
The \emph{ultraproduct monad} $\UP$ on $\Fam(\Set)$ sends a family $P$ over $I$ to the family $(\int_I P \diff u)_{u \in \U I}$.
Its algebras are sheaves over compact Hausdorff spaces~\cite{k81}.

A \emph{Boolean sheaf} consists of a Boolean algebra $B$ and a \emph{sheaf} on $B$, which is a functor $F \colon B^\op \rightarrow \Set$ satisfying the \emph{sheaf condition}:
writing $F(a \le b)(x)$ as $x|_a$ for every $x \in Fb$, it states that whenever $b = c \lor d$ and $x \in Fc, y \in Fd$ are \emph{compatible} in the sense that $x|_{c \land d} = y|_{c \land d}$, then there exists a unique $z \in Fb$ with $z|_c = x, z|_d = y$.
[Remark: a Boolean sheaf is \emph{not} a sheaf over the site $\BA$ of all Boolean algebras, but over the site $B$.]
A morphism of Boolean sheaves $(B, F) \rightarrow (C, G)$ consists of a homomorphism $h \colon B \rightarrow C$ of Boolean algebras and a natural transformation $\gamma \colon F \rightarrow G \cdot h^\op $.
    \begin{equation}
    \label{eq:sh-mor}
    \begin{tikzcd}[column sep=small]
      B^\op \ar{rr}{h^\op} \drar[swap]{F} \ar{rr}[yshift=-20pt, phantom]{\overset{\gamma}{\Rightarrow}} &  & C^\op \dlar{G} \\
      & \Set &
    \end{tikzcd}
  \end{equation}
We denote the category of Boolean sheaves by $\BSh$ with the full subcategory $\BSh_\f$ of sheaves over finite Boolean algebras.

The adjunction inducing the ultraproduct monad is suggested by \eqref{eq:ultraproduct}:
The left adjoint $\Sh \colon \Fam(\Set^\op) \rightarrow \BSh^\op$ sends a family $(I, P)$
to its discrete sheaf:
\begin{equation}
  \label{eq:sh}
  \Sh(I, P) = (2^I, K \mapsto \prod_{k \in K}P_k).
\end{equation}
The right adjoint $\Stalk\colon \BSh^\op\to \Fam(\Set^\op)$ sends a Boolean sheaf $(B, F)$ to the family of its stalks:
\begin{equation}
  \label{eq:stalk}
  \Stalk(B, F) = (\mathcal{U} B, u \mapsto \colim_{b \in u^\op}Fb).
\end{equation}
Somewhat remarkably, we can describe this adjunction in elementary categorical terms:
Recall that the \emph{weighted limit} $ \langle W, F \rangle$ of a functor $F \colon I \rightarrow \D$ with weight $W \colon I \rightarrow \Set$ is given by $(\Ran_{\tilde{y}^\op} F)W$, and the \emph{weighted colimit} $V \star F$ of $F$ by a weight $V \colon \I^\op \rightarrow \Set$ by $(\Lan_y F)V$, where $y$ and $\tilde{y}$ are the Yoneda embeddings from~\eqref{eq:codensity-setting-isbell}.

\begin{prop}\label{prop:adj-sh}
  There is an adjunction $\Sh \dashv \Stalk$ between $\Fam(\Set^\op)$ and $\BSh$ induced by taking weighted limits and colimits: it is given by
  \begin{align*}
    \Sh \colon \Fam(\Set^\op) &\rightarrow \BSh & \Stalk \colon \BSh &\rightarrow \Fam(\Set^\op) \\
    (I, P) &\mapsto (2^I,  \langle -, P \rangle \colon (2^I)^\op \rightarrow \Set) & (B, F) &\mapsto (\mathcal{U} B, (-) \star F \colon \mathcal{U} B \rightarrow \Set), \\
  \end{align*}
  where we identify subsets $K \subseteq 2^I$ with $\{\emptyset, 1\}$-valued functors $I \rightarrow \Set$ and ultrafilters $u \in \mathcal{U} B$ with $\{\emptyset, 1\}$-valued functors $B \rightarrow \Set$.
\end{prop}
\begin{proof}
  The formulas \eqref{eq:sh} and \eqref{eq:stalk} are simply the weighted (co)limits spelled out.
  We consider $I$ as a discrete category and the weights are discrete in the sense that they only take values in the subcategory $\{\emptyset, 1\}$ of $\Set$. Then for every $K \in 2^I$, we have
  \[\Sh(I, P)K = \langle K, P \rangle \cong \int_{i \in I} K(i) \power P_i \cong \prod_{i \in K}P_i, \]
  and for every $u \in \U B$,
  \[ \Stalk(B, F)u = u \star F \cong \int^{b \in B^\op} u(b) \cdot F(b) \cong \colim_{b \in u^\op}Fb.\]
  The adjunction isomorphism \[\Fam(\Set^\op)((I, P), \Stalk(B, F)) \cong \BSh((B, F), \Sh(I, P))\]
  is on the index-level the familiar adjunction $2^{(-)} \dashv \BA(-, 2)$ between sets and Boolean algebras.
  Given $f \in \Set(I, \BA(B, 2))$ with adjoint transpose $\flip f \in \BA(B, 2^I)$ and uncurried form $\uncurry f \colon I \times B \rightarrow 2$ this induces a bijection between natural transformations
  \[[I, \Set](f \star F, P) \cong [I \times B, \Set](\uncurry f, \Set(F-, P-)) \cong [B^\op, \Set](F, \langle \flip f, P \rangle ). \qedhere\]
\end{proof}

\takeout{
\fbnote{I guess we should leave the fibration stuff out}
It may be interesting that this adjunction can also be understood from a fibrational perspective, which may bridge the gap to categorical logic, e.g.~as in \cite{bsss21}.
\begin{prop}
  The functors \[\Fam(\Set^\op) \rightarrow \Set \quad (I, P) \mapsto I, \qquad \text{ and } \qquad \BSh \rightarrow \BA, \quad (B, F) \mapsto B\] are bifibrations.
\end{prop}
\begin{proof}
  We start with $\Fam(\Set^\op)$. The reindexing of a family $(J, Q)$ along $f \colon I \rightarrow J$ is given by pullback, that is $(f^* Q)_i = Q_{f(i)}$ and the cartesian morphism over $f$ is pointwise identity.
  The pushforward of $(I, P)$ along $f$ is given by dependent product: $(f_! P)_j = \prod_{i \in f^{-1}(j)}P_i$, and the cocartesian morphism over $f$ is given by projection.

  As for $\BSh$, the reindexing of a sheaf $(C, G)$ along $h \colon B \rightarrow C$ is often called the \emph{direct image sheaf}  with $(h^*G) b = Ghb$.
  The pushforward of $(B, F)$ along $h$ is the \emph{inverse image sheaf} $(h_! G) c = \colim_{hb \ge c} Fb$, where the directed colimit is taken over the reverse order.
  The cocartesian morphism is given by the obvious coprojection.
  It remains to verify that $h_! G$ is indeed a sheaf. The restriction action \[(h_! G) c' \rightarrow (h_! G) c \qquad \text{ for } c \le c'\] is simply induced by the coprojections.
  We verify the sheaf condition for binary joins: let $c = c_1 \lor c_2$, and let $y_i \in (h_! F)c_i$ for $i = 1, 2$ be a compatible family.
  Then $y_i = \kappa_i(x_i) \in Fb_i$ for some $b_i$ with $hb_i \ge c_i$ and $x_i \in F b_i$.
  By assumption the $y_i$ are compatible, that is, $y_i|_{c_1 \land c_2} \in (h_! G)(c_1 \land c_2)$.
  But the colimit of $(h_! G)(c_i \land c_2)$ is directed, so the $y_i = \kappa_i(x_i)$ are joined in some $Fb'$ with $h(b') \ge c_1 \land c_2$.
  This means that the $x_i$ are compatible, so there exists, since $F$ is a sheaf, a unique amalgation $x \in F(b_1 \lor b_2)$.
  Then $h(b_1 \lor b_2) \ge c$, and the element $\kappa_{b_1 \lor b_2}(x) \in (h_!G)c$ is an amalgation.
  \fbnote{Maybe draw diagram}
  Uniqueness is shown by a similar argument.
\end{proof}
Under this perspective, the adjunction from \Cref{prop:adj-sh} is an extension of the adjunction between sets and Boolean algebras to fibers via pushforward:
A family $(I, P)$ is sent to $(2^I, K \mapsto (K_! P)1)$ and a sheaf $(B, F)$ to $(\mathcal{U} B, u \mapsto (u_! F)1)$.
In particular, we get commuting squares: \fbnote{TODO I think this is an adjunction of bifibrations but I didn't check it.}
\[
  \begin{tikzcd}[sep=large]
    \mathllap{\Fam(}\Set^\op) \dar{} \rar[bend left=15]{\Sh} \rar[phantom]{\rotatebox{-90}{$\dashv$}}  & \BSh^\op \dar{} \lar[ bend left=15 ]{\Stalk} \\
    \Set \rar[bend left=15]{2^{-}} \rar[phantom]{\rotatebox{-90}{$\dashv$}} & \BA \lar[bend left=15]{\BA(-, 2)}
  \end{tikzcd}
\]
}

The codensity situation for the ultraproduct monad is the following:
\begin{equation}
  \label{eq:codensity-setting-ultraproduct}
  \begin{tikzcd}
    \Fam_\f(\Set^\op)
    \rar[hook]{I} &
    \Fam(\Set^\op)
    \ar[out=20,in=-20,loop,looseness=3, "\UP",]
    \ar[yshift=0pt, bend right=20,]{d}[swap]{\Sh}
    \ar[yshift=0pt, bend left=20, leftarrow, swap]{d}[swap]{\Stalk}
    \ar[phantom]{d}{\dashv}
    \\
    \BSh_\f^\op \ar[hook]{r}{J^\op}
    \ar{u}{|_{\At}}[swap]{\rotatebox{-90}{$\simeq$}}
    &
    \BSh^\op
  \end{tikzcd}
\end{equation}
Here $I$ and $J$ are the inclusions, and the functor $(-)|_{\At}$ is the restriction of the functor $\Stalk$ to sheaves over finite Boolean algebras that acts by restriction to atoms:
Given $(B, F)$ where $B \cong 2^A$ for a finite set of atoms $A$, we have that every ultrafilter is principal, that is, $\eta_A \colon A \cong \mathcal{U}2^A$ is an isomorphism.
Since $\{a\}$ is the top element of $\eta(a)^\op$ we have $\colim_{b \in \eta(a)^\op} Fb \cong F\{a\}$ and thus $\Stalk(B, F) \cong (A, F|_\At)$.
Second, we prove density of~$J$:
\takeout{
\begin{lem}\label{lem:sh-incl-dense}
  The inclusion $J \colon \BSh_\f \incl \BSh $ is dense.
\end{lem}
\begin{proof}
  We give a direct proof.
  Let $(B, F)$ be a Boolean sheaf with canonical cone $(f_{i}, \varphi_{i}) \colon (B_{i}, F_{i}) \rightarrow (B, F)$ over all Boolean sheaves with $B_{i}$ finite.

  (1) Since the inclusion $\BA_\f \incl \BA$ is dense the family $f_{i}$ is jointly surjective.
  We show the same for the family $\varphi_{i}$, that is,
  given $x \in Fb$ there exists a finite Boolean subalgebra $\iota \colon A \incl B$ with $b \in B$.
  We get the obvious inclusion
  \[(\iota \colon A \incl B, \id \colon F\iota^\op \Rightarrow F \iota^\op) \colon (A, F\iota^\op) \incl (B, F)\]
  such that $x$ is contained in the image of the fibre map at $b$.

  (2) To show that the canonical cocone $(f_i, \varphi_i) \colon (B_i, F_i) \rightarrow (B, F)$ is indeed the colimit of the canonical diagram let $(g_{i}, \gamma_{i}) \colon (B_{i}, F_{i}) \rightarrow  (C, G)$ be a cocone.
  Since $B \cong \colim_{i} B_{i}$, there  exists a unique Boolean algebra morphism $g \colon B \rightarrow C$ with $g \cdot f_{i} = g_{i}$.
  We define the natural transformation
  \(\gamma_b \colon Fb \rightarrow Ggb\) such that $x \in Fb$ is mapped to $\gamma_{i,b}(x)$,
  where $ f_{i} \colon B_{i} \incl B$ is an inclusion of a finite Boolean subalgebra $B_i$ containing $b$ as in point (1).
  Then $\gamma_b$ is well-defined, since if $b$ is contained in another $f_j \colon B_j \incl B$, then there exists a finite subalgebra $f_k \colon B_k \incl B$ containing both $B_i$ and $B_j$.
  Since th connecting morphisms $(B_i, F \cdot f_i^\op) \rightarrow (B_k, F \cdot f_k^\op)$ induced by the inclusion $B_i \subseteq B_k$ are fibrewise the identity as in (1) we trivially have \(\gamma_{i, b}(x) = \gamma_{k, b}(x) = \gamma_{j, b}(x)\).

  By point (1) uniqueness also follows, since this definition is forced upon us.
\end{proof}
}
%
\begin{lem}\label{lem:bsh-dense}
The inclusion $\BSh_\f \incl \BSh$ is~dense.
\end{lem}

\begin{proof}
  Let $(\beta_i, \varphi_i)\colon (B_i, F_i) \rightarrow (B, F)$ ($i \in I$) be the canonical cocone of all morphisms of Boolean sheaves from those in $\BSh_\f$ to $(B, F)$.
  We prove that this is a colimit.
  Suppose $(\gamma_i, \psi_i) \colon (B_i, F_i) \rightarrow (C, G)$ is any cocone.

  First note that for every Boolean algebra homomorphism $h \colon B_i \rightarrow B$ there exists a $B_i$-indexed sheaf $F \cdot h^\op$ and a sheaf morphism $(h, \id) \colon (B_i, F \cdot h^\op) \rightarrow (B, F)$.
  Thus, the $\beta_i\colon B_i \to B$ ($i\in I$) form the canonical cocone of all Boolean algebra homomorphisms from finite Boolean algebras to $B$. So by the density of the inclusion $\BA_\f \subto \BA$ (\Cref{ex:dense-alg:fp}), that cocone forms a colimit. Since the $\gamma_i \colon B_i \rightarrow C$ ($i \in I$) form a cocone, we obtain a unique $\gamma \colon B \rightarrow C$ such that $\gamma \cdot \beta_i = \gamma_i$ holds for all $i \in I$.

  We now prove the existence of a natural transformation $\psi \colon F \rightarrow G \cdot\gamma^\op$. Note that for every $b \in B$, there is some $\beta_i \colon B_i \rightarrow B$ and $b_i \in B_i$ such that $\beta_i(b_i) = b$.
  We define the $b$-component $\psi(b)\colon Fb \to (G \cdot \gamma^\op)b = G(\gamma(b))$ via the cocone morphism $\psi_i\colon (B_i, F\cdot \beta_i^\op) \rightarrow (C, G)$ from above as
  \begin{equation}
    \label{eq:con-mor}
    \begin{tikzcd}
      Fb \rar{\psi(b)} \dar[equals]{\id(b)} & G (\gamma b) \dar[equals]{} \\
      (F \cdot \beta_i^\op) b_i \rar{\psi_i(b_i) } & G (\gamma_i b_i).
    \end{tikzcd}
  \end{equation}
  We show that the above definition of $\psi(b)$ this is independent of the choice of $i$ and $b_i$. Note that the colimit $b_i\colon B_i \to B$ ($i \in I$) is directed. So if $b = \beta_j(b_j)$ for some other $b_j \in B_j$, then there exists some $k \ge j, i$ in $I$ such that $\beta_{i,k}(b_i) = b_k = \beta_{j,k}(b_j)$ and $\beta_k(b_k) = b$.
  Since $(\beta_{i,k}, \id) \colon (B_i, F \beta_i^\op) \rightarrow (B_k, F \beta_k^\op)$ is a connecting morphism of our given diagram (and similarly for $j$) and the $\psi_i$ form a cocone, we obtain $Fb = (F \beta_i^\op)b_i =  (F \beta_k^\op)b_k = (F \beta_j^\op)b_j$, $G(\gamma b) = G(\gamma_i b_i) = G(\gamma_j b_j) = G(\gamma_k b_k)$, and
  \[
    \psi_i(b_i) = \psi_k(b_k) = \psi_j(b_j) \colon Fb  \rightarrow G (\gamma b) = (G \cdot \gamma^\op)b.
  \]
  We verify that $\psi$ is natural, so let $b \ge b'$.
  Note that the canonical cocone for $B$ is \emph{filtered}:
  this implies that we may assume that $b$ and $b'$ are in the image of the same $\beta_i$, that is, there exists $i \in I$ such that $b = \beta_i(b_i)$ and $b' = \beta_i(b_i')$.
  Futhermore, by taking image factorizations we may assume that $\beta_i$ is an embedding, so $b_i \ge b_i'$.
  We then obtain the following prism proving naturality:
  \[
    \begin{tikzcd}[column sep=small]
      & Fb \ar{rr}{\psi(b)}  \ar{ddl}{} \dar[equals]{}& & G \gamma b \ar{ddl}{} \\
      & (F \beta_i) b_i \ar{urr}{\psi_i(b_i)} \ar{ddl}{} & & \\
      Fb' \ar[crossing over]{rr}{\psi(b')} \dar[equals]{}& & G \gamma b' \\
      (F \beta_i) b_i' \ar[swap]{rru}{\psi_i(b_i')} & & &
    \end{tikzcd}
  \]

  The equation $\psi \cdot \varphi_i = \psi_i$ is satisfied at every $i \in I$:
  Every morphism $(\beta_j, \varphi_j) \colon (B_j, F_j) \rightarrow (B, F)$
  factorizes as
  \[
    (B_j, F_j) \xlongrightarrow{(\id, \varphi_j)} (B_j, F \cdot \beta_j^\op) \xlongrightarrow{(\beta_j, \id)} (B, F).
  \]
  Hence we have a diagram
  \begin{equation*}
    \begin{tikzcd}
       &  F (\beta_j(c_j)) \rar{\psi(\beta_j(c_j))} \dar[equals]{\id(\beta_j(c_j))} & G (\gamma \beta_jc_j) \dar[equals]{} \\
       F_jc_j \ar[shiftarr={yshift=-20pt},swap]{rr}{\psi_j(c_j)} \urar{\varphi_j(c_j)} \rar[swap]{\varphi_j(c_j)} & (F \cdot \beta_j^\op) c_j \rar{\psi_{i}(c_j) } & G (\gamma_j c_j),
    \end{tikzcd}
  \end{equation*}
  where $\psi_{i}$ is second coordinate of the component of the given cocone $(\gamma_k,\psi_k)$ ($k \in I$) at $(\beta_j, \id) \colon (B_j, F \cdot \beta_j^\op) \rightarrow (B, F)$ from \eqref{eq:con-mor}.
  The outer path commutes:
  the right square commutes by definition \eqref{eq:con-mor} of $\psi$, the left triangle commutes since the $\varphi_k$ form a cocone, and the lower part commutes, since the $\psi_k$ form a cocone.

  The uniqueness of the morphism $(\gamma, \psi)$ is obvious from its definition in \eqref{eq:con-mor}: every natural transformation $\psi' \colon F \rightarrow G$ with $\psi' \cdot \varphi_i = \psi_i$ satisfies \eqref{eq:con-mor}, so $\psi' = \psi$.
\end{proof}

From \Cref{thm:codensity} we obtain:
  \begin{thmC}[\cite{l13}]\label{thm:up}
    The codensity monad of the inclusion $\Fam_\f(\Set^\op) \incl \Fam(\Set^\op)$ is isomorphic to the ultraproduct monad $\mathcal{U} \mathcal{P}$.
  \end{thmC}
  \begin{rem}
    The above reasoning extends to Boolean sheaves over any category $\C$ with limits, i.e.\ functors $B^\op \rightarrow \C$ that sends pullbacks $b = c \lor d$ of diagrams $c \ge c \land d \le d$ in $B^\op$ to pullbacks in $\C$.
    Then, generalizing \Cref{lem:bsh-dense}, the inclusion $\BSh_\f(\C) \incl \BSh(\C)$ of finite Boolean sheaves into Boolean sheaves over $\C$ is dense, and \Cref{thm:up} generalizes to ultraproducts on categories of families $\Fam(\C)$ when $\C$ has limits and filtered colimits.

    In an orthogonal direction, the argument can be generalized to different dualities at the index level, by altering the dual adjunction between $\Set$ and $\BA$. For example, the dual adjunction between $\Set$ and $\MSL$ used for the codensity presentation of the filter monad (\Cref{sec:filter}) yields the \emph{reduced product}.
  \end{rem}

\section{Vietoris Monads}
\label{sec:vietoris-monads}

As another important class of codensity monads, we study monads on categories of topological spaces that associate to a given space a suitably topologized space of interesting subsets.

\subsection{The Vietoris Monad on Stone Spaces}
\label{sec:vietoris-monad-stone}

We start with the Vietoris monad on the category $\Stone$ of Stone spaces and continuous
maps. The \emph{Vietoris space} \(\Viet X\) of a Stone space \(X\) is given by the set of
all closed subsets of~\(X\) equipped with the \emph{hit-or-miss topology}, which has the
following subbasic open sets:
\[
  \Diamond U = \{\text{$C \subseteq X$ closed} \mid C \cap U \ne \emptyset \},
  \
  \Box U = \{\text{$C \subseteq X$ closed} \mid C \cap U  = \emptyset \},
  \
  \text{for $U\subseteq X$ clopen}.
\]
We can represent $\V X$ as a double dual space as follows. Let $\JSL$ be the category of
join-semilattices with bottom. Morphisms $f$ from $J$ to the join-semilattice $2=\{0,1\}$ (with
join given by maximum) correspond to ideals $f^{-1}(0)$: non-empty subsets of $J$ that are downwards
closed and closed under join. The set $\JSL(J,2)$ of ideals is topologized as a subspace of
the product space $2^{|J|}$, where $2$ carries the discrete topology. We then have the
isomorphism
$
  \V X \cong \JSL(\Stone(X,2),2)
$
which identifies a closed set $C\subseteq X$ with the ideal of all clopen sets
$U\subseteq X$ with $C\cap U=\emptyset$ (i.e.\ $C\in \Box U$). The object mapping
$X\mapsto \V X$ thus naturally extends to a monad,
for which the appropriate codensity setting is given by:
\begin{equation}\label{eq:codensity-setting-vietoris}
\begin{tikzcd}[column sep=40, row sep=25]
    \Kl_\f(\Pfin)
    \ar{r}{U_\Pfin} &
    \Stone
   \ar[out=20,in=-20,loop,looseness=3, "\Viet",pos=0.47]
    \ar[yshift=0pt, bend right=20,]{d}[swap]{\Stone(-,2)}
    \ar[yshift=0pt, bend left=20, leftarrow, swap]{d}[swap]{\JSL(-,2)}
    \ar[phantom]{d}{\dashv}
    \\
    \Kl_\f(\Pfin)^\op \ar[hook]{r}{I_\Pfin^\op}
    \ar{u}{(-)^\op}[swap]{\rotatebox{-90}{$\simeq$}}
    &
    \JSL^\op
  \end{tikzcd}
\end{equation}
Here $U_\Pfin\colon \Kl_f(\Pfin) \to \Set$ from \eqref{eq:codensity-setting-filter-kleisli} is
lifted to $\Stone$, identifying $U_\Pfin X=\Pfin X$ with the discrete space $2^X$. The functor $I_\Pfin$ is
the dense inclusion (\iref{ex:dense-alg}{kleisli-finitary}), using
$\EM(\Pfin) \cong \JSL$. Commutativity of \eqref{eq:codensity-setting-vietoris} is shown similarly as for \eqref{eq:codensity-setting-filter-kleisli}. Thus, from \Cref{thm:codensity} we obtain: 

\begin{thmC}[\cite{gpr20}]
  The Vietoris monad $\V$ on $\Stone$ is the codensity monad of the functor $U_\Pfin\colon \Kl_\f(\Pfin)\to \Stone$.
\end{thmC}

\subsection{The Lower Vietoris Monad on Topological Spaces}
\label{sec:vietoris-monad-top}
By varying the ingredients of the above setting a bit, codensity presentations for other Vietoris-type monads emerge easily. Here we consider the \emph{lower
  Vietoris monad}, a.k.a.~\emph{Hoare hyperspace monad}~\cite{fpr21}, on the category $\Top$ of topological spaces.
Given a topological space~$X$, the \emph{lower Vietoris space} is the topological space~$\V_\downarrow X$ of all closed subsets of $X$ with the topology generated by the subbasic open sets $\Diamond U$ (as defined above) for $U \subseteq X$ \emph{open}.
Similar to $\V X$, we can represent $\V_\downarrow X$ as a double dual space. Let $\CJsl$ be the category of complete join-semilattices and join-preserving maps (i.e.~algebras for the power set monad~$\P$).
Morphisms in $\CJsl$ from~$J$ to $2=\{0,1\}$ correspond to complete ideals:
subsets of $J$ that are downwards closed and closed under arbitrary joins. The
set $\CJsl(J,2)$ of all complete ideals is topologized as a subspace of the product space
$\Sierp^{|J|}$ of the Sierpinski space $\Sierp$. We then have
$
  \V_\downarrow X \cong \CJsl(\Top(X,\Sierp),2),
$
where the isomorphism identifies a closed set $C\subseteq X$ with the complete ideal of all open sets $U\subseteq X$
with $C\cap U=\emptyset$. Thus, the object mapping $X\mapsto \V_\downarrow X$ extends to a monad,
whose codensity setting is given by the square below: 
\begin{equation*}
\begin{tikzcd}[column sep=40, row sep=25]
    \Kl(\P)
    \ar{r}{U_\P} &
    \Top
   \ar[out=20,in=-20,loop,looseness=3, "\Viet_\downarrow",pos=0.49]
    \ar[yshift=0pt, bend right=20,]{d}[swap]{\Top(-,\Sierp)}
    \ar[yshift=0pt, bend left=20, leftarrow, swap]{d}[swap]{\JSL(-,2)}
    \ar[phantom]{d}{\dashv}
    \\
   \Kl(\P)^\op \ar[hook]{r}{I_\P^\op}
    \ar{u}{(-)^\op}[swap]{\rotatebox{-90}{$\simeq$}}
    &
    \CJsl^\op
  \end{tikzcd}
\end{equation*}
 Note that this setup is very similar
to the filter monad on \Top from \Cref{sec:filter}, except that we now work
with complete join-semilattices in lieu of (finitary) meet-semilattices and with the full
Kleisli category for the monad~$\P$ in lieu of its finite restriction.
The canonical functor $U_\P\colon \Kl(\P) \to \Set$ is taken as a functor to $\Top$ by
identifying $U_\P X=\P X$ with the topological space $\Sierp^X$, and $I_\P$ is the dense
inclusion (\iref{ex:dense-alg}{kleisli}), using that $\EM(\P) \cong \CJsl$. The outside
commutes analogously to~\eqref{eq:codensity-setting-vietoris}. From \Cref{thm:codensity} we obtain the following new result:

\begin{thm}
  The lower Vietoris monad $\V_\downarrow$ on $\Top$ is the codensity monad of $U_\P$.
\end{thm}
\section{Topological Reflections}
We have seen in \Cref{sec:ultrafilter} that many codensity monads may be viewed as enhancements of spaces:
the ultrafilter monad on $\Set$ may be seen as a reflector from sets, viewed as discrete spaces, into compact Hausdorff spaces, and the ultrafilter monad on $\Top$ may be seen as a reflector into Stone spaces.
In this section we consider some further examples of this phenomenon: sobrification and compactification.

\subsection{The Sobrification Monad on Topological Spaces}
\label{sec:sobr-monad-top}
\emph{Sober spaces}~\cite{i72,j82} (i.e.\ topological spaces where every non-empty irreducible closed set is the closure of a single point) play an important role in point-free topology and are captured by the
\emph{sobrification monad}. Recall that a \emph{frame} is a
poset~$X$ which has all finite meets and all joins and which satisfies the infinite
distributive law $x\wedge \bigvee_{i\in I} {x_i} = \bigvee_{i\in I} x\wedge x_i$ for all
$x, x_i \in X$, $i \in I$. We denote by $\Frm$ the category of all frames and maps preserving finite meets and all joins.
It is isomorphic to the category  \(\EM(\Ffrm)\) of algebras for the \emph{free frame monad} \(\Ffrm\) on \Set sending a set~\(X\) to the set~$\Ffrm X$ of upwards closed subsets of \(\P_\f X\)~\cite[C1.1, Lemma~1.1.3]{j02}.

The \emph{sobrification} monad on $\Top$ is given by $\Sob X = \Frm(\Top(X,\Sierp), 2)$; more precisely, $\Sob$~is the monad induced by the dual adjunction in \eqref{eq:codensity-setting-sobrification}. Sipo\c{s}~\cite{s18} has shown that algebras for~$\Sob$ correspond to sober spaces and
has given a codensity presentation of this monad, which in our duality-based framework corresponds to the codensity setting given by:
\begin{equation}\label{eq:codensity-setting-sobrification}
\begin{tikzcd}[column sep=25, row sep=25]
  \Sier
    \ar[hook]{r}{J} &
    \Top
   \ar[out=20,in=-20,loop,looseness=3, "\Sob",pos=0.49]
    \ar[yshift=0pt, bend right=20,]{d}[swap]{\Top(-,\Sierp)}
    \ar[yshift=0pt, bend left=20, leftarrow, swap]{d}[swap]{\Frm(-,2)}
    \ar[phantom]{d}{\dashv}
    \\
    \Kl(\Ffrm)^\op \ar[hook]{r}{{I}_{\Ffrm}^\op}
    \ar{u}{E}[swap]{\rotatebox{-90}{$\simeq$}}
    &
    \Frm^\op
  \end{tikzcd}
\end{equation}
Here $J$ is the inclusion of the full subcategory $\Sier$ of $\Top$ consisting of powers
$\Sierp^X$, where~$X$ is any set, and $I_{\Ffrm}$ is the
canonical dense inclusion functor from \iref{ex:dense-alg}{kleisli}.
The equivalence functor \(E\) and commutation of the diagram are given by the following lemma.

\begin{lem}\label{lem:sierp-dual}
  The object mapping $X\mapsto \Sierp^X \cong \Frm(\Ffrm X, 2)$ defines a dual equivalence $E$ between the Kleisli category
  $\Kl(\Ffrm)$ and the category $\Sier$.
\end{lem}
\begin{proof}
  We prove that the object mapping $X\mapsto \Sierp^X$ defines the dual equivalence $E$.
  By definition of the product topology, for every set $B$, the frame of open sets of $\Sierp^B$ is isomorphic to the free frame generated by $B$:
  $\Top(\Sierp^B,\Sierp) \cong \Ffrm B$.
  Using the universal property of powers we obtain
  \[
    \Top(\Sierp^B, \Sierp^A)
    \cong
    \Set(A, \Top(\Sierp^B, \Sierp))
    \cong
    \Set(A, \Ffrm B)
    \cong \Kl_{\Ffrm}(A, B). \qedhere
  \]
\end{proof}

From \Cref{thm:codensity} we obtain the following result due to Sipo\c{s}~\cite{s18}:
\begin{thmC}[\cite{s18}]
  The sobrification monad \(\Sob\) on $\Top$ is the codensity monad of the embedding $\Sier \subto \Top$.
\end{thmC}

\subsection{The Stone-\v{C}ech Compactification}
\label{sec:stone-cech}
Next, we derive a novel codensity presentation of the general version of the Stone-\v{C}ech compactification for topological spaces, which is induced by a dual adjunction with $\mathsf{C}^*$-algebras.

We recall some fundamentals from operator theory \cite{murphy90}:
A \emph{(unital) $\mathsf{C}^*$-algebra} $A$ is a complex unital Banach algebra with an anti-linear involution $(-)^* \colon A \rightarrow A$ satisfying the $*$-identity.
Spelled out, this means that $A$ is a normed complete complex vector space together with an associative unital bilinear multiplication $\cdot \colon A \times A \rightarrow A$ such that $\| a \cdot b \| \le \| a \| \cdot \| b \|$, and for which the involution satisfies for all $a, b \in A$ and $\lambda, \mu \in \Cpx$ the equations
\[ 1^* = 1, \quad (a^*)^* = a, \quad (a \cdot b)^* = b^* \cdot a^*, \quad (\lambda a + \mu b)^* = \bar \lambda a^* + \bar \mu b^*, \quad \| a^* \cdot a \| = \| a \|^2,\]
where $\bar{(-)}$ denotes complex conjugation.
For example, the complex numbers are a $\mathsf{C}^*$-algebra with complex conjugation as involution, and for every topological space $X$ the set $\Cnt_{\mathsf{b}}(X) = \Top_{\mathsf{b}}(X, \Cpx)$ of bounded continuous maps into the complex numbers form a $\mathsf{C}^*$-algebra with the sup-norm, and operations defined pointwise.

A \emph{$^*$-homomorphism} between $\mathsf{C}^*$-algebras $A, B$ is a map $f \colon A \rightarrow B$ that is a $^*$-preserving $\Cpx$-algebra homomorphism, viz.\ $f$ is linear and
\[f(a \cdot b) = f(a)\cdot f(b),\quad f(1) = 1,\quad f(a^*) = f(a)^*,\]
which yields a category $\Cstar$ of $\mathsf{C}^*$-algebras and $^*$-homomorphisms. 
One can show that every $^*$-homomorphism is norm-decreasing, that is, $\| f(a) \| \le \| a \|$ for all $a \in A$.

A fundamental result of operator theory is \emph{Gelfand duality}, which states that the category $\Cstar$ is dually equivalent to the category $\CHaus$ of compact Hausdorff spaces and continuous maps. The equivalence sends a $\mathsf{C}^*$-algebra $A$ to its \emph{spectrum} $\Spec A = \Cstar(A, \Cpx)$ with the equipped with the weak $^*$-topology, and conversely it sends a compact Hausdorff space $X$ to the $\mathsf{C}^*$-algebra $\Cnt(X) = \Top(X, \Cpx)$ of continuous functions on $X$. 

The category $\Cstar$ is a curious example of a concrete category of algebraic flavour that is monadic over $\Set$~\cite{d69,n71}. The right adjoint $\Cstar \rightarrow \Set$ is not only forgetting structure, but sends
a $\mathsf{C}^*$-algebra $A$ to its unit ball
\(BA = \{a \in A \mid \| a \| \le 1\}\) and a $^*$-homomorphism to its restriction on unit balls.
Its left adjoint is given by  the functor \[F \colon \Set \rightarrow \Cstar, \qquad X \mapsto \mathsf{C}(D^X),\]
where $D = B \Cpx$ is the complex unit disc, so $\Cstar \simeq \EM(\mathcal{C})$~\cite{n71} for the monad $\mathcal{C} = BF$.
The category $\EM(\mathcal{C})$ is known to be $\omega_1$-presentable \cite{i82,p93} (however we do not need a concrete presentation by operations and equations)
This means by (the countable version of) \iref{ex:dense-alg}{kleisli-finitary} that the restriction $\Kl_{\omega}(\mathcal{C})$ of the Kleisli category of $\mathcal{C}$ to countable sets is dense in $\Cstar$. 
Under Gelfand duality $\Cstar \simeq^\op \CHaus$, which restricts the adjunction $\Cnt_{\mathsf{b}} \dashv \Spec$, the subcategory $\Kl_\omega(\mathcal{C})$ is dual to the full subcategory $D_\omega \incl \CHaus \incl \Top$ of countable powers $D^i, i \le \omega$ of the unit disc (see \cite[Section 2]{n71}).

Put together, in the square~\eqref{eq:codensity-setting-sc} below the includion $I_{\mathcal{C}}$ is dense and the outer square commutes, so it is a codensity setting for the monad $\beta = \Spec \Cnt_{\mathsf{b}}$ on topological spaces.

\begin{equation}
  \label{eq:codensity-setting-sc}
\begin{tikzcd}[column sep=25, row sep=25]
  D_\omega
    \ar[hook]{r}{I} &
    \Top
   \ar[out=20,in=-20,loop,looseness=3, pos=0.49, "\beta"]
    \ar[yshift=0pt, bend right=20,]{d}[swap]{\Cnt_{\mathsf{b}}}
    \ar[yshift=0pt, bend left=20, leftarrow, swap]{d}[swap]{\Spec}
    \ar[phantom]{d}{\dashv}
    \\
    \Kl_{\omega}(\mathcal{C})^\op \ar[hook]{r}{{I}_{\mathcal{C}}^\op}
    \ar{u}{E}[swap]{\rotatebox{-90}{$\simeq$}}
    &
    \Cstar^\op
  \end{tikzcd}
\end{equation}

The space $\beta X$ is a way to describe the \emph{Stone-\v{C}ech compactification} of the space $X$~\cite{gj60}.
This description is often accompanied by the assumption that $X$ is Tychonoff so that the unit is a dense embedding. 

From \Cref{thm:codensity} we obtain the following new result:
\begin{thm}
  The Stone-\v{C}ech compactification $\beta$ is the codensity monad of the inclusion $I \colon D_{\omega} \incl \Top$.
\end{thm}

\section{Probability Monads}
\label{sec:probability-monads}
Finally, we investigate monads of interest in probabilistic and quantum
computation, and give a uniform treatment of their codensity presentations based on restrictions of the Kleisli category of the (countable) distribution monad as suggested in the work of Shirazi~\cite{s24}.
As noted by van Belle~\cite{vb22}, integral representation theorems play a key role in such presentations, and we show that in order to apply our framework, they are in fact the \emph{only} non-trivial part.





\subsection{The Expectation Monad}
\label{sec:expect-monad-sets}

We start with the \emph{expectation monad}~\cite{z10,jmf16}, a probability monad that has received a lot of attention in recent years, e.g.~in quantum foundations.
A \emph{finitely additive probability measure} on a set \(X\) is a finitely additive probability measure on the discrete measurable space \((X, \Sigma_{X} = 2^X)\), that is, a map \(p \colon 2^X \rightarrow [0, 1]\) with $p(X)=1$ and \(p(A + B) = p(A) + p(B)\) for disjoint \(A, B \in 2^X\).
If \(X\) is finite, then every finitely additive probability measure on \(X\) is \emph{discrete}:
recall that a \emph{discrete finite probability measure} on \(X\) is a map \(p \colon X \rightarrow [0, 1]\) with \(\sum_{x \in X}p(x) = 1\) and \(p(x) =0\) for all but finitely many \(x \in X\).
 We denote the set of finitely additive probability measures on a set \(X\) by \(\Exp X\), and the set of discrete finite probability measures by \(\Dist_\f X\).
For the codensity presentation of the expectation monad we embed these notions into an algebraic setting given by \emph{effect algebras} and \emph{effect modules}~\cite{jm12,g99}.

An \emph{effect algebra} is a partial commutative monoid \((A, \oplus, 0)\) together with an operator \((-)^{\bot} \colon A \rightarrow A\) satisfying that \(x^{\bot}\) is the unique element in \(A\) with \(x \oplus x^{\bot} = 1\), where \(1 = 0^{\bot}\), and \(x \oplus 1\) is defined iff \(x = 0\).
Morphisms of effect algebras preserve all this structure, forming a category \(\EA\).
The unit interval $[0,1]$ forms an effect algebra with \(r \oplus s = r + s\) defined if \(r + s \le 1\), and \(r^\bot = 1 - r\).
Boolean algebras form a full subcategory of effect algebras, with \(a \oplus b = a \lor b\) defined whenever \(a \land b = 0\), and \(a^{\bot} = \neg a\).
The following non-trivial density result, which is not covered by \Cref{ex:dense-alg}, is due to Staton and Uijlen~\cite[Cor.~10]{su18}:
  \begin{prop}\label{thm:ba-ea-dense}
    The inclusion \(J \colon \BA_\f \incl \EA\) is dense.
  \end{prop}
  Effect algebras carry a monoidal structure \( \otimes \) such that \(A \otimes B\) represents bilinear morphisms~\cite{jm12}.
  In more detail, given effect algebras $A$, $B$ and $C$ a map $f\colon A\times B \to C$ is \emph{bilinear} if $f(a,-)\colon B \to C$ and $f(-,b)\colon A \to C$ are morphisms of effect algebras for every $a \in A$ and $b\in B$.
  The tensor product of $A$ and $B$ is an effect algebra $A\times B$ equipped with a canonical bilinear map $\eta\colon A\times B \to A \otimes B$ such that for every bilinear map $f\colon A\times B \to C$, there exists a unique effect algebra morphism $\hat f\colon A\otimes B \to C$ such that $\hat f\cdot \eta = f$. 
  Note that the tensor unit is \(2 = \{0, 1\}\).
  For every $a \in A$ and $b \in B$, we will write $a\otimes b$ for $\eta(a,b)$. 
  The unit interval \([0, 1]\) with multiplication forms a monoid \([0, 1] \otimes [0, 1]  \xra{m} [0, 1]\) for this monoidal structure.

Effect algebras are a partial version of commutative monoids, so there exists a corresponding notion of module:
An \emph{effect module} is an effect algebra \(E\) that is a \([0, 1]\)-module for the above tensor product, that is,
\(E\) is equipped with an effect algebra morphism $a\colon [0,1] \otimes E \to E$ such that the following square commutes:
\[
  \begin{tikzcd}[column sep = 40]
    [0,1] \otimes  [0,1] \otimes E
    \ar{r}{m \otimes E}
    \ar{d}[swap]{a}
    &[]
    [0,1]\otimes E
    \ar{d}{a}
    \\[]
    [0,1]\otimes E
    \ar{r}{a}
    &
    E
  \end{tikzcd}
\]
For $r \in [0,1]$ and $x \in E$ we write $rx$ for $a(r\otimes x)$.

Effect module morphisms are effect algebra morphisms preserving the action, forming a category \(\EMod\).
Its forgetful functor to \(\EA\) has a left adjoint. 
This adjunction yields a monad~\(\Act\) on \(\EA\) given by \(\Act A= [0,1] \otimes A\), 
whose algebras are effect modules (this is analogous to $R$-modules, for a ring $R$, being algebras for the monad $R \otimes (-)$ on abelian groups).
On finite Boolean algebras we have an isomorphism
\begin{align*}
  \Act(2^X) = [0, 1] \otimes 2^X &\cong [0, 1]^X \\
  r \otimes A &\mapsto r \chi_A, \\
  \sum_{x \in X} d(x) \otimes \delta_x &\mapsfrom d,
\end{align*}
where $\chi_A$ is the characteristic map of $A\subseteq X$ given by $\chi_A(x)=1$ if $x\in A$, and $\chi_A(x)=0$ if $x\not\in A$, and $\delta_x$ is given by $\delta_x(x)=1$ and $\delta_x(y)=0$ for $y\neq x$. 

We extend the density result from \Cref{thm:ba-ea-dense} to effect modules, where we denote the full subcategory of \(\Kl(\Act)\) given by finite Boolean algebras by \(\Kl_{\BA_\f}(\Act)\):

\begin{prop}\label{prop:baM-emod-dense}
  The inclusion $I_{\Act} \colon \Kl_{\BA_\f}(\Act) \incl \EMod$, $2^{X} \mapsto [0, 1]^{X}$, is dense.
\end{prop}
\begin{proof}
  We show that the functor
  \begin{align*}
   \y_{I_{\Act}} \colon \EMod &\rightarrow [\Kl_{\BA_\f}(\Act)^{\op}, \Set] \\
    X & \mapsto \EMod(I_{\Act}(-), X)
  \end{align*}
  is fully faithful.

  \begin{enumerate}
  \item The square
  \[
    \begin{tikzcd}[column sep = 40]
      \EMod \dar{\y_{I_{\Act}}} \rar{|-|} & \EA \dar{\y_{J}} \\
      {[\Kl_{\BA_\f}(\Act)^{\op}, \Set]} \rar{[F_{\Act}^{\op}, \Set]} & {[\BA_{\f}^{\op}, \Set]}
    \end{tikzcd}
  \]
  commutes, where \(F_{\Act} \colon \BA_{\f} \incl \Kl_{\BA_\f}(\Act)\) denotes the canonical
  left adjoint to the Kleisli category restricted to finite Boolean algebras.
  The composite \(\y_{J} \cdot |-|\) is faithful, since \(\y_{J}\) is faithful by \Cref{thm:ba-ea-dense}. This implies \(\y_{I_{\Act}}\) is faithful.
  
  \item To show that \(\y_{I_{\Act}}\) is full, let \(\varphi \colon \y_{I_{\Act}}X \rightarrow \y_{I_{\Act}}Y\) be a natural transformation.
    Since \(\y_{J}\) is fully faithful due to \Cref{thm:ba-ea-dense}, we obtain an effect \emph{algebra} morphism \(f \colon |X| \rightarrow |Y|\) such that $y_J f = [F_{\Act}^{\op}, \Set]\varphi \colon \y_{J}X \rightarrow \y_{J}Y$. 
    To show that \(f\) is an effect module morphism it just remains to prove it preserves the action. 
    To this end, we recall from the proof in~\cite{su18} the concrete definition of \(f\):
    For a finite set $A$, an \(\EA\)-morphism \(2^{A} \rightarrow X\) (called a \emph{test} in op.~cit.) corresponds to a family \((x_{a})_{a \in A}\) with \(\sum_{a}x_{a} = 1\), and so we denote it as the tuple
    \[
      (x_{a})_{a} \in \y_{J}X(2^{A}) = \EA(2^A, X).
    \]
    Hence, elements \(x \in X\) correspond to \(\EA\)-morphisms \((x, x^{\bot}) \colon 2^{2} \rightarrow X\), and moreover, we have \(f(x) = y\) whenever \(\varphi_{2^2}(x, x^{\bot}) = (y, y^{\bot})\). Thus, we have natural isomorphisms
    \[
      X \cong \EA(2^2,X) \cong \EMod([0,1] \otimes 2^2, X) = y_{I_\Act}X.
    \]

  Now, given \(r\in [0,1]\) and \(x \in X\), consider the Kleisli morphism \(k\colon 2^{2} \rightarrow [0, 1]
  \otimes 2^{2}\) defined by \(x \mapsto r \otimes x\). Using the above natural isomorphism,
  it is easy to show that 
  \[
    (rx, (rx)^{\bot}) = \y_{I_{\Act}}X(k)(x, x^{\bot}) \in \y_{I_{\Act}}X(2^{2}),
  \]
  By the naturality of $\varphi$ we have
  \[
    \varphi_{2^2}(rx, (rx)^{\bot}) = \varphi_{2^2}(\y_{I_{\Act}}(k)(x, x^{\bot})) = \y_{I_{\Act}}(k)(\varphi_{2^2}(x, x^{\bot})) = \y_{I_{\Act}}(k)(y, y^{\bot}) = (ry, (ry)^{\bot}).
  \]
  Thus, \(f(rx) = ry = rf(x)\) as desired.
  \qedhere
\end{enumerate}
\end{proof}

Finitely additive probability measures on a set $X$ are, by definition, precisely effect algebra morphisms from $2^X$ to $[0,1]$, that is, $\Exp X = \EA(2^X, [0, 1])$.
They also correspond to effect module morphisms from $[0, 1]^{X}$ to $[0, 1]$, which is the content of the following `discrete integral representation theorem':

\begin{thmC}[{\cite[Prop.~33]{jmf16}}]\label{prop:exp-rep}
  For every set $X$, there is a natural bijection
  \[\EMod([0, 1]^{X}, [0, 1]) \cong \EA(2^{X}, [0, 1]) = \Exp X.\]
\end{thmC}
The bijection sends $f\in\EMod([0, 1]^{X}, [0, 1])$ to the measure $p$ with $p(A)=f(\chi_A)$.
We thus define the \emph{expectation monad} $\Exp$ on $\Set$ by \(\Exp X = \EA(2^{X}, [0, 1])\); by \Cref{prop:exp-rep}, $\Exp$  is given by the adjunction in the codensity setting below:
\begin{equation}\label{eq:codensity-setting-exp}
\begin{tikzcd}[column sep=25, row sep=25]
    \Kl_\f({\Dist_\f})
    \ar{r}{U_{\Dist_\f}}&
    \Set
    \ar[out=20,in=-20,loop,looseness=3, "\Exp",pos=0.49]
    \ar[yshift=0pt, bend right=20,]{d}[swap]{\Set(-,[0, 1])}
    \ar[yshift=0pt, bend left=20, leftarrow, swap]{d}[swap]{\EMod(-, [0, 1])}
    \ar[phantom]{d}{\dashv}
    \\
    \Kl_{\BA_\f}(\Act)^\op \ar[hook]{r}{I_{\Act}^\op}
    \ar{u}{E}[swap]{\rotatebox{-90}{$\simeq$}}
    &
    \EMod^\op
  \end{tikzcd}
\end{equation}
The inclusion \(I_\Act\) is dense by \Cref{prop:baM-emod-dense}.
The duality $E$ of the left extends the duality $\BA_\f^\op\simeq \Set_\f$.
On objects, it does the same, and for morphisms we have for $A,B\in \Set_\f$ the following natural bijection:
\begin{align*}
  \Kl_{\f}(\Dist_\f)(A, B) &= \Set(A, \Dist_\f B) \cong \EA(2^B, [0, 1]^A) \cong \EA(2^B,
  \mathcal{M}(2^A)) \\
  &\cong \Kl_{\BA_\f}(\Act)(EB, EA).
\end{align*}
To see that the outer square commutes, note that on finite sets $A$ we have
\[
  \EMod([0, 1]^A, [0, 1]) \cong \EA(2^A,[0,1]) \cong \Dist_\f A,
\]
where the first step holds by \Cref{prop:exp-rep} and the second one since finitely
additive probability measures on the finite set $A$ are discrete.

Thus, from \Cref{thm:codensity} we obtain a novel codensity presentation of $\Exp$:
\begin{thm}\label{thm:codensity-exp}
The expectation monad $\Exp$ is the codensity monad of $U_{\Dist_\f}\colon \Kl_\f(\Dist_\f)\to \Set$.
\end{thm}
\begin{rem}\label{rem:exp-functor}
  The expectation \emph{functor} \(\Exp X \cong \EA(2^{X}, [0, 1])\) has a much simpler presentation:
  One can show, by (co)limit manipulation very similar to our proof of \Cref{thm:codensity}, that~$\Exp$ is the right Kan extension of \(\Dist_\f \colon \Set_\f \rightarrow \Set\) along \(I\colon \Set_\f \subto \Set\).
  This only requires density of \(\BA_\f \incl \EA\)  (\Cref{thm:ba-ea-dense}), not the representation result for $\Exp$ (\Cref{prop:exp-rep}).
\end{rem}

\subsection{The Giry Monad and its Variants}
\label{sec:giry-monad-meas}

For classical probability theory one uses countably (instead of finitely) additive probability measures.
A \emph{probability measure} on a measurable space \((X, \Sigma_{X})\) is a map
\(p \colon \Sigma_{X} \rightarrow [0, 1]\) such that \(p(X) = 1\) and \(p(\sum_{i}A_{i}) = \sum_{i}p(A_{i})\) for every countable family \((A_{i})_{i}\) of pairwise disjoint sets.
We denote the set of probability measures on \(X\) by \(\Giry X\).
The assignment \(X \mapsto \Giry X\) forms a monad on the category \(\Meas\) of measurable spaces and measurable maps called the \emph{Giry monad}~\cite{giry1982}.
Here, \(\Giry X\) is equipped with the \(\Sigma\)-subalgebra of the power \([0, 1]^{\Sigma_{X}}\) of \([0, 1]\) in \(\Meas\), where $[0,1]$ is equipped with the usual Borel $\sigma$-algebra.
On a countable discrete space \((X, \Sigma_{X} = 2^{X})\) a probability measure corresponds to a map \(p \colon X \rightarrow [0, 1]\) with \(\sum_{x}p(x) = 1\).
We denote the set of all such discrete probability measures by \(\Dist X\).
If we equip \(\Dist X\) with the \(\Sigma\)-subalgebra of \([0, 1]^{X}\), then \(\Dist\) is a functor \(\Setc \rightarrow \Meas\), where \(\Setc \incl \Set\) is the full subcategory of countable sets.

Again, these constructions can be expressed in the language of effect algebras, though we have to use a countably infinitary operation.
A \emph{\(\sigma\)-effect algebra}~\cite{bgm19} (or $\sigma$-complete effect algebra~\cite{at06}) is an effect algebra \(A\) with a partial commutative associative operation \(\oplus_{n < \omega} \colon A^{\omega} \rightharpoonup A\) that behaves as expected with unit \(0\) and addition \(\oplus\) of \(A\) (here commutativity means that the operation is invariant under permutation of the arguments).
Morphisms of \(\sigma\)-effect algebras preserve all this structure, forming a category \(\sEA\).

The category $\sEA$ contains the category of \(\sigma\)-algebras and the category of countably complete Boolean algebras as subcategories, analogously to the finite case:
for a countably complete Boolean algebra $B$ the operation $a_0 \oplus a_1 \oplus a_2 \oplus \cdots$ is defined for a sequence $(a_i)_{i < \omega}$ iff $a_i \land a_j = 0$ for all $i \ne j$ and the result is the union $\bigvee_{i < \omega} a_i$.
We obtain a dense subcategory:
Restricting the familiar dual equivalence between \(\Set\) and the category \(\Caba\) of complete atomic Boolean algebras to countable sets \(\Setc\), we obtain a dual equivalence $\Setc \simeq^\op \Cabac$ between the category $\Setc$ of countable sets and the category $\Cabac$ of Boolean algebras of the form \(2^{A}\) for a countable set \(A\).
\begin{lem}\label{lem:cba-cea-dense-app}
  The inclusion \(I \colon \Cabac \incl \sEA\) is dense.
\end{lem}
\noindent
The proof for  is the same as for effect algebras~\cite[Thm.~9]{su18}, just replacing finite sums with countable ones.
  We include the adjusted proof for the convenience of the reader.
\begin{proof}
  We show that the functor
  \[
    T \colon \sEA \rightarrow [\Cabac^{\op}, \Set], \qquad X \mapsto \sEA(I-, X)
  \]
  is fully faithful.
  Under the equivalence $2^{-} \colon \Setc \simeq \Cabac^\op$, this is equivalent to showing that the functor
\[
    T \colon \sEA \rightarrow [\Setc, \Set], \qquad X \mapsto \sEA(2^{-}, X)
  \]
  is fully faithful.
We define an \emph{$A$-test} on a countable set \(A\) to be a morphism \(f \in \sEA(2^{A}, X)\): it corresponds to an \(A\)-indexed family of elements \((x_a)_{a \in A}\) with \(\sum_{a}x_{a} = 1\), so we also denote $f$ by $(x_a)_{a \in A}$.

  \begin{enumerate}
    \item  To show that \(T\) is faithful, let $X$ and $Y$ be \(\sigma\)-effect algebras, and let \(f \ne g \colon X \rightarrow Y\) be \(\sigma\)-effect algebra homomorphisms.
          Then \(f(x) \ne g(x)\) for some \(x \in X\).
          Then \(Tf \ne Tg\), since at \(2 \in \Setc\) the \(2\)-test \((x, x^{\bot})\) is in \(\sEA(2^{2}, X)\) and \[(Tf)_{2}(x, x^{\bot}) = (f(x), f(x^{\bot})) \ne (g(x), g(x^{\bot})) = (Tg)_{2}(x, x^{\bot}).\]

    \item To show that \(T\) is full let \(\varphi \colon TX \rightarrow TY\).
          We define the map
          \begin{equation}
            \label{eq:f}
            f \colon X \rightarrow Y, \qquad f(x) = y \quad \text{ where } \quad  (y', y) = \varphi_{2}(x^{\bot}, x),
          \end{equation}
          is the \(\varphi\)-image of the \(2\)-test \((x^{\bot}, x)\).
          We have to show that \(f\) preserves countable sums.
          If \(\sum_{a \in A}x_{a} = x\) is a countable sum in \(X\), then \((x^{\bot}, (x_{a})_{a \in A})\) is an \(1 + A\)-test.
          We denote its image under \(\varphi_{1+A}\) by
          \[
          \varphi_{1+A}(x^{\bot}, (x_{a})_{a \in A}) = (y', (y_{a})_{a \in A}) \in \sEA(2^{1+A}, Y).
          \]
          Now for any fixed \(a \in A\), define \(k_a \colon 1 + A \rightarrow 2 = \{0, 1\}\) by \(a \mapsto 1\) and \(b \mapsto 0\) for \(b \ne a\). Then we have
          \begin{align*}
            \varphi_{2}(x_{a}^{\bot}, x_{a})
            &=
              \varphi_{2}(TX(k_a)(x^{\bot}, (x_{a})_{a \in A}))
            \\
            &=
              TY(k_a)(\varphi_{1 + A}(x^{\bot}, (x_{a})_{a \in A}))
            \\
            &=
              TY(k_a)(y', (y_{a})_{a \in A})
            \\
            &=
              (y' + \sum_{b \ne a} x_{b}, y_{a}),
          \end{align*}
          so for all $a$ we have $f(x_a) = y_a$ by definition \eqref{eq:f} of $f$.
          Now consider the map \(j = [0!, 1!] \colon 1 + A \rightarrow 2\) that is the characteristing map of $A$.
          We have 
          \begin{align*}
            &\varphi_{2}(x^{\bot}, \sum_{a \in A}x_{a}) \\
            &= \varphi_{2}(T X (j)(x^{\bot}, (x_{a})_{a \in A})) \\
            &= TY (j)(\varphi_{1+A}(x^{\bot}, (x_{a})_{a \in A})) \\
            &= TY(j)(y', (y_{a})_{a \in A}) \\
            &= (y', \sum_{a}y_{a \in A}), 
          \end{align*}
          and together with the first calculation above we conclude \(f(\sum_{a}x_{a}) = \sum_{a}y_{a} = \sum_{a}f(x_{a})\).
          The proof that \(f\) preserves \(0\) and \(1\) is similar:
          consider $j \colon 1 \rightarrow 2$ with $j(\ast) = 1$, then
          \[\varphi_2(0, 1) = \varphi_2(TX(l)(1)) = TY(l)(\varphi_1(1)) = TY(l)(1) = (0, 1)\]
          since \(\varphi\) must send the unique $1$-test $(1)$ on $X$ to $(1)$ on $Y$.
          By definition of $f$ we have $f(1) = 1$; the proof for $f(0) = 0$ is analogous.
          \qedhere
  \end{enumerate}
\end{proof}

A \emph{\(\sigma\)-effect module} $X$ is an effect module on a \(\sigma\)-effect algebra $X$ whose action satisfies an infinite distributive law for the infinitary sum operation:
for all $r \in [0, 1]$ and countable sequence $(x_a)_{a \in A}$ we have that in
\[r \sum_{a \in A}x_a = \sum_{a \in A} rx_a\]
the left side is defined iff the right side is, and in this case they coincide.
This yields a category \(\sEMod\), and a monad \(\Act\) on  \(\EA_{\sigma}\) sending \(X\) to the free \(\sigma\)-effect module \([0, 1] \otimes X\) analogous to the finitary case.
The full subcategory of the Kleisli category on complete atomic Boolean algebras is denoted \(\Kl_{\Cabac}(\Act)\), and we have a density result analogous to \Cref{prop:baM-emod-dense}, which is proved in exactly the same way:
\begin{prop}\label{cor:cabacM-semod-dense}
  The inclusion \(I_{\Act} \colon \Kl_{\Cabac}(\Act) \incl \sEMod\) is dense.
\end{prop}

A probability measure on a measurable space \(X\) is a \(\sigma\)-effect algebra morphism \(\Sigma_{X} \rightarrow [0, 1]\), so \(\Giry X \cong \sEA(\Sigma_{X}, [0, 1])\).
For the codensity presentation of the Giry monad we use the following integral representation theorem:
\begin{notheorembrackets}
  \begin{thm}[{\cite[Cor.~3.8]{vb22}}]\label{thm:giry-rep}
    For every measurable space \(X\) we have
    \[\Giry X \cong \sEMod(\Meas(X, [0, 1]), [0, 1]) .\]
  \end{thm}
\end{notheorembrackets}
We therefore obtain the following codensity setting for the Giry monad:
\begin{equation}\label{eq:codensity-setting-giry}
\begin{tikzcd}[column sep=25, row sep=25]
    \Kl_\cnt(\Dist)
    \ar{r}{U_{\Dist}}&
    \Meas
   \ar[out=20,in=-20,loop,looseness=3, "\Giry",pos=0.49]
    \ar[yshift=0pt, bend right=20,]{d}[swap]{\Meas(-,[0, 1])}
    \ar[yshift=0pt, bend left=20, leftarrow, swap]{d}[swap]{\sEMod(-,[0, 1])}
    \ar[phantom]{d}{\dashv}
    \\
    \Kl_{\Cabac}(\Act)^\op \ar[hook]{r}{I_{\Act}^{\op}}
    \ar{u}{E}[swap]{\rotatebox{-90}{$\simeq$}}
    &
    \sEMod^\op
  \end{tikzcd}
\end{equation}
Here \(\sEMod(E, [0, 1])\) is topologized as a \(\Sigma\)-subalgebra of the power \([0, 1]^{E}\).
The functor~\(I_\Act\) is dense by \Cref{cor:cabacM-semod-dense}, and the equivalence on the left and commutativity of the outer square are analogous to the finite case for the expectation monad.
Thus, from \Cref{thm:codensity} we obtain the following codensity presentation, similar to \Cref{thm:codensity-exp}:
  \begin{thm}\label{prop:giry-cod}
    The Giry monad $\Giry$ is the codensity monad  of \(U_\Dist \colon \Kl_\cnt(\Dist ) \rightarrow \Meas\).
  \end{thm}
From this result, we can easily derive an even simpler codensity presentation of $\Giry$ by the functor \(\Dist \colon \Setc \rightarrow \Meas\) due to van Belle~\cite{vb22}. Indeed, consider the codensity setting on the left below, whose adjunction is the composition of the two adjunctions on the right:

\noindent
\begin{minipage}{.3\textwidth}
\begin{equation*}\label{eq:codensity-setting-giry-simple}
\begin{tikzcd}[column sep=25, row sep=25]
    \Setc
    \ar{r}{\Dist}&
    \Meas
    \ar[yshift=0pt, bend right=20,]{d}[swap]{\Meas(-,[0, 1])}
    \ar[yshift=0pt, bend left=20, leftarrow, swap]{d}[swap]{\sEA(-,[0, 1])}
    \ar[phantom]{d}{\dashv}
    \\
    \Cabac^{\op} \ar[hook]{r}{I_{\Act}^{\op}}
    \ar{u}{\rotatebox{-90}{$\simeq$}}
    &
    \sEA^\op,
  \end{tikzcd}
\end{equation*}
\end{minipage}
\begin{minipage}{.1\textwidth}
~
\end{minipage}
\begin{minipage}{.4\textwidth}
  \[
    \begin{tikzcd}[column sep=50]
      \Meas \ar[bend right=10]{r}[swap,pos=.6]{\Meas(-,[0, 1])} \ar[phantom]{r}{\top}
      &
      {\sEMod^\op} \ar[bend right=10]{l}[swap,pos=.4]{\sEMod(-,[0, 1])} \ar[bend right=10, swap]{r}{|-|^{\op}} \ar[phantom]{r}{\top}
      &
      \sEA^{\op} \ar[bend right=10]{l}[swap]{\Act^{\op}}
\end{tikzcd}.
\]
\end{minipage}

\medskip\noindent This yields a codensity presentation of the monad \(\sEA(\Meas(X, [0, 1]), [0, 1])\) on \Meas by the functor $\Dist$.
This monad is isomorphic to the Giry monad,
since every \(\sigma\)-effect algebra homomorphism \(\Meas(X, [0, 1]) \rightarrow [0, 1]\) is already a morphism of $\sigma$-effect modules~\cite[Lemma A.3]{vb22}:}
\[\Giry X \cong \sEMod(\Meas(X, [0, 1]), [0, 1]) \cong \sEA(\Meas(X, [0, 1]), [0, 1]) \cong \Cody{\Dist}.\]
Hence, we recover the presentation due to van Belle~\cite{vb22}:
\begin{notheorembrackets}
  \begin{thm}[\cite{vb22}]\label{thm:giry-rep-simple}
    The Giry monad $\Giry$ is the codensity monad  of \(\Dist \colon \Setc \to \Meas\).
  \end{thm}
\end{notheorembrackets}
By replacing $\Meas$ with $\Set$ (regarded as discrete measurable spaces) in the above codensity settings, we also obtain a similar presentation of the \emph{countable expectation monad}~$\Exp_\sigma$:
\begin{notheorembrackets}
  \begin{thm}[\cite{s24}]\label{thm:expsigma-rep}
    The monad $\Exp_\sigma$ is the codensity monad  of \(\Dist \colon \Setc \to \Set\).
  \end{thm}
\end{notheorembrackets}

\begin{rem}\label{prop:giry-ran-app}
  As noted by van Belle~\cite{vb22}, the Giry \emph{functor} can also be presented as the right Kan extension of \(\Dist \colon \Setc \rightarrow \Meas\) along the inclusion \(I \colon \Setc \subto \Meas\) equipping a countable set with the discrete \(\Sigma\)-algebra.
  This also follows easily via a computation analogous to~\eqref{thm:codensity}, and the proof does not need an integral representation theorem.
\end{rem}

  \takeout{
\subsection{The Giry Monad}
\label{sec:giry-monad-meas}

To capture general probability theory, we move from finitely additive to countably additive probability measures.
A \emph{probability measure} on a measurable space \((X, \Sigma_{X})\) is a map
\(p \colon \Sigma_{X} \rightarrow [0, 1]\) such that \(p(X) = 1\) and \(p(\sum_{i}A_{i}) = \sum_{i}p(A_{i})\) for every countable family \((A_{i})_{i}\) of pairwise disjoint sets.
We denote the set of probability measures on \(X\) by \(\Giry X\).
The assignment \(X \mapsto \Giry X\) forms a monad on the category \(\Meas\) of measurable spaces and measurable maps, called the \emph{Giry monad}~\cite{giry1982}.
Here, \(\Giry X\) is equipped with the \(\Sigma\)-subalgebra of the power \([0, 1]^{\Sigma_{X}}\) of \([0, 1]\) in \(\Meas\), where $[0,1]$ carries the usual Borel $\sigma$-algebra.
On a countable discrete space \((X, \Sigma_{X} = 2^{X})\), a probability measure corresponds to a map \(p \colon X \rightarrow [0, 1]\) with \(\sum_{x}p(x) = 1\).
We denote the set of all such discrete probability measures by \(\Dist X\).
If we equip \(\Dist X\) with the \(\Sigma\)-subalgebra of \([0, 1]^{X}\), then \(\Dist\colon \Setc \rightarrow \Meas\) forms a functor, where \(\Setc \incl \Set\) denotes the full subcategory of countable sets.

Again, these constructions can be expressed in the language of effect algebras, though we have to use a countably infinitary sum operation.
A \emph{\(\sigma\)-effect algebra}~\cite{at06} is an effect algebra \(A\) with a partial commutative associative operation \(\oplus_{n < \omega} \colon A^{\omega} \rightharpoonup A\) that behaves as expected with unit \(0\) and addition \(\oplus\) of \(A\) (here commutativity means that the operation is independent of permutation of the arguments).
Morphisms of \(\sigma\)-effect algebras preserve all this structure, forming a category \(\sEA\).
It contains the category of \(\sigma\)-algebras and the category of countably complete Boolean algebras as subcategories.
The familiar dual equivalence between \(\Set\) and the category \(\Caba\) of complete atomic Boolean algebras~\cite[Sec.~VI.4.6]{j82} restricts to countable sets, yielding a dual equivalence between $\Setc$ and the category $\Cabac$ of countably complete Boolean algebras of the form \(2^{A}\) for a countable set \(A\).

The codensity presentation of $\Giry$ again rests on a representation theorem~\cite[Cor.~3.8]{vb22}:
  \begin{thm}\label{thm:giry-rep}
    For every $X\in \Meas$ we have
    $\Giry X \cong \sEA(\Meas(X, [0, 1]), [0, 1])$.
  \end{thm}
The isomorphism is analogous to that of \Cref{prop:exp-rep}. This means that the Giry monad is given by the adjunction in the diagram below, which also establishes a codensity setting:
\begin{equation*}
\begin{tikzcd}[column sep=25, row sep=25]
    \Setc
    \ar{r}{\Dist}&
    \Meas
    \ar[out=20,in=-20,loop,looseness=3, "\Giry",pos=0.49]
    \ar[yshift=0pt, bend right=20,]{d}[swap]{\Meas(-,[0, 1])}
    \ar[yshift=0pt, bend left=20, leftarrow, swap]{d}[swap]{\sEA(-,[0, 1])}
    \ar[phantom]{d}{\dashv}
    \\
    \mathbf{CABA}_\mathsf{c}^{\op} \ar[hook]{r}{I^{\op}}
    \ar{u}{\rotatebox{90}{$\simeq$}}
    &
    \sEA^\op,
  \end{tikzcd}
\end{equation*}
Here $I$ is the inclusion, and we have:

\begin{lem}\label{lem:cba-cea-dense}
  The inclusion \(I\colon \Cabac \incl \sEA\) is dense.
\end{lem}
From \Cref{thm:codensity} we obtain the codensity presentation of $\Giry$ due to van Belle~\cite{vb22}.

  \begin{thmC}[\cite{vb22}]\label{thm:giry-rep-simple}
    The Giry monad $\Giry$ on measurable spaces is the codensity monad of \(\Dist \colon \Setc \to \Meas\).
  \end{thmC}
Let us mention that in much the same way, \Cref{thm:codensity} yields codensity presentations of
\begin{enumerate}
\item the Giry monad by the forgetful functor
$U_\Dist\colon \Kl_\mathsf{c}(\Dist) \to \Meas$, analogous to \Cref{thm:codensity-exp};
\item the \emph{countable expectation monad} $\Exp_\sigma X = \sEMod([0,1]^X,[0,1])$ on $\Set$ by the functor $\Dist\colon \Set_\mathsf{c}\to \Set$, analogous to \Cref{thm:giry-rep-simple}. 
\end{enumerate}
Here $\sEMod$ is the category of \emph{$\sigma$-effect modules}~\cite{at06}, the infinitary version of effect modules. Both results rest on the fact that every $\sEA$-morphism of type $[0,1]^X\to [0,1]$ is an $\sEMod$-morphism, i.e.~linear~\cite[Lem.~A.3]{vb22}. 

\begin{rem}\label{prop:giry-ran}
  As already noted in~\cite[Rem.~4.3]{vb22}, the Giry \emph{functor} $\Giry$ can also be presented as the right Kan extension of \(\Dist \colon \Setc \rightarrow \Meas\) along the inclusion \(I \colon \Setc \subto \Meas\) equipping a countable set with the discrete \(\Sigma\)-algebra.
  Similar to the case of $\Exp$
(\Cref{rem:exp-functor}), this follows via an easy computation as in the proof of~\Cref{thm:codensity}, and does not need an integral representation theorem.
\end{rem}

} 

\subsection{The Radon Monad and the Kantorovich Monad}
\label{sec:radon-monad-compact}
Variants of the Giry monad on subcategories of measurable spaces are treated similarly in our framework. Here we consider the \emph{Radon monad}~\cite{fj15}, which captures an important subclass of probability measures. A \emph{Radon} probability measure on a compact Hausdorff space \(X\) is a probability measure \(\mu\) on \(\Bo_X\) (the Borel $\sigma$-algebra generated by the open sets of $X$) satisfying the condition
  $\mu(A) = \sup \{\,\mu(K) \mid K \subseteq A \text{ and } K \text{ compact}\,\}$ for all \(A \in \Bo_X\).
  The set of all Radon probability measures on \(X\) is denoted by \(\Rad X\); it forms a compact Hausdorff space when topologized as a subspace of the product space \([0, 1]^{\Bo_X}\) in \CHaus, the category of all compact Hausdorff topological spaces and continuous maps between them.

  Similar to \Cref{thm:giry-rep}, Radon measures come with a representation theorem, which yields a particularly simple codensity presentation.
  It is a variant of the Riesz-Markov representation theorem; van Belle~\cite{vb22} proved it via the Daniell-Stone representation:
\begin{notheorembrackets}
\begin{thm}\label{thm:rad-pres}
  For every $X\in \CHaus$, we have 
  $\Rad X \cong \EA(\CHaus(X, [0, 1]), [0, 1])$.\end{thm}
\end{notheorembrackets}
We therefore obtain the following codensity setting for the Radon monad:
\begin{equation*}
\begin{tikzcd}[column sep=50, row sep=25]
    \Set_\f
    \ar{r}{\Dist}&
    \CHaus
   \ar[out=14,in=-14,loop,looseness=4, "\Rad",pos=0.5]
    \ar[yshift=0pt, bend right=20,]{d}[swap]{\CHaus(-,[0, 1])}
    \ar[yshift=0pt, bend left=20, leftarrow, swap]{d}[swap]{\EA(-,[0, 1])}
    \ar[phantom]{d}{\dashv}
    \\
    \BA_\f^\op \ar[hook]{r}{J^\op}
    \ar{u}{\rotatebox{90}{$\simeq$}}
    &
    \EA^\op
  \end{tikzcd}
\end{equation*}
Here $\Dist X$ is topologized as a subspace of the hypercube $[0,1]^X$, and $J$ is the inclusion functor, which is dense by~\Cref{thm:ba-ea-dense}.
Commutativity of the outer square is a special case of commutativity of the setting \eqref{eq:codensity-setting-exp}.
Thus, from \Cref{thm:codensity} we obtain the following codensity presentation of the Radon monad:

\begin{notheorembrackets}
  \begin{thmC}[{\cite{vb22}}]
    The Radon monad on compact Hausdorff spaces  is the codensity monad of \(\Dist \colon \Set_\f \rightarrow \CHaus\).
  \end{thmC}
\end{notheorembrackets}
\begin{rem}
  The Radon monad restricts to the \emph{Kantorovich}~\cite{vb05} or \emph{bounded Lipschitz monad}~\cite{vb22}  on the category \(\CMet\) of compact metric spaces and nonexpansive maps by equipping it with the \emph{Kantorovich metric}.
  Similarly, the above functors \(\Dist\) and \(\EA(-, [0, 1])\) corestrict to \(\CMet\) by equipping them with the Kantorovich and supremum metric, respectively.
  Replacing in the above codensity setting \(\CHaus\) by \(\CMet\) and the functors by their metric counterparts yields the codensity presentation of the Kantorovich monad by $\Dist$~\cite{vb22}.
\end{rem}

\subsection{Finitely Additive Measures in a Finite Semiring}
\label{sec:finit-addit-meas}

As our last example in this section, we depart from classical probability measures and generalize to measures valued in a finite semiring~\cite{r20}. Fix a finite semiring \(S\), viewed as a discrete topological space.
  An \emph{\(S\)-valued measure} on a Stone space \(X\) is a map \(p \colon \Cl X \rightarrow S\), where $\Cl X$ denotes the set of all clopen subsets of $X$, satisfying \(p(\emptyset) = 0\) and \(p(A + B) = p(A) + p(B)\) for disjoint clopen subsets $A, B \subseteq X$.
  Let \(\SMeas X \) denote the set of such measures.

  The set \(\SMeas X \subseteq S^{\Cl X}\) is a closed subspace~\cite[Lemma~3.4]{r20}, whence a Stone space. The integral representation theorem in this case is particularly simple:
Given $X\in \Stone$ we topologize \(\Mod_S(\Stone(X, S), S)\) as a subspace of \(S^{\Stone(X, S)}\),
where the set \(\Stone(X, S)\) is an \(S\)-module under pointwise operations.
Then we obtain
\begin{lem}\label{lem:smeas-rep}
  For every Stone space \(X\) we have
  \(\SMeas X \cong \Mod_S(\Stone(X, S), S)\).
\end{lem}
\begin{proof}
  The homeomorphism extends an \(S\)-valued measure \(p \in \SMeas(X)\) to an \(S\)-module morphism \(\hat{p} \colon \Stone(X, S) \rightarrow S\) in the obvious way:
  For clopens $K \subseteq X$ we have continuous characteristic maps $\chi_{K} \colon X \rightarrow 2$,
  and clearly want $\hat{p}(\chi_K) = p(K)$.
  As we can express every map \(k \in \Stone(X, S)\) as a linear combination
  \begin{equation}
    \label{eq:k}
    k = \sum_{s \in S}s \cdot \chi_{k^{-1}(s)}.
  \end{equation}
  over disjoint clopens $k^{-1}(s)$
  we are thus forced to define 
  \begin{equation}
    \label{eq:hatp}
    \hat{p}(k) = \sum_{s \in S} s \cdot p(k^{-1}(s)),
  \end{equation}
  which makes $\hat{p}$ $S$-linear by definition.

  Conversely, restricting an \(S\)-module morphism \(h \colon \Stone(X, S) \rightarrow S\) to characteristic maps of clopens under $\iota \colon \Stone(X, \{0, 1\}) \incl \Stone(X, S)$ yields an $S$-valued measure $h \iota$.
  Both constructions are clearly continuous. For $p$ and $h$ as above we have we have $\hat{p}\iota = p$ by definition and $\widehat{h \iota}$ is obvious from \eqref{eq:k} and \eqref{eq:hatp}. 
\end{proof}

This induces a monad \(\SMeas\) with suitable codensity setting given by the diagram below:
\begin{equation*}
\begin{tikzcd}[column sep=50, row sep=25]
    \Kl_\f(\mathcal{S})
    \ar{r}{U_\mathcal{S}}&
    \Stone
   \ar[out=20,in=-20,loop,looseness=3, "\SMeas",pos=0.48]
    \ar[yshift=0pt, bend right=20,]{d}[swap]{\Stone(-,S)}
    \ar[yshift=0pt, bend left=20, leftarrow, swap]{d}[swap]{\Mod_S(-,S)}
    \ar[phantom]{d}{\dashv}
    \\
    \Kl_\f(\mathcal{S})^\op \ar[hook]{r}{I_\mathcal{S}^\op}
    \ar{u}{(-)^\op}[swap]{\rotatebox{-90}{$\simeq$}}
    &
    \Mod_S^\op
  \end{tikzcd}
\end{equation*}
The
   inclusion \(I_{\mathcal{S}}\) is dense by \iref{ex:dense-alg}{kleisli-finitary}, using that $\EM(\mathcal{\mathcal{S}})\cong \Mod_S$, and the outer square commutes similar to \eqref{eq:codensity-setting-filter-kleisli}.
  Thus, from \Cref{thm:codensity} we obtain the following result:

    \begin{thmC}[{\cite{r20}}]\label{thm:codensity-meas-s}
      The monad $\SMeas$ on \Stone is the codensity monad of the functor $U_\mathcal{S}\colon \Kl_\f(\mathcal{S})\to \Stone$.
    \end{thmC}
    By \Cref{rem:monad-ext} we can also see it as a pushforward of the restriction of $\mathcal{S}$ to $\Set_\f$:
    \begin{cor}
      The monad $\SMeas$ is the pushforward of the monad $\mathcal{S}$ on $\Set_\f$ along the inclusion $\Set_\f \incl \Stone$.
    \end{cor}

Note that we recover the codensity presentation of the Vietoris monad on Stone spaces from \Cref{sec:vietoris-monad-stone} by setting $S = \mathbf{2}$.

\section{Application: Börger's Characterization of the Ultrafilter Monad}
\label{sec:appl}
Up to this point, we gave numerous derivations of codensity presentations using our duality-based framework. In particular, each of those presentations
yields a universal property of the respective monad (\Cref{prop:monad-ext-univ-prop}).
In this section we give an illustration of how to use this universal property: we prove a generalization of Börger's characterization~\cite{b87} of the ultrafilter functor/monad as the final finite-coproduct-preserving functor/monad on $\Set$.

We start with a general criterion for preservation of finite coproducts by codensity monads.
Recall that a category $\C$ with finite coproducts is \emph{extensive} if for any two objects $A$ and $B$ the sum functor 
\[
  (\C \downarrow A) \times (\C \downarrow B) \xra{+} \C \downarrow (A + B).
\]
is an equivalence of categories; see~\cite{clw93} for several alternative characterizations.

\begin{prop}\label{prop:cod-pres-+}
  Let $\A$ be a category with finite coproducts and let $\C$ be extensive.
  If $F \colon \A \rightarrow \C$ preserves finite coproducts and
  its pointwise codensity monad exists, then $\Cody{F}$ preserves finite coproducts.
\end{prop}

Before we prove the proposition, we require some auxiliary results for computing limits:
Recall that a functor $F \colon J \rightarrow I$ is \emph{initial} if precomposition with $F$ does not change limits:
for every diagram $D \colon I \rightarrow \C$, if $\lim DF$ exists, then so does $\lim F$ and they are isomorphic.
This is equivalent to every coslice $F \downarrow i$  ($i \in I$) being non-empty and connected~\cite[Ch.~IX.3]{m98}.
\begin{lem}\label{lem:plus-initial}
  Let $\A$ and $\C$ have finite coproducts and let $F \colon \A \rightarrow \C$ be a functor preserving coproducts.
  Then for all $X_1$ and $X_2$, the following functor is initial:
  \begin{align*}
    + \colon (X_1 \downarrow F) \times (X_2 \downarrow F) \rightarrow (X_1 + X_2) \downarrow F \\
    (f_1, f_2) \mapsto f_1 + f_2 \colon X_1 + X_2 \rightarrow FA_1 + FA_2 \cong F(A_1 + A_2).
  \end{align*}
\end{lem}
\begin{proof}
  For every $c=[c_1, c_2]\colon X_1 + X_2 \to FA$ in $(X_1 + X_2 \downarrow F)$, the coslice $+ \downarrow c$ is non-empty and connected since the codiagonal $\nabla \colon c_1 + c_2 \rightarrow c$ is a terminal object in $+ \downarrow c$.
\end{proof}

The following fact is well-known (see e.g.~\cite{k21} for a proof).
\begin{lem}\label{lem:ext-connected-lim}
  In every extensive category, connected limits commute with finite coproducts.
\end{lem}
\begin{proof}[{Proof of \Cref{prop:cod-pres-+}}]
  Let $F \colon \A \rightarrow \C$ preserve finite coproducts.
    We have the following chain of isomorphisms:
    \begin{align}
      \label{thm:boerger:pres-+:step:1}
      \Cody F (X_1 + X_2)
      &\cong \lim_{X_1 + X_2 \rightarrow FA}FA \\
      \label{thm:boerger:pres-+:step:2}
      &\cong \lim_{\substack{X_1  \rightarrow FA_1 \\ X_2 \rightarrow FA_2}}FA_1 + FA_2 \\
      \label{thm:boerger:pres-+:step:3}
      &\cong (\lim_{X_1 \epi FA_1}FA_1) + (\lim_{X_2 \epi FA_2} FA_2) \\
      \label{thm:boerger:pres-+:step:4}
      &\cong \Cody{X_1} + \Cody{X_2}.
    \end{align}
    The isomorphism \eqref{thm:boerger:pres-+:step:1} is just the limit formula \eqref{eq:ran} for right Kan extensions.
    Step \eqref{thm:boerger:pres-+:step:2} is \Cref{lem:plus-initial}.
    Step \eqref{thm:boerger:pres-+:step:3} is \Cref{lem:ext-connected-lim}.
    Step \eqref{thm:boerger:pres-+:step:4} is just the limit formula for $\Ran$ again.

    Finally, $\Cody F$ preserves the initial object:
    \[\Cody F( 0) = \int_{A \in \C} \C(0,  FA) \power FA \cong \int_{A \in \C} FA = \lim F \cong 0, \]
    since $0$ is an initial object of the diagram scheme \A.
  \end{proof}

We now turn to the ultrafilter functor and give a new, purely categorical, proof of Börger's characterization of the ultrafilter functor using its codensity presentation.
For a recent exposition on this result with orthogonal generalizations see Garner's work~\cite{g20}.
Recall that the ultrafilter monad on $\Set$ is the codensity monad of the inclusion $I \colon \Set_\f \incl \Setf$.
Since $\Set$ is extensive and~$I$ preserves finite coproducts, \Cref{lem:plus-initial} yields:
\begin{cor}
  The functor $\U$ preserves finite coproducts.
\end{cor}
\begin{rem}\label{rem:ba-rep-pres-prod}
  For an alternative argument one can use the presentation $\mathcal{U} \cong \BA(2^-, 2)$:
  the exponential $2^- \colon \Set \rightarrow \BA^\op$ is a left adjoint and thus preserves coproducts, and the representable $\BA(-, 2) \colon \BA^\op \rightarrow \Set$ preserves finite coproducts by Stone duality:
  \[
    \BA(A \times B, 2)
    \cong
    \Stone(1, \hat{A} + \hat{B})
    \cong
    \Stone(1, \hat{A}) + \Stone(1, \hat B)
    \cong
    \BA(A, 2)+ \BA(B, 2),
  \]
  where $\hat{A} = \BA(A, 2)$ is the Stone spectrum of $A$ (i.e.~the Stone space dual to $A$).
\end{rem}
\begin{exa}
  From \Cref{lem:plus-initial} we get that various other codensity monads considered in~\cite{as21} preserve coproducts:
  all those codensity monads of inclusions $\C_\f \incl \C$ where $\C$ is extensive and $\C_\f$ consists of finitely carried $\C$-objects.
  For example, $\C$ can be the category of graphs, posets or $M$-sets for a finite monoid $M$.
\end{exa}

\begin{rem}\label{rem:2-cat}
  We denote by $\cat{CAT}_+ \incl \cat{CAT}$ the sub-2-category of all categories with finite products, all functors finite-coproduct preserving functors and all natural transformations between them.
  Note that for $F, G \in [\A, \B]_+$ all natural transformations $\varphi \colon (F, +, 0) \rightarrow (G, +, 0)$ are monoidal, that is, the following squares commute: 
  \[
    \begin{tikzcd}[column sep = 40]
      F(A + B) \ar{d}[swap]{\rotatebox{90}{$\cong$}} \rar{\varphi_{A + B}}
      &
      G(A + B) \ar{d}{\rotatebox{-90}{$\cong$}}
      \\
      FA + FB \rar{\varphi_A + \varphi_B}
      & GA + GB 
    \end{tikzcd}
  \]
  To see this precompose the square by $F\inl\colon FA \to F(A+B)$ and $F\inr\colon FB\to F(A+B)$, respectively, where $\inl$ and $\inr$ are the coproduct injections, note that so are  $F\inl, F\inr$ and use the naturality of $\varphi$. 
\end{rem}
\begin{thmC}[{\cite{b87}}]\label{thm:boerger}
  The ultrafilter functor $\U \colon \Set \rightarrow \Set$ is the terminal object in the category of finite-coproduct preserving set functors.
\end{thmC}
\begin{proof}
  To see that for every finite-coproduct preserving set functor $G$, there exists a unique natural transformation $G \rightarrow \U$ consider the following chain of isomorphisms (here $[\Set,\Set]_+$ denotes the category of all finite-coproduct preserving set functors):
  \begin{align}
    \notag
     [\Set, \Set](G, \U) & \cong [\Set, \Set](G, \Cody{I}) \\
    \label{thm:boerger:unique:step:1}
    & \cong [\Set_\f, \Set](GI, I) \\ 
    \label{thm:boerger:unique:step:2}
    &\cong [\Set_\f,\Set]_+(GI, I) \\
    \label{thm:boerger:unique:step:3}
    &\cong [\cat{1}, \Set](GIi, Ii) \\
    \notag
    &\cong \Set(G1, 1)  \\
    \notag
    &\cong  1. 
  \end{align}
  Step \eqref{thm:boerger:unique:step:1} uses the universal property of the right Kan extension.
  Step \eqref{thm:boerger:unique:step:2} uses that $\cat{CAT}$ and $\cat{CAT}_+$ have the same 2-cells (natural transformations), see \Cref{rem:2-cat}.
  Step \eqref{thm:boerger:unique:step:3} is induced by the fact that the inclusion $i \colon 1 \subto \Set_\f$ is the \emph{finite-coproduct completion} of the terminal category $1$:
  it is the unit at $1$ of the corresponding 2-adjunction between the 2-category $\cat{CAT}$ and the 2-category $\cat{CAT}_+$, so precomposition with $i$ induces an isomorphism $[\Set_\f, \Set]_+ \cong [1, \Set]$.
  Since $I$ preserves finite coproducts (being the free extension of the embedding $1! \colon 1 \incl \Set$, viz.\ $Ii = 1!$), we obtain the desired bijection of 2-homsets.
\end{proof}

\begin{cor}
  The ultrafilter monad is the terminal object in the category of finite-coproduct preserving set monads.
\end{cor}
\begin{proof}
  If $G$ is a monad, then the condition that $\alpha \colon G \rightarrow \U \cong \Cody I$ is a monad morphism corresponds to the unique map $G1 \rightarrow 1$ being a $G$-algebra structure, which is true.
\end{proof}

\begin{rem}
  We emphasize that \emph{(re-)presentation matters}:
  We did not necessarily choose the \emph{simplest}, but a \emph{canonical} presentation of the ultrafilter monad.
  For example, one can show that the ultrafilter monad (\Cref{sec:ultrafilter}) is also the codensity monad of the inclusion $\{3\} \incl \Set$ of the full subcategory on the three-element set~\cite{l13};
  however, it seems difficult to derive the universal property of $\U$ from this codensity presentation.
\end{rem}

\section{Conclusion and Future Work}
We have introduced a general, unifying approach to codensity presentations of monads, based on the simple core principle of relating codensity to density via duality. We have shown that numerous known and new codensity presentations for important monads, e.g.~(ultra)filter, Vietoris, and probability monads, emerge as instances of our framework, in many cases almost for free by invoking standard duality and density results from the literature.

A previous approach towards a general understanding of codensity presentations is due to Ad\'amek and Sousa~\cite{as21}.
They introduce a notion of ultrafilter monad on a symmetric monoidal closed category that is locally finitely presentable, and prove that this monad is the codensity monad of the full subcategory of finite presentable objects, provided that the category has a nice cogenerator.
Most of their concrete instances of ultrafilter monads are easily captured by our duality framework (along the lines of \Cref{sec:ultrafilter}); whether their general presentation result is an instance of our \Cref{thm:codensity} remains an open problem.

Finally, on the more applied side, we aim to use our framework to study codensity presentations of additional monads. We are particularly interested in variants of the Vietoris monad related to compact spaces and probabilistic powerdomains~\cite{fpr21}. Here, the choice of the corresponding dual adjunction (and even of the dualizing objects) is far from obvious.

\section*{Acknowledgement} The authors wish to thank Nathanel Arkor for pointing out the connection between our \Cref{thm:codensity} and the recent work by Do\~na Mateo~\cite{m25} (\Cref{rem:dona-mateo}).
We are also grateful to Andrew Krenz for suggesting that results from his PhD thesis~\cite{k25} on double-dualization monads for median algebras (\Cref{sec:meas-medi-algebr}) are captured by our framework.

\bibliographystyle{alphaurl}
\bibliography{bibliography}

@string{springer="Springer"}

@string{elsevier="Elsevier"}

@string{dagstuhl="Schloss Dagstuhl -- Leibniz-Zentrum f{\"u}r Informatik"}

@string{dagstuhl="Schloss Dagstuhl"}

@string{lipics="LIPIcs"}

@book{murphy90,
  title={C*-algebras and operator theory},
  doi = {10.1016/c2009-0-22289-6},
  author={Murphy, Gerald J},
  year={1990},
  publisher={Academic press}
}

@InProceedings{ApplegateT69,
  author = 	 {Applegate,  Harry and Tierney, Miles},
  title = 	 {Categories with models},
  doi = {10.1007/BFb0083086},
  OPTcrossref =  {},
  OPTkey = 	 {},
  booktitle = {Seminar on Triples and Categorical Homology Theory},
  year = 	 {1969},
  editor = 	 {Eckmann, Benno and Tierney, Miles},
  volume = 	 {80},
  OPTnumber = 	 {},
  OPTseries = 	 {Lecture Notes Math.},
  OPTpages = 	 {},
  OPTmonth = 	 {},
  OPTaddress = 	 {},
  OPTorganization = {},
  publisher = {Springer},
  note = 	 {{A} {\itshape Reprints in Theory Appl.~Categories} (18) 2008, 122-185.},
  OPTannote = 	 {}
}

@inproceedings{LenkeEA26,
  author       = {Lenke, Fabian and
                  Wittrock, Nico and
                  Milius, Stefan and
                  Urbat, Henning},
  title        = {Demystifying Codensity Monads via Duality},
  booktitle    = {STACS 2026},
  series       = lipics,
  volume       = {364},
  pages        = {65:1--65:20},
  publisher    = {Schloss Dagstuhl - Leibniz-Zentrum f{\"{u}}r Informatik},
  year         = {2026},
  doi          = {10.4230/LIPICS.STACS.2026.65},
}

@Book{gu71,
  author    = {Peter Gabriel and Friedrich Ulmer},
  title     = {Lokal präsentierbare Kategorien},
  series    = {Lecture Notes in Mathematics},
  volume    = {221},
  publisher = {Springer-Verlag},
  address   = {Berlin, Heidelberg},
  year      = {1971},
  isbn      = {978-3-540-05578-5},
  doi       = {10.1007/BFb0059396},
}

@article{at06,
title = {Classical properties of measure theory on effect algebras},
journal = {Fuzzy Sets and Systems},
volume = {157},
number = {15},
pages = {2139-2143},
year = {2006},
doi = {https://doi.org/10.1016/j.fss.2006.03.010},
author = {A. Aizpuru and M. Tamayo},
}

@article{fritz20,
title = {A synthetic approach to {M}arkov kernels, conditional independence and theorems on sufficient statistics},
journal = {Advances in Mathematics},
volume = {370},
pages = {107239},
year = {2020},
doi = {https://doi.org/10.1016/j.aim.2020.107239},
author = {Tobias Fritz},
}

@article{cj2019, 
title={Disintegration and {B}ayesian inversion via string diagrams}, 
volume={29}, 
DOI={10.1017/S0960129518000488}, 
number={7}, 
journal={Mathematical Structures in Computer Science}, 
author={Cho, Kenta and Jacobs, Bart}, 
year={2019},
pages={938–971}
}

@InProceedings{hinze12,
author="Hinze, Ralf",
OPTeditor="Gibbons, Jeremy
and Nogueira, Pablo",
title="Kan Extensions for Program Optimisation Or: Art and {D}an Explain an Old Trick",
booktitle="11th International Conference on Mathematics of Program Construction (MPC 2012)",
year="2012",
publisher="Springer",
OPTaddress="Berlin, Heidelberg",
pages="324--362",
doi = "10.1007/978-3-642-31113-0_16",
isbn="978-3-642-31113-0"
}

@inproceedings{uacm17,
  author       = {Henning Urbat and
                  Jir{\'{\i}} Ad{\'{a}}mek and
                  Liang{-}Ting Chen and
                  Stefan Milius},
  OPTeditor       = {Kim G. Larsen and
                  Hans L. Bodlaender and
                  Jean{-}Fran{\c{c}}ois Raskin},
  title        = {Eilenberg Theorems for Free},
  booktitle    = {MFCS 2017},
  series       = {LIPIcs},
  volume       = {83},
  pages        = {43:1--43:15},
  publisher    = {Schloss Dagstuhl - Leibniz-Zentrum f{\"{u}}r Informatik},
  year         = {2017},
  OPTurl          = {https://doi.org/10.4230/LIPIcs.MFCS.2017.43},
  doi          = {10.4230/LIPICS.MFCS.2017.43},
  bibsource    = {dblp computer science bibliography, https://dblp.org}
}

@inproceedings{camu16,
  author       = {Liang{-}Ting Chen and
                  Jir{\'{\i}} Ad{\'{a}}mek and
                  Stefan Milius and
                  Henning Urbat},
  OPTeditor       = {Bart Jacobs and
                  Christof L{\"{o}}ding},
  title        = {Profinite Monads, Profinite Equations, and {R}eiterman's Theorem},
  booktitle    = {FOSSACS 2016},
  series       = {Lecture Notes in Computer Science},
  volume       = {9634},
  pages        = {531--547},
  publisher    = {Springer},
  year         = {2016},
  OPTurl          = {https://doi.org/10.1007/978-3-662-49630-5\_31},
  doi          = {10.1007/978-3-662-49630-5\_31},
  bibsource    = {dblp computer science bibliography, https://dblp.org}
}

@article{acmu21,
  author       = {Jir{\'{\i}} Ad{\'{a}}mek and
                  Liang{-}Ting Chen and
                  Stefan Milius and
                  Henning Urbat},
  title        = {Reiterman's Theorem on Finite Algebras for a Monad},
  journal      = {{ACM} Trans. Comput. Log.},
  volume       = {22},
  number       = {4},
  pages        = {23:1--23:48},
  year         = {2021},
  url          = {https://doi.org/10.1145/3464691},
  doi          = {10.1145/3464691},
  bibsource    = {dblp computer science bibliography, https://dblp.org}
}

@article{moggi91,
title = {Notions of computation and monads},
journal = {Information and Computation},
volume = {93},
number = {1},
pages = {55-92},
year = {1991},
OPTnote = {Selections from 1989 IEEE Symposium on Logic in Computer Science},
issn = {0890-5401},
doi = {https://doi.org/10.1016/0890-5401(91)90052-4},
OPTurl = {https://www.sciencedirect.com/science/article/pii/0890540191900524},
author = {Eugenio Moggi}
}

@article{isbell60,
  author       = {Isbell, John R.},
  title        = {Adequate subcategories},
  doi = {10.1215/ijm/1255456274},
  journal      = {Illinois Journal of Mathematics},
  volume       = {4},
  number       = {3},
  pages        = {541--552},
  year         = {1960},
}

@incollection{giry1982,
  author    = {Mich{\`e}le Giry},
  title     = {A categorical approach to probability theory},
  booktitle = {Categorical Aspects of Topology and Analysis},
  series    = {Lecture Notes in Mathematics},
  volume    = {915},
  pages     = {68--85},
  publisher = {Springer},
  year      = {1982},
  doi       = {10.1007/BFb0092872}
}

@article{gpr20,
  author       = {Mai Gehrke and
                  Daniela Petrisan and
                  Luca Reggio},
  title        = {Quantifiers on languages and codensity monads},
  journal      = {Math. Struct. Comput. Sci.},
  volume       = {30},
  number       = {10},
  pages        = {1054--1088},
  year         = {2020},
  doi          = {10.1017/S0960129521000074},
}

@article{day75,
  author       = {Alan Day},
  title        = {Filter Monads, Continuous Lattices and Closure Systems},
  journal      = {Canadian Journal of Mathematics},
  volume       = {27},
  number       = {1},
  pages        = {50--59},
  year         = {1975},
  doi          = {10.4153/CJM-1975-008-8},
}

@book{pontryagin46,
  author    = {Lev Semenovich Pontryagin},
  title     = {Topological Groups},
  publisher = {Princeton University Press},
  year      = {1946},
  address   = {Princeton, NJ},
  OPTnote      = {Translated from the Russian by Emma Lehmer},
}

@book{morris77,
  author       = {Morris, Sidney A.},
  title        = {Pontryagin Duality and the Structure of Locally Compact Abelian Groups},
  series       = {London Mathematical Society Lecture Note Series},
  volume       = {29},
  publisher    = {Cambridge University Press},
  address      = {Cambridge},
  year         = {1977},
  pages        = {viii + 128},
  isbn         = {9780521215435, 0521215439, 9780511600722, 0511600720},
  doi          = {10.1017/CBO9780511600722}
}

@unpublished{kock66,
  author       = {Anders Kock},
  title        = {Continuous Yoneda Representation of a Small Category},
  note         = {Preprint},
  year         = {1966},
  month        = oct,
  institution  = {Aarhus University},
  url          = {https://tildeweb.au.dk/au76680/CYRSC.pdf}
}

@article{isbell66,
  author  = {Isbell, John R.},
  title   = {Structure of categories},
  doi = {10.1090/S0002-9904-1966-11541-0},
  journal = {Bulletin of the American Mathematical Society},
  volume  = {72},
  number  = {4},
  pages   = {619--655},
  year    = {1966},
}

@Book{Manes76,
  author = 	 {Manes, Ernest G.},
  ALTeditor = 	 {},
  title = 	 {Algebraic Theories},
  doi = {10.1007/978-1-4612-9860-1},
  publisher = {Springer},
  year = 	 {1976},
  OPTkey = 	 {},
  volume = 	 {26},
  OPTnumber = 	 {},
  series = 	 {Graduate Texts in Math.},
  OPTaddress = 	 {},
  OPTedition = 	 {},
  OPTmonth = 	 {},
  OPTnote = 	 {},
  OPTannote = 	 {}
}

@article{l13,
title = {Codensity and the ultrafilter monad},
author = {Tom Leinster},
year = {2013},
language = {English},
volume = {28},
pages = {332--370},
journal = {Theory and Applications of Categories},
issn = {1201-X},
publisher = {Mount Allison University},
number = {13},
doi = {10.1016/j.jpaa.2015.08.017}
}

@book{ar94,
place={Cambridge},
series={London Mathematical Society Lecture Note Series},
title={Locally Presentable and Accessible Categories},
dOI={10.1017/CBO9780511600579},
publisher={Cambridge University Press},
author={Ji\v{r}\'\i\ Ad\'{a}mek and Ji\v{r}\'\i\ Rosick\'y},
year={1994},
collection={London Mathematical Society Lecture Note Series}
}

@article{a16,
title = {Codensity and the {G}iry monad},
journal = {Journal of Pure and Applied Algebra},
volume = {220},
number = {3},
pages = {1229-1251},
year = {2016},
issn = {0022-4049},
doi = {https://doi.org/10.1016/j.jpaa.2015.08.017},
OPTurl = {https://www.sciencedirect.com/science/article/pii/S0022404915002248},
author = {Tom Avery}
}

@article{vb22,
  author       = {Van Belle, Ruben},
  title        = {Probability monads as codensity monads},
  journal      = {Theory and Applications of Categories},
  volume       = {38},
  number       = {21},
  pages        = {811--842},
  year         = {2022},
  note         = {arXiv:2111.01250},
}

@article{m25,
  author  = {Do\~{n}a Mateo, Adri\'{a}n},
  title   = {Pushforward monads},
  journal = {Theory and Applications of Categories},
  volume  = {44},
  number  = {30},
  pages   = {934--971},
  year    = {2025},
  url     = {http://www.tac.mta.ca/tac/volumes/44/30/44-30.pdf},
  eprint  = {2406.15256},
  archivePrefix = {arXiv},
  primaryClass  = {math.CT},
}

@article{s18,
OPTurl = {https://doi.org/10.1515/ms-2017-0080},
title = {Codensity and {S}tone spaces},
author = {Andrei Sipo\c{s}},
pages = {57--70},
volume = {68},
number = {1},
journal = {Mathematica Slovaca},
doi = {10.1515/ms-2017-0080},
year = {2018},
}

@misc{s24,
      title={Commutativity and liftings of codensity monads of probability measures},
      author={Zev Shirazi},
      year={2024},
      eprint={2405.12917},
      archivePrefix={arXiv},
      primaryClass={math.CT},
      OPTurl={https://arxiv.org/abs/2405.12917},
}

@article{fpr21,
title={Probability, valuations, hyperspace: Three monads on top and the support as a morphism},
volume={31},
DOI={10.1017/S0960129521000414},
number={8},
journal={Mathematical Structures in Computer Science},
author={Fritz, Tobias and Perrone, Paolo and Rezagholi, Sharwin},
year={2021},
pages={850–897}
}

@phdthesis{z10,
  TITLE = {{Formal certification of game-based cryptographic proofs}},
  AUTHOR = {Zanella-B{\'e}guelin, Santiago},
  URL = {https://pastel.hal.science/pastel-00584350},
  NUMBER = {2010ENMP0050},
  SCHOOL = {{{\'E}cole Nationale Sup{\'e}rieure des Mines de Paris}},
  YEAR = {2010},
  HAL_ID = {pastel-00584350},
  HAL_VERSION = {v1},
}

@article{jmf16,
  author       = {Bart Jacobs and
                  Jorik Mandemaker and
                  Robert Furber},
  title        = {The expectation monad in quantum foundations},
  journal      = {Inf. Comput.},
  volume       = {250},
  pages        = {87--114},
  year         = {2016},
  OPTurl          = {https://doi.org/10.1016/j.ic.2016.02.009},
  doi          = {10.1016/J.IC.2016.02.009},
  bibsource    = {dblp computer science bibliography, https://dblp.org}
}

@article{fj15,
  author       = {Robert W. J. Furber and
                  Bart Jacobs},
  title        = {From {K}leisli Categories to Commutative {C}*-algebras: Probabilistic
                  {G}elfand Duality},
  journal      = {Log. Methods Comput. Sci.},
  volume       = {11},
  number       = {2},
  year         = {2015},
  url          = {https://doi.org/10.2168/LMCS-11(2:5)2015},
  doi          = {10.2168/LMCS-11(2:5)2015},
  bibsource    = {dblp computer science bibliography, https://dblp.org}
}

@article{as21,
title = {D-ultrafilters and their monads},
journal = {Advances in Mathematics},
volume = {377},
pages = {107486},
year = {2021},
issn = {0001-8708},
doi = {https://doi.org/10.1016/j.aim.2020.107486},
OPTurl = {https://www.sciencedirect.com/science/article/pii/S0001870820305144},
author = {Jiří Adámek and Lurdes Sousa}
}

@article{jm12,
 author={Jacobs, Bart and Mandemaker, Jorik},
 title={Coreflections in Algebraic Quantum Logic},
 journal={Foundations of Physics},
 year={2012},
 month={Jul},
 day={01},
 volume={42},
 number={7},
 pages={932-958},
 issn={1572-9516},
 doi={10.1007/s10701-012-9654-8},
}

@article{su18,
  author       = {Sam Staton and
                  Sander Uijlen},
  title        = {Effect algebras, presheaves, non-locality and contextuality},
  journal      = {Inf. Comput.},
  volume       = {261},
  pages        = {336--354},
  year         = {2018},
  OPTurl          = {https://doi.org/10.1016/j.ic.2018.02.012},
  doi          = {10.1016/J.IC.2018.02.012},
  bibsource    = {dblp computer science bibliography, https://dblp.org}
}

@article{r20,
title = {Codensity, profiniteness and algebras of semiring-valued measures},
journal = {Journal of Pure and Applied Algebra},
volume = {224},
number = {1},
pages = {181-205},
year = {2020},
issn = {0022-4049},
doi = {https://doi.org/10.1016/j.jpaa.2019.05.002},
OPTurl = {https://www.sciencedirect.com/science/article/pii/S0022404919301203},
author = {Luca Reggio}
}

@book{m98,
  author = {Mac Lane, Saunders},
  title = {Categories for the Working Mathematician},
  publisher = {Springer},
  year = {1998},
  doi = {10.1007/978-1-4757-4721-8},
}

@book{arv10,
place={Cambridge},
series={Cambridge Tracts in Mathematics},
title={Algebraic Theories: A Categorical Introduction to General Algebra},
publisher={Cambridge University Press},
author={Adámek, J. and Rosický, J. and Vitale, E. M.},
year={2010},
doi = {10.1017/CBO9780511760754},
collection={Cambridge Tracts in Mathematics}}

@article{dl20,
title = {Codensity: Isbell duality, pro-objects, compactness and accessibility},
journal = {Journal of Pure and Applied Algebra},
volume = {224},
number = {10},
pages = {106379},
year = {2020},
issn = {0022-4049},
doi = {https://doi.org/10.1016/j.jpaa.2020.106379},
OPTurl = {https://www.sciencedirect.com/science/article/pii/S0022404920300773},
author = {{Di~Liberti}, Ivan}
}

@article{ksu18,
  author       = {Shin{-}ya Katsumata and
                  Tetsuya Sato and
                  Tarmo Uustalu},
  title        = {Codensity Lifting of Monads and its Dual},
  journal      = {Log. Methods Comput. Sci.},
  volume       = {14},
  number       = {4},
  year         = {2018},
  url          = {https://doi.org/10.23638/LMCS-14(4:6)2018},
  doi          = {10.23638/LMCS-14(4:6)2018},
  bibsource    = {dblp computer science bibliography, https://dblp.org}
}

@article{kg71,
 author = {Kennison, J. F. and Gildenhuys, Dion},
 title = {Equational completion, model induced triples and pro-objects},
 fjournal = {Journal of Pure and Applied Algebra},
 journal = {J. Pure Appl. Algebra},
 issn = {0022-4049},
 volume = {1},
 pages = {317--346},
 year = {1971},
 language = {English},
 doi = {10.1016/0022-4049(71)90001-6},
 zbMATH = {3370550},
 Zbl = {0234.18006}
}

@book{j82,
  author={Johnstone, Peter T.},
  title={{Stone Spaces}},
  isbn={9780521337793},
  series={Cambridge Studies in Advanced Mathematics},
  year={1982},
  publisher={Cambridge University Press}
}

@article{s72,
title = {The formal theory of monads},
journal = {Journal of Pure and Applied Algebra},
volume = {2},
number = {2},
pages = {149-168},
year = {1972},
issn = {0022-4049},
doi = {https://doi.org/10.1016/0022-4049(72)90019-9},
OPTurl = {https://www.sciencedirect.com/science/article/pii/0022404972900199},
author = {Ross Street}
}

@misc{vb05,
author={Franck van Breugel},
title={The Metric Monad for Probabilistic Nondeterminism},
year={2005},
url={http://www.cse.yorku.ca/~franck/research/drafts/monad.pdf},
}

@article{g99,
author={Gudder, Stanley},
title={Convex Structures and Effect Algebras},
journal={International Journal of Theoretical Physics},
year={1999},
month={Dec},
day={01},
volume={38},
number={12},
pages={3179-3187},
issn={1572-9575},
doi={10.1023/A:1026678114856},
url={https://doi.org/10.1023/A:1026678114856}
}

@book{j02,
	author = {Peter T. Johnstone},
	publisher = {Clarendon Press},
	title = {Sketches of an Elephant: A Topos Theory Compendium, Volume 1},
	year = {2002}
}

@article{j18e,
  author       = {Bart Jacobs},
  title        = {From probability monads to commutative effectuses},
  journal      = {J. Log. Algebraic Methods Program.},
  volume       = {94},
  pages        = {200--237},
  year         = {2018},
  OPTurl          = {https://doi.org/10.1016/j.jlamp.2016.11.006},
  doi          = {10.1016/J.JLAMP.2016.11.006},
  bibsource    = {dblp computer science bibliography, https://dblp.org}
}

@article{b37,
author = {Garrett Birkhoff},
title = {{Rings of sets}},
volume = {3},
journal = {Duke Mathematical Journal},
number = {3},
publisher = {Duke University Press},
pages = {443 -- 454},
year = {1937},
doi = {10.1215/S0012-7094-37-00334-X},
URL = {https://doi.org/10.1215/S0012-7094-37-00334-X}
}

@article{i72,
title={Atomless Parts of Spaces.},
volume={31},
DOI={10.7146/math.scand.a-11409},
journal={MATHEMATICA SCANDINAVICA},
author={Isbell,
John R.},
year={1972},
month={Jun.},
pages={5–32} }

@article{e74,
title = {Sheaves of structures and generalized ultraproducts},
journal = {Annals of Mathematical Logic},
volume = {7},
number = {2},
pages = {163-195},
year = {1974},
issn = {0003-4843},
doi = {https://doi.org/10.1016/0003-4843(74)90014-X},
author = {David P. Ellerman}
}

@article{k81,
title = {Triples and compact sheaf representation},
journal = {Journal of Pure and Applied Algebra},
volume = {20},
number = {1},
pages = {13-38},
year = {1981},
issn = {0022-4049},
doi = {https://doi.org/10.1016/0022-4049(81)90046-3},
author = {John F. Kennison}
}

@article{i80,
 title = {Median Algebra},
 author= {John R. Isbell},
 year = {1980},
 journal = {Trans. Ameri. Math. Soc.},
 volume = {260},
 pages = {319-362},
 doi = {10.1090/S0002-9947-1980-0574784-8}
}

@phdthesis{k25,
  TITLE = {{Categorical Median Algebra}},
  AUTHOR = {Andrew Krenz},
  URL = {https://asset.library.wisc.edu/1711.dl/MYJVRL2RC2LIS8A/R/file-825e2.pdf},
  SCHOOL = {{University of Wisconsin-Madison}},
  YEAR = {2025},
}

@article{bk47,
author = {Garrett Birkhoff and S. A. Kiss},
title = {{A ternary operation in distributive lattices}},
doi = {10.1090/S0002-9904-1947-08864-9},
volume = {53},
journal = {Bulletin of the American Mathematical Society},
number = {8},
publisher = {American Mathematical Society},
pages = {749 -- 752},
year = {1947},
}

@article{hk04,
title = {A Coalgebraic Perspective on Monotone Modal Logic},
journal = {Electronic Notes in Theoretical Computer Science},
volume = {106},
pages = {121-143},
year = {2004},
note = {Proceedings of the Workshop on Coalgebraic Methods in Computer Science (CMCS)},
issn = {1571-0661},
doi = {10.1016/j.entcs.2004.02.028},
author = {Helle Hvid Hansen and Clemens Kupke}
}

@misc{gv10,
      title={A view of canonical extension},
      author={Mai Gehrke and Jacob Vosmaer},
      year={2010},
      eprint={1009.2803},
      archivePrefix={arXiv},
      primaryClass={math.LO},
      url={https://arxiv.org/abs/1009.2803},
}

@article{bgm19,
  TITLE = {{Measures on effect algebras}},
  AUTHOR = {Barbieri, Giuseppina and Garc{\'i}a-Pacheco, Francisco J and Moreno-Pulido, Soledad},
  JOURNAL = {{Mathematica Slovaca}},
  PUBLISHER = {{Matematick{\'y} {\'u}stav SAV}},
  VOLUME = {69},
  NUMBER = {1},
  PAGES = {159-170},
  YEAR = {2019},
  MONTH = Feb,
  DOI = {10.1515/ms-2017-0211},
  HAL_ID = {hal-04448041},
  HAL_VERSION = {v1},
}

@article{k70,
author = {Kock, Anders},
journal = {Mathematica Scandinavica},
pages = {151-165},
title = {On Double Dualization Monads.},
url = {http://eudml.org/doc/166154},
volume = {27},
year = {1970},
}

@inproceedings{bsss21,
  author       = {Filippo Bonchi and
                  Alessio Santamaria and
                  Jens Seeber and
                  Pawel Sobocinski},
  editor       = {Fabio Gadducci and
                  Alexandra Silva},
  title        = {On Doctrines and Cartesian Bicategories},
  booktitle    = {9th Conference on Algebra and Coalgebra in Computer Science, {CALCO}
                  2021, Salzburg, Austria, August 31 - September 3, 2021},
  series       = {LIPIcs},
  pages        = {10:1--10:17},
  publisher    = {Schloss Dagstuhl - Leibniz-Zentrum f{\"{u}}r Informatik},
  year         = {2021},
  doi          = {10.4230/LIPICS.CALCO.2021.10},
  bibsource    = {dblp computer science bibliography, https://dblp.org}
}

@article{b87,
title = {Coproducts and ultrafilters},
journal = {Journal of Pure and Applied Algebra},
volume = {46},
number = {1},
pages = {35-47},
year = {1987},
issn = {0022-4049},
doi = {10.1016/0022-4049(87)90041-7},
author = {Reinhard Börger}
}

@article{g20,
	author = {Richard Garner},
	doi = {10.1016/j.apal.2020.102831},
	journal = {Annals of Pure and Applied Logic},
	number = {10},
	pages = {102831},
	title = {Ultrafilters, Finite Coproducts and Locally Connected Classifying Toposes},
	volume = {171},
	year = {2020}
}

@book{k82,
title = {Basic Concepts of Enriched Category Theory},
author = {Kelly, G. Max},
series = {London Mathematical Society Lecture Note Series},
volume = {64},
publisher = {Cambridge University Press},
year = {1982},
url = {mta.ca}
}

@article{clw93,
title = "Introduction to extensive and distributive categories",
author = "Aurelio Carboni and Stephen Lack and Walters, \{R. F C\}",
year = "1993",
month = feb,
day = "3",
doi = "10.1016/0022-4049(93)90035-R",
language = "English",
volume = "84",
pages = "145--158",
journal = "Journal of Pure and Applied Algebra",
issn = "0022-4049",
publisher = "Elsevier",
number = "2",
}

@phdthesis{d16,
  title={Codensity, compactness and ultrafilters},
  author={Devlin, Barry-Patrick},
  year={2016},
  school={University of Edinburgh}
}

@inproceedings{d69,
  title={Variations on Beck's tripleability criterion},
  author={John Williford Duskin},
  doi ={10.1007/BFb0059143},
  series = {Lecture Notes in Mathematics},
  volume = {106},
  year={1969},
}

@article{n71,
  title={Duality in analysis from the point of view of triples},
  author={Joan W Negrepontis},
  doi = {10.1016/0021-8693(71)90105-0},
  journal={Journal of Algebra},
  year={1971},
  volume={19},
  pages={228-253},
}

@article{i82,
  title={Generating the algebraic theory of {$C(X)$}},
  doi = {10.1007/BF02483718},
  author={Isbell, John. R.},
  journal={Algebra Universalis},
  volume={15},
  number={2},
  pages={153--155},
  year={1982},
  publisher={Springer}
}

@article{p93,
  title={On the equational theory of ${C}^*$-algebras},
  doi = {10.1007/BF01196099},
  author={Pelletier, Joan Wick and Rosick{\'y}, Ji{\v{r}}{\'\i}},
  journal={Algebra Universalis},
  volume={30},
  number={1},
  pages={275--284},
  year={1993},
  publisher={Springer}
}

@MISC{k21,
    TITLE = {Does the coproduct preserve limits?},
    AUTHOR = {Alex Kruckman},
    HOWPUBLISHED = {Mathematics Stack Exchange},
    URL = {https://math.stackexchange.com/q/4128892},
    year ={2021}
}

@article{ak43,
author = {Hirotada Anzai and Shizuo Kakutani},
title = {{Bohr compactifications of a locally compact abelian group II}},
volume = {19},
journal = {Proceedings of the Imperial Academy},
number = {9},
publisher = {The Japan Academy},
pages = {533 -- 539},
year = {1943},
doi = {10.3792/pia/1195573310},
URL = {https://doi.org/10.3792/pia/1195573310}
}

@book{gj60,
author={Gillman, Leonard.
and Jerison, Meyer},
title={Rings of continuous functions},
doi = {10.1007/978-1-4615-7819-2},
series={The University series in higher mathematics},
year={1960},
publisher={Van Nostrand},
}

\end{document}